\documentclass[twoside,11pt]{article}

\usepackage{arxiv}

\usepackage[utf8]{inputenc} % allow utf-8 input
\usepackage[T1]{fontenc}    % use 8-bit T1 fonts
\usepackage{url}            % simple URL typesetting
\makeatletter
\g@addto@macro{\UrlBreaks}{\UrlOrds}
\makeatother
\usepackage{booktabs}       % professional-quality tables
\usepackage{amsfonts}       % blackboard math symbols
\usepackage{nicefrac}       % compact symbols for 1/2, etc.
\usepackage{microtype}      % microtypography
\usepackage{lipsum}

\usepackage{cancel}
\usepackage[tbtags]{amsmath}
\usepackage{graphicx}
\usepackage{bm}
\usepackage{amsthm}
\usepackage{makecell}
\usepackage{appendix}
\usepackage{xcolor}
\usepackage{natbib}
\usepackage{multirow}

\usepackage{algorithm}
\usepackage{algpseudocode}
\usepackage{lastpage}
\usepackage{hyperref}       % hyperlinks
\usepackage{tablefootnote}

\newtheorem{definition}{Definition}
\newtheorem{lemma}{Lemma}
\newtheorem{proposition}{Proposition}

\DeclareMathOperator*{\argmax}{argmax}

\newcommand{\gmin}[0]{g_{\min}}
\newcommand{\gmax}[0]{g_{\max}}

\newcommand{\betahvcf}[0]{\bm{\hat{\beta}_\text{\normalfont cf}}}
\newcommand{\betahvgf}[0]{\bm{\hat{\beta}_\text{\normalfont gf}}}
\newcommand{\betahvegf}[0]{\bm{\hat{\beta}_\text{\normalfont egf}}}
\newcommand{\Icd}[0]{\bm{I_\text{\normalfont cd}}}
\newcommand{\Igd}[0]{\bm{I_\text{\normalfont gd}}}
\newcommand{\Iegd}[0]{\bm{I_\text{\normalfont egd}}}
\newcommand{\dbetahv}[0]{\bm{\Delta\hat{\beta}}}
\newcommand{\dbetahvcd}[0]{\bm{\Delta\hat{\beta}_\text{\normalfont cd}}}
\newcommand{\dbetahvgd}[0]{\bm{\Delta\hat{\beta}_\text{\normalfont gd}}}
\newcommand{\dbetahvegd}[0]{\bm{\Delta\hat{\beta}_\text{\normalfont egd}}}
\newcommand{\betahv}[0]{\bm{\hat{\beta}}}
\newcommand{\etav}[0]{\bm{\eta}}

\newcommand{\Iegds}[0]{\bm{I_\text{\normalfont egd,sd}}}
\newcommand{\dbetahvs}[0]{\bm{\Delta\hat{\beta}_\text{\normalfont sd}}}
\newcommand{\dbetahvcds}[0]{\bm{\Delta\hat{\beta}_\text{\normalfont cd,sd}}}
\newcommand{\dbetahvgds}[0]{\bm{\Delta\hat{\beta}_\text{\normalfont gd,sd}}}
\newcommand{\dbetahvegds}[0]{\bm{\Delta\hat{\beta}_\text{\normalfont egd,sd}}}
\newcommand{\dbetahvg}[0]{\bm{\Delta\hat{\beta}_\text{\normalfont gs}}}
\newcommand{\dbetahvcdg}[0]{\bm{\Delta\hat{\beta}_\text{\normalfont cd,gs}}}
\newcommand{\dbetahvgdg}[0]{\bm{\Delta\hat{\beta}_\text{\normalfont gd,gs}}}
\newcommand{\dbetahvegdg}[0]{\bm{\Delta\hat{\beta}_\text{\normalfont egd,gs}}}
\newcommand{\dbetahvegdgc}[0]{\bm{\Delta\hat{\beta}_\text{\normalfont egd,gs,c}}}

\newcommand{\Icf}[0]{\bm{I_\text{\normalfont cf}}}
\newcommand{\Igf}[0]{\bm{I_\text{\normalfont gf}}}
\newcommand{\Iegf}[0]{\bm{I_\text{\normalfont egf}}}
\newcommand{\I}[0]{\bm{I_\text{\normalfont egf}}}
\newcommand{\Icfi}[0]{\bm{I}^i_\text{\normalfont cf}}
\newcommand{\Iegfi}[0]{\bm{I}^i_\text{\normalfont egf}}
\newcommand{\Ii}[0]{\bm{I}^i_\text{\normalfont egf}}

\newcommand{\yv}[0]{\bm{y}}

\newcommand{\Xm}[0]{\bm{X}}
\newcommand{\Om}[0]{\bm{\Omega}}
\newcommand{\Omi}[0]{\bm{\Omega}^i}

\newcommand{\betav}[0]{\bm{\beta}}

\newcommand{\betahnollv}[0]{\bm{\hat{\beta}_0}}
\newcommand{\betahatolsv}[0]{\bm{\hat{\beta}}^\text{\normalfont OLS}}
\newcommand{\betahatols}[0]{\hat{\beta}^\text{\normalfont OLS}}
\newcommand{\betahatv}[0]{\bm{\hat{\beta}}}
\newcommand{\betahat}[0]{\hat{\beta}}
\newcommand{\betavti}[0]{\bm{\beta}(t_i)}

\newcommand{\lambdamax}[0]{\lambda_{\text{\normalfont max}}}
\newcommand{\imax}[0]{{i_{\text{\normalfont max}}}}
\newcommand{\R}[0]{\mathbb{R}}
\newcommand{\Ss}[0]{\bm{\hat{\Sigma}}}
\newcommand{\Imm}[0]{\bm{I}}
\newcommand{\gv}[0]{\bm{g}}

\newcommand{\vv}[0]{\bm{v}}

\newcommand{\svi}[0]{\bm{s}^i}
\newcommand{\si}[0]{s^i}
\newcommand{\gsvi}[0]{\bm{\zeta}^i}
\newcommand{\gsi}[0]{\zeta^i}
\newcommand{\ckv}[0]{\bm{c}^{k-1}}
\newcommand{\ck}[0]{c^{k-1}}

\newcommand{\sgn}[0]{\text{\normalfont sign}}
\newcommand{\N}{\mathcal{N}}
\newcommand{\Oo}{\mathcal{O}}

\newcommand*{\horzbar}{\rule[.5ex]{2.5ex}{0.5pt}}

\title{Elastic Gradient Descent, an Iterative Optimization Method Approximating the Solution Paths of the Elastic Net}

\author{%
  Oskar Allerbo \\
  Mathematical Sciences\\
  University of Gothenburg and Chalmers University of Technology\\
  \texttt{allerbo@chalmers.se} \\
  \And
  Johan Jonasson\\
  Mathematical Sciences\\
  University of Gothenburg and Chalmers University of Technology\\
  \texttt{jonasson@chalmers.se} \\
  \And
  Rebecka J\"ornsten\\
  Mathematical Sciences\\
  University of Gothenburg and Chalmers University of Technology\\
  \texttt{jornsten@chalmers.se} \\
}

\begin{document}
\maketitle

\begin{abstract}%
The elastic net combines lasso and ridge regression to fuse the sparsity property of lasso with the grouping property of ridge regression. The connections between ridge regression and gradient descent and between lasso and forward stagewise regression have previously been shown. Similar to how the elastic net generalizes lasso and ridge regression, we introduce elastic gradient descent, a generalization of gradient descent and forward stagewise regression.
We theoretically analyze elastic gradient descent and compare it to the elastic net and forward stagewise regression. Parts of the analysis are based on elastic gradient flow, a piecewise analytical construction, obtained for elastic gradient descent with infinitesimal step size.
We also compare elastic gradient descent to the elastic net on real and simulated data and show that it provides similar solution paths, but is several orders of magnitude faster. Compared to forward stagewise regression, elastic gradient descent selects a model that, although still sparse, provides considerably lower prediction and estimation errors.
\end{abstract}


\section{Introduction}
Lasso \citep{tibshirani1996regression} is a popular method for combining regularization and model selection in linear regression. The objective is to minimize
\begin{equation}
\label{eq:lasso}
\frac 1{2n}\|\yv-\Xm\betav\|_2^2+\lambda \|\betav\|_1
\end{equation}
with respect to the parameter vector $\betav\in \R^p$, where $\Xm \in \R^{n\times p}$ is the design matrix, $\yv \in \R^n$ is the response vector, and $\lambda>0$ is the regularization strength. Provided that the regularization is large enough, the lasso estimates of some parameters in $\betav$ become exactly zero, thus eliminating the corresponding variables from the model, which results in a simpler representation. Since the introduction of lasso, many extensions have been proposed, such as the adaptive lasso \citep{zou2006adaptive}, with individual regularization strengths to each $\beta_i$; the group lasso \citep{yuan2006model}, which regularizes predefined groups of parameters together; the fused lasso \citep{tibshirani2005sparsity}, which accounts for spatial and/or temporal dependencies; the graphical lasso \citep{friedman2008sparse}, for sparse inverse covariance estimation; and the elastic net \citep{zou2005regularization}, which is a convex combination of lasso and ridge regression, generalizing Equation \ref{eq:lasso} into
\begin{equation}
\label{eq:en}
\frac 1{2n}\|\yv-\Xm\betav\|_2^2+\lambda(\alpha\|\betav\|_1+(1-\alpha)\|\betav\|_2^2),\ \alpha\in[0,1].
\end{equation}
The motivation behind adding the squared $\ell_2$ penalty of ridge regression to the elastic net is two-fold. First, in the high-dimensional setting, when $p>n$, lasso can select at most $n$ variables. Second, if two or more variables are highly correlated, lasso tends to include only one of these in the model, and to be quite indifferent as to which. Both of these shortcomings are alleviated by the elastic net.

As can be seen in Equation \ref{eq:lasso}, a larger value of $\lambda$ enforces a smaller value of $\|\betav\|_1$. Thus, provided $n\geq p$, the lasso estimate, $\betahatv$, shrinks (in $\ell_1$ norm) with increasing $\lambda$ from the ordinary least squares solution, $\betahatolsv:=(\Xm^\top\Xm)^{-1}\Xm^\top\yv$ for $\lambda=0$, to $\bm{0}$ for $\lambda \geq \lambdamax:=\frac1n\|\Xm^\top \yv\|_\infty$. Due to Lagrangian duality, the solution path of $\betahatv$ as a function of $\lambda$ from 0 to $\lambdamax$ can equivalently be expressed in terms of $\|\betahatv\|_1$, where $\lambda=0$ corresponds to $\|\betahatv\|_1=\|\betahatolsv\|_1$ and $\lambda=\lambdamax$ corresponds to $\|\betahatv\|_1=0$.

Several authors have addressed the striking similarities between the lasso solution path and the solution path of forward stagewise linear regression (see e.g.\ work by \citealt{rosset2004boosting}, \citealt{efron2004least} and \citealt{hastie2007forward}). Forward stagewise regression is an iterative method for solving linear regression. Starting at $\betahatv=\bm{0}$, the solution moves toward $\betahatolsv$, successively adding more variables to the model, resulting in a solution path very similar to that of lasso. Selecting a solution before convergence, something that is often referred to as early stopping, can thus be thought of as applying lasso with a regularization strength $\lambda \in(0,\lambdamax)$.
\citet{tibshirani2015general} proposed a generalization of forward stagewise regression to be used with any convex function as opposed to just the $\ell_1$ norm, and used it to obtain solution paths for group lasso, nuclear norm regularized matrix completion (e.g.\ \citet{candes2009exact}) and ridge logistic regression. \citet{vaughan2017stagewise} used the general stagewise procedure to obtain solution paths for sparse group lasso \citep{simon2013sparse}, while \citet{zhang2019forward} used it for clustering. 

Just as forward stagewise regression and lasso provide similar solution paths, so do gradient descent and ridge regression. \citet{ali2019continuous} investigated these similarities for infinitesimal optimization step size. They argued that, just as for forward stagewise regression, optimization time can be thought of as an inverted penalty, and that early stopping at time $t$ roughly corresponds to ridge regression with penalty $1/t$.

In this paper, we combine forward stagewise regression and gradient descent into elastic gradient descent, an iterative optimization method that produces a solution path similar to that of the elastic net. Analogously to how the elastic net is a combination of lasso and ridge regression, elastic gradient descent is a combination of forward stagewise regression and gradient descent.

In Section \ref{sec:egd}, we introduce the elastic gradient descent algorithm.
In Section \ref{sec:prop}, we theoretically analyze the algorithm, and compare it to the elastic net and to forward stagewise regression.
In Section \ref{sec:experiments}, we compare elastic gradient descent to the elastic net and forward stagewise regression on synthetic and real data sets. 

Our main contributions are:
\begin{itemize}
\item
We define elastic gradient descent, an iterative optimization algorithm, that generalizes forward stagewise regression (also known as coordinate descent) and gradient descent, with solution paths very similar to those of the elastic net. 
\item
We theoretically analyze the convergence properties of elastic gradient descent, the similarities and differences between elastic gradient descent with and without momentum, and the similarities and differences between elastic gradient descent and the elastic net and forward stagewise regression.
\item
We show on real and synthetic data that 
\begin{itemize}
\item
compared to the elastic net, elastic gradient descent selects similar models, but is orders of magnitude faster.
\item
compared to forward stagewise regression, elastic gradient descent is able to select a sparse model with considerably lower prediction and estimation errors.
\end{itemize}
\end{itemize}
All proofs are deferred to Appendix \ref{sec:proofs}.

\section{Elastic Gradient Descent}
\label{sec:egd}
Gradient descent is an iterative optimization method, where, in each time step, the solution is updated in the direction of the negative gradient. For the related method coordinate descent, each optimization step is constrained to update only one coordinate, namely the one with the largest absolute gradient value. For linear regression, coordinate descent and forward stagewise regression coincide, and thus we will henceforth use the name coordinate descent. For both coordinate and gradient descent, one optimization step can be expressed as
\begin{equation}
\label{eq:cgd}
\begin{aligned}
\betahatv(t+\Delta t)=\betahatv(t)-\Delta t\cdot\dbetahv(t),
\end{aligned}
\end{equation}
where $\dbetahv$ differs between the two algorithms. 

We let $\gv$ denote the gradient of the loss function, i.e.\ $\gv(t):=\nabla_{\betav(t)}L\left(\Xm,\yv,\betav(t)\right)$. (For least squares, with $L\left(\Xm,\yv,\betav(t)\right)=\frac1{2n}\|\yv-\Xm\betav(t)\|_2^2$, $\gv(t)=-\frac1n\Xm^\top(\yv-\Xm\betav(t))$.) For coordinate descent, $\dbetahvcd$ is defined according to

\begin{equation*}
\label{eq:coord_desc_step}
\begin{aligned}
&m(t):=\argmax_d|g_d(t)|,\\
&\dbetahvcd(t)=\sgn(g_m(t))\cdot \bm{e_m}(t)=\Icd(t)\cdot\sgn(\gv(t)),
\end{aligned}
\end{equation*}
where $\bm{e_m}$ is the $m$-th standard basis vector, $\Icd$ is a matrix of only zeros, except $(\Icd)_{mm}$ which is 1, and the sign of vector $\gv$ is taken element-wise. Multiplying the matrix $\Icd$ with the vector $\sgn(\gv)$ we obtain a vector where all elements are zero, except element $m$ which is exactly $\sgn(g_m)$.
For gradient descent,
\begin{equation*}
\label{eq:grad_desc_step}
\begin{aligned}
&\dbetahvgd(t)=\gv(t)=\Igd\cdot \gv(t),
\end{aligned}
\end{equation*}
where $\Igd=\bm{I}$ is the identity matrix, which is included to emphasize the similarities to coordinate descent.

Naively combining coordinate and gradient descent with inspiration from the elastic net, Equation \ref{eq:en}, suggests that
$$\dbetahvegd(\alpha,t)=\alpha\cdot\dbetahvcd(t)+(1-\alpha)\cdot\dbetahvgd(t)=\alpha\cdot\Icd(t)\cdot \sgn(\gv(t))+(1-\alpha)\cdot\Igd\cdot \gv(t),\ \alpha\in[0,1],$$
where coordinate and gradient descent are recovered as special cases for $\alpha=1$ and $\alpha=0$. However, this proposal does not share the desirable model selection property of the elastic net since $\dbetahvgd$ updates all parameters at all time steps, thus making all parameters non-zero already in the first time step. Therefore, we need a combination with the ability to keep some parameters fixed. Hence, we define
\begin{equation}
\label{eq:egd}
\dbetahvegd(\alpha,t):=\Iegd(\alpha,t)\cdot\left(\alpha\cdot\sgn(\gv(t))+(1-\alpha)\cdot\gv(t)\right),
\end{equation}
where $\Iegd$ is a diagonal matrix with zeros and ones on the diagonal, such that $\Iegd(0,t)=\Igd=\bm{I}$ and $\Iegd(1,t)=\Icd(t)$.  $\Iegd$ could be defined in multiple ways. We, however, choose the following simple definition:

\begin{definition}[$\Iegd$]~\\
\label{dfn:iegd}
\begin{equation*}
\begin{aligned}
&\textrm{For } m(t)=\argmax_d|g_d(t)|,\\
&\Iegd(\alpha,t)_{d_1d_2}:=
\begin{cases}
1\textrm{ if } d_1=d_2=d \textrm{ and } |g_d(t)|\geq\alpha\cdot|g_m(t)|\\
0\textrm{ else}.
\end{cases}
\end{aligned}
\end{equation*}
\end{definition}

That is, for large gradient components, where ``large'' means ``larger than $\alpha$ times the maximum component'', the corresponding value in $\Iegd$ is 1, while for small components it is 0. Note that if $\alpha=0$, all components are considered large, while for $\alpha=1$ only the maximum component is. Our definition of large gradient components coincides with that by \citet{friedman2004gradient}, but the update directions differ since we include the signed gradient in $\dbetahvegd$.
The reason for including the sign gradient is for elastic gradient descent to generalize coordinate descent, and thus to obtain a distinct connection to the elastic net.

Elastic gradient descent is summarized in Algorithm \ref{alg:egd}.

\begin{algorithm}
\caption{Elastic Gradient Descent}
%\textbf{Input:} Parameters at time $t$, $\{(\bm{W_l})^t, (\bm{b_l})^t\}_{l=l}^L$; data, $\bm{x}$; learning rate, $\alpha$; regularization strength, $\lambda$.\\
%\textbf{Output:} Parameters at time $t+1$: $\{(\bm{W_l})^{t+1}, (\bm{b_l})^{t+1}\}_{l=l}^L$.\\
\begin{algorithmic}[1]
  \State Initialize $\betahv=\bm{0}$.
  \Repeat
  \State $\gv=\nabla_{\betahv}L\left(\Xm,\yv,\betahv\right)$, where $L(\cdot)$ denotes the loss function.
  \State $\Iegd=\text{diag}\left(\mathbb{I}\left[\frac{|\gv|}{\max_d|g_d|}\geq \alpha\right]\right)$, where $\mathbb{I}[\cdot]$ denotes the indicator function, which is taken element-wise, and where diag$(\cdot)$ creates a diagonal matrix from a vector.
  \State $\betahv=\betahv-\Delta t \cdot \Iegd\cdot\left(\alpha\cdot\sgn(\gv)+(1-\alpha)\cdot\gv\right)$.
  \Until{convergence or other stopping criterion.}
\end{algorithmic}
\label{alg:egd}
\end{algorithm}

In Figure \ref{fig:path_demo}, we demonstrate the similarities between the solution paths of explicit regularization and iterative optimization. We compare the solution paths of ridge regression and gradient descent (GD), lasso and coordinate descent (CD), and the elastic net and elastic gradient descent (EGD) for a simple linear model with two correlated parameters, $\beta_1$ and $\beta_2$. The solution paths of the corresponding algorithms are similar, although not identical.

Even though our definition of elastic gradient descent includes an element of arbitrariness, it proves to work well, as is shown in Section \ref{sec:experiments}. In Appendix \ref{sec:steep_general}, we investigate two slightly different definitions, based on the frameworks of steepest descent \citep{boyd2004convex} and the general stagewise procedure \citep{tibshirani2015general}. The three definitions provide virtually identical solutions.

\subsection{Elastic Gradient Flow}
\label{sec:egf}
Gradient descent with infinitesimal step size, $\Delta t$, is often referred to as gradient flow, which, since $\Delta t\to 0$, can be interpreted as a differential equation in training time, $t$. For some problems, including linear regression, this differential equation has a closed-form solution, which opens up for a better theoretical understanding of the algorithm. For elastic gradient descent, the corresponding differential equation becomes quite complicated. However, in Appendix \ref{sec:flow}, we use it to construct something we refer to as elastic gradient flow, in analogy with gradient flow.
Elastic gradient flow helps us to establish a theoretical connection between elastic gradient descent and the elastic net. To improve readability, in this section, we just state the equations of elastic gradient flow, and its special cases gradient flow and coordinate flow; for details, see Appendix \ref{sec:flow}.

\begin{figure}
  \center
  \includegraphics[width=1.\textwidth]{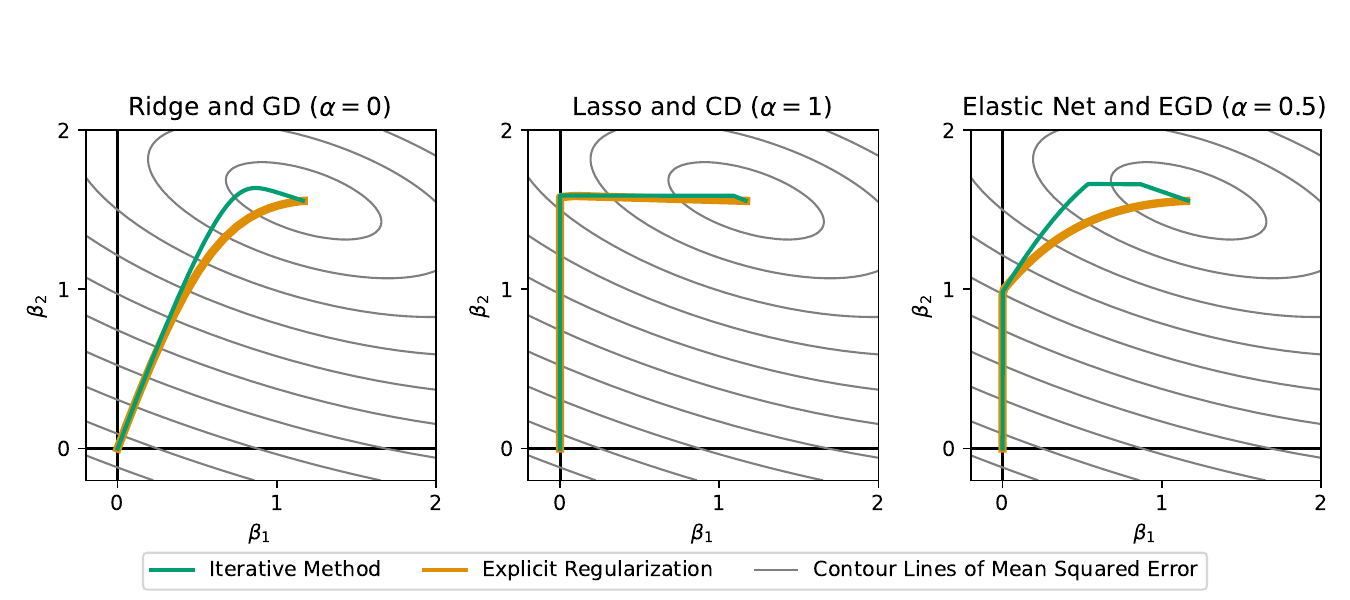}
  \caption{Solution paths of explicitly regularized and early stopping methods. In all three cases, the two methods follow similar, although not identical, solution paths.}
  \label{fig:path_demo}
\end{figure}

The elastic gradient flow estimate at time $t$ is given by Equation \ref{eq:beta_ef1},
\begin{equation}
\begin{aligned}
\label{eq:beta_ef1}
\betahvegf(t)=& \betahvegf(t_i)+ \left((1-\alpha)\Ss\right)^{-1}\left(\Imm-\exp\left(\Omi(t_i,t)\right)\right)\\
&\cdot\left(\alpha\cdot\sgn\left(\Ss(\betahatolsv-\betahvegf(t_i))\right)+(1-\alpha)\cdot\Ss\left(\betahatolsv-\betahvegf(t_i)\right)\right),\\
&t\in [t_i,t_{i+1}),
\end{aligned}
\end{equation}
where $\Ss:=\frac 1n \Xm^\top \Xm$ is the empirical covariance matrix, 
$\Omi(t_i,t)$ is the Magnus expansion \citep{magnus1954exponential} of $-\frac{1-\alpha}{1-\gamma}\Ss\Iegfi(\alpha,t)$,
$\{\Iegfi\}_{i=0}^\imax$ are the continuous-time versions of $\Iegd(\alpha,t)$ for $t\in [t_i,t_{i+1})$, 
and $\{t_i\}_{i=0}^\imax$ are the times when parameters enter or leave the model.

The parameter $\gamma\in [0,1)$ is the strength of the momentum \citep{polyak1964some}, which is a generalization of gradient descent discussed in Section \ref{sec:momentum}. For standard gradient descent without momentum, $\gamma=0$.

For $\alpha=0$, Equation \ref{eq:beta_ef1} simplifies to gradient flow,
\begin{equation}
\label{eq:beta_gf1}
\betahvgf(t)=\left(\Imm-\exp\left(-\frac t{1-\gamma}\Ss\right)\right)\betahatolsv,
\end{equation}
and for $\alpha=1$, it simplifies to what we refer to as coordinate flow,
\begin{equation}
\label{eq:beta_cf1}
\betahvcf(t)=\betahvcf(t_i)+\frac{t-t_i}{1-\gamma}\Icfi\cdot\sgn\left(\Ss(\betahatolsv-\betahvcf(t_i))\right),\quad t\in [t_i,t_{i+1}),
\end{equation}
where $\{\Icfi\}_{i=0}^\imax$ are continuous-time versions of $\Icd(t)$ for $t\in [t_i,t_{i+1})$.

\section{Properties of Elastic Gradient Descent}
\label{sec:prop}
In this section, we theoretically investigate elastic gradient descent and flow and make comparisons to the elastic net and coordinate descent, assessing the similarities and differences.

\subsection{Convergence of Elastic Gradient Descent}
For a small enough step size, $\Delta t$, elastic gradient descent always moves downhill in the optimization landscape. In Proposition \ref{thm:conv_anal}, we present bounds for the step size that guarantee an improvement when applying elastic gradient descent to a strongly convex problem.
\begin{proposition}~\\
\label{thm:conv_anal}
Assume that the loss function, $L(\betav)=L(\Xm,\yv,\betav)$, is strongly convex with Hessian bounded according to $\nabla^2L(\betav)\preceq M\Imm$ (i.e.\ $M\Imm-\nabla^2L(\betav)$ is a positive semi-definite matrix), for some $M>0$.
Denote 
$$\gv:=\nabla L(\betav),\ \gmax=g_m:=\max_d|g_d|\text{ and }\gmin:=\min_{\substack{d:\ |g_d|\geq \alpha,\\ g_d\neq 0}}|g_d|.$$

Then, 
\begin{subequations}
\begin{align}
&\Delta t< \frac2M\cdot\gmax\cdot\frac{\alpha+(1-\alpha)\gmax}{(1_{\alpha>0}+(1-\alpha)\gmax)^2}\label{eq:bound1}\\ 
\implies &\Delta t< \frac2M\cdot\frac{\gmin^2}{\gmax}\cdot \frac{\alpha+(1-\alpha)\gmax}{(\alpha+(1-\alpha)\gmin)^2}\label{eq:bound2}\\
\implies &L(\betahv-\Delta t \cdot \dbetahvegd)-L(\betahv)\leq 0,\notag
\end{align}
\end{subequations}
where 
$$1_{\alpha>0}=\begin{cases}1\text{ if } \alpha>0\\0\text{ if } \alpha=0.\end{cases}$$
\end{proposition}
\noindent\textbf{Remark 1:} The bound in \ref{eq:bound2} allows for a greater value of $\Delta t$ than that in Equation \ref{eq:bound1}, but requires knowledge of the minimum gradient value in addition to the maximum gradient value.\\
\noindent\textbf{Remark 2:} For $\alpha=0$ and $\alpha=1$, the bounds for $\Delta t$ become $\frac2M$ and $\frac{2\gmax}M$ respectively. Note than for $\alpha=1$, $\gmin=\gmax$.\\ 
\noindent\textbf{Remark 3:} If a fixed value is used for $\Delta t$, once $\gmax$ gets small enough, the loss function is not guaranteed to decrease, unless $\alpha=0$. In this case, training should be interrupted when the solution starts to worsen.\\
\noindent\textbf{Remark 4:} For linear regression, $M$ is the maximum eigenvalue of the empirical covariance matrix, $\Ss=\frac1n\Xm^\top\Xm$.

\subsection{Calculating Solution Paths}
In the original implementation of the elastic net, ridge regression in the penalized version,
$$\min_{\betav} \frac 1{2n} \|\yv-\Xm\betav\|_2^2+\lambda_2\|\betav\|_2^2,$$
is combined with the LARS algorithm \citep{efron2004least}, which solves the constrained version of the lasso problem,
$$\min_{\betav} \frac 1{2n} \|\yv-\Xm\betav\|_2^2,\textrm{ s.t. } \|\betav\|_1\leq R_1,$$
returning the entire solution path as a function of $R_1$.
This implies that the elastic net problem is formulated as
$$\min_{\betav} \frac 1{2n} \|\yv-\Xm\betav\|_2^2+\lambda_2\|\betav\|_2^2,\textrm{ s.t. } \|\betav\|_1\leq R_1$$
and that call to the algorithm returns the solution path for $\betahv$ for different values of $R_1$, with a fixed value of $\lambda_2$. Thus, each solution corresponds to a combination ($R_1$, $\lambda_2$), rather than the more intuitive combination ($\alpha$, $\lambda$). Later versions, including those by \citet{friedman2010regularization}, use iterative methods to obtain solutions expressed as combinations of ($\alpha$, $\lambda$), where the solution for a given $\lambda$ is calculated independently of the others by running an iterative algorithm to convergence.

Elastic gradient descent is also an iterative algorithm, but here the solution at each iteration is of interest by itself and corresponds to a combination $(\alpha,t)$. Running the algorithm to convergence once returns all values of $t$ between 0 and $t_\textrm{max}$. In contrast, the elastic net algorithm has to be run to convergence once for every value of $\lambda$. 

Comparing elastic gradient descent to coordinate descent, while there is no restriction on the number of parameters elastic gradient descent can update in each iteration, coordinate descent always only updates one parameter per iteration. Thus, especially for problems with many dimensions, elastic gradient descent has a computational advantage compared to coordinate descent.

In Section \ref{sec:experiments}, we verify the faster computational speed of elastic gradient descent compared to those of the elastic net and coordinate descent.

\subsection{Differences in the Solution Paths}
\label{sec:path_diff}
For $\alpha>0$, both the elastic net and elastic gradient descent tend to set some parameters to 0, but this is done using two different techniques.
The elastic net has no closed-form solution, unless for isotropic features, i.e.\ $\bm{\Sigma}=\Imm$, for which the solution is given by
\begin{equation}
\label{eq:beta_en}
\betahat^{\text{en}}_d(\lambda)=\frac{\sgn(\betahatols_d)\cdot\max(0,|\betahatols_d|-\alpha\lambda)}{1+(1-\alpha)\lambda}.
\end{equation}
Consider the numerator of Equation \ref{eq:beta_en}. Compared to the ordinary least squares solution, each $\betahat^{\text{en}}_d$ is translated toward zero, and once it changes sign it is set to exactly zero, i.e.\ the elastic net shifts all paths toward 0.
Elastic gradient descent, in contrast, by Definition \ref{dfn:iegd} stops updating a parameter when the corresponding gradient value is small. If this occurs when the parameter value is 0, the value will constantly remain so, but it might also stay constant at some other level. This is illustrated in Figures \ref{fig:path_demo}, \ref{fig:diff_demo} and \ref{fig:diab_paths_zoom}.

\subsection{Susceptibility to Correlations}
\label{sec:corrs}
The ridge estimate is usually written as $\betahatv(\lambda):=(\Xm^\top\Xm+n\lambda \Imm)^{-1}\Xm^\top\yv$, but according to Lemma \ref{thm:ridge_rewrite} it can be reformulated in a way that resembles the gradient flow estimate.
\begin{lemma}~\\
\label{thm:ridge_rewrite}
With $\Ss:=\frac 1n \Xm^\top \Xm$ and $\bm{\hat{\beta}}^\textrm{\emph{OLS}}:=(\Xm^\top\Xm)^+\Xm^\top\yv$, where $(\cdot)^+$ denotes the Moore-Penrose pseudoinverse, the ridge estimate can be written as
\begin{equation}
\label{eq:beta_ridge}
\betahatv(\lambda)=\left(\Imm-\left(\Imm+\frac1 \lambda \Ss\right)^{-1}\right)\betahatolsv.
\end{equation}
\end{lemma}
Comparing Equation \ref{eq:beta_ridge} to Equation \ref{eq:beta_gf1},
$$\betahvgf(t)=\left(\Imm-\exp\left(-\frac t{1-\gamma}\Ss\right)\right)\betahatolsv=\left(\Imm-\exp\left(\frac t{1-\gamma}\Ss\right)^{-1}\right)\betahatolsv,$$
we see that if we define $\lambda:=(1-\gamma)/t$, the ridge estimate can be thought of as a first-order Taylor approximation of the gradient flow estimate. 
The fact that the ridge estimate depends linearly on $\Ss$, whereas the gradient flow estimate depends exponentially, suggests that elastic gradient descent takes correlations into larger consideration than the elastic net does, with an even stronger tendency to, for standardized data, assign similar parameter values to correlated variables. This is further illustrated in Section \ref{sec:experiments}.

%\subsection{The Connection between $\lambda$ and $t$}
\subsection{The Connection between \texorpdfstring{$\lambda$}{lambda} and \texorpdfstring{$t$}{t}}
\label{sec:lbda_t}
As stated above, lasso and the elastic net have no closed-form solutions, unless for isotropic features, where the elastic net solution is given by Equation \ref{eq:beta_en},
which for lasso ($\alpha=1$) simplifies to
\begin{equation}
\begin{aligned}
\label{eq:beta_l}
\betahat^{\text{lasso}}_d(\lambda)=\sgn(\betahatols_d)\cdot\max(0,|\betahatols_d|-\lambda).
\end{aligned}
\end{equation}
In Proposition \ref{thm:t_lbda} we investigate the connection between Equations \ref{eq:beta_en} (elastic net with isotropic features) and \ref{eq:beta_ef1} (elastic gradient flow), and, as a special case, between Equations \ref{eq:beta_l} (lasso with isotropic features) and \ref{eq:beta_cf1} (coordinate flow) when $\Ss=\Imm$ by requiring $\betahat^\text{\normalfont en}_d(\lambda)=(\betahvegf(t))_d=:\betahat^{\text{\normalfont egf}}_d(t)$.
\begin{proposition}~\\
\label{thm:t_lbda}
Solving $\betahat^\text{\normalfont en}_d(\lambda)=\betahat^{\text{\normalfont egf}}_d(t):=(\betahvegf(t))_d$ for $\Ss=\Imm$, with $\betahvegf(t)$ according to Equation \ref{eq:beta_ef1} and $\betahat^{\text{\normalfont en}}_d(\lambda)$ according to Equation \ref{eq:beta_en}, we obtain
\begin{equation}
\label{eq:t_lbda_en}
\lambda_d=\max\left(\frac{|\betahatols_d|-|\betahat^{\text{\normalfont egf}}_d(t_i)|-\vv}
{\alpha +(1-\alpha) \left(|\betahat^{\text{\normalfont egf}}_d(t_i)|+\vv\right)},0\right),
\end{equation}
where $$\vv=\frac1{1-\alpha}\left(1-\exp\left(-\frac{1-\alpha}{1-\gamma}\int_{t_i}^t(\Ii)_{dd}(\alpha,\tau)d\tau\right)\right)\left(\alpha+(1-\alpha)\left(|\betahatols_d|-|\betahat^{\text{\normalfont egf}}_d(t_i)|\right)\right),$$
which implies $\frac{\partial \lambda_d(t)}{\partial t}\leq 0$.

For $\alpha=1$, Equation \ref{eq:t_lbda_en} reduces to
\begin{equation*}
\begin{aligned}
\label{eq:t_lbda_cf}
\lambda_d=\max\left(|\betahatols_d|-|\betahat^{\text{\normalfont cf}}_d(t_i)|-\frac{t-t_i}{1-\gamma}(\Icfi)_{dd},0\right).
\end{aligned}
\end{equation*}
\end{proposition}
We note that while for gradient flow, the relationship between $\lambda$ and $t/(1-\gamma)$ is approximately the multiplicative inverse, $\lambda\approx (1-\gamma)/t$, for coordinate flow it is approximately the additive inverse, $\lambda \approx -t/(1-\gamma)$.  For the elastic net, it is something in between. 
Furthermore, for the elastic net, the relationship between $\lambda$ and $t$ depends on $\int_{t_i}^t(\Iegfi)_{dd}(\alpha,\tau)d\tau$. When this integral is close to zero for a parameter, the relationship between $\lambda_d$ and $t$ becomes almost linear, while for a larger value, the relation becomes almost exponential. Since the value of the integral may vary with $d$, for a given value of $\alpha$ the relation between $\lambda_{d}$ and $t$ might be almost linear for some parameters and almost exponential for others. Furthermore, since $\Iegfi$ is recalculated at times $t_i$, for some time $t_i$ the relation might change between almost linear and almost exponential for a parameter.

In summary, Proposition \ref{thm:t_lbda} reveals that, while always decreasing with optimization time, the rate of the decrease of the regularization might vary substantially, both between parameters and during optimization.

\subsection{The Effect of Momentum}
\label{sec:momentum}
Momentum \citep{polyak1964some} is a way to introduce memory into gradient-based optimization methods. The idea is to increase the computational stability and speed by allowing not only for current, but also for past, gradient values to influence the update direction, analogous to how a ball rolls down a slope: with increased momentum (and speed), it does not respond immediately to changes in the slope.
Introducing momentum, Equation \ref{eq:cgd} generalizes into
$$\betahv(t+\Delta t)=\betahv(t)+\gamma\left(\betahv(t)-\betahv(t-\Delta t)\right)-\Delta t\cdot \gv(t),$$
where $\gamma \in [0,1)$ is the strength of the momentum.

For elastic gradient flow, and its special cases gradient and coordinate flow, $\gamma>0$ has the effect of rescaling the gradient selection matrix. $\Igf$ is replaced by $\frac1{1-\gamma}\Igf$ (where $\Igf=\Imm$), $\Icf$ by $\frac1{1-\gamma}\Icf$, and $\Iegf$ by $\frac1{1-\gamma}\Iegf$. Thus, for (very) small step sizes, momentum does not affect the solution path, it just increases the speed. 

For larger step sizes, the addition of momentum may, in addition to increasing the computational speed, change the solution path. With momentum, the gradient values at the beginning of the training contribute more to the solution, than without it. For elastic gradient descent, in the early stages of training many parameters have small gradient values and are not yet included in the model. This suggests that elastic gradient descent with momentum would promote sparser models, compared to elastic gradient descent without momentum, something that is supported by the experiments in Section \ref{sec:experiments} and Appendix \ref{sec:more_exps}.

\section{Experiments}
\label{sec:experiments}
In this section, we compare elastic gradient descent with and without momentum to the elastic net and coordinate descent on twelve different data sets. In order to illustrate the path differences between elastic gradient descent and the elastic net as discussed in Section \ref{sec:path_diff}, we use a very simple data set with only three variables, and the diabetes data set used by \citet{efron2004least}.\footnote{Available at \url{https://web.stanford.edu/~hastie/Papers/LARS/diabetes.data}.}
We then compare model selection accuracy and performance on a synthetic data set consisting of two blocks of parameters, where one block is included in the true model, and the other is not, for different correlations. Finally, we compare the performance of the algorithms on nine relatively large real data sets.

For elastic gradient and coordinate descent, a step size of 0.01 was used in all experiments except for the first, simple experiment, where 0.001 was used. We stopped the training when the training error no longer decreased, which, for $\alpha>0$, eventually happens according to Proposition \ref{thm:conv_anal}.
For elastic gradient descent with momentum, we consistently used $\gamma=0.5$.
For the elastic net the \texttt{enet\_path} method in the Scikit-learn library \citep{scikit-learn} was used.
All experiments, except those in Section \ref{sec:synth_exp} (and the corresponding additional experiments in Appendix \ref{sec:more_exps}), were run in Python on a Dell Latitude 7480 laptop, with an Intel Core i7, 2.80 GHz processor with four kernels. The experiments in Section \ref{sec:synth_exp}, and the corresponding experiments in the appendix, were run on a cluster with Intel Xeon Gold 6130, 2.10 GHz processors.

\subsection{Solution Paths for Simple Synthetic Data}
\label{sec:exp_diff}
To illustrate the different path properties of elastic gradient descent and the elastic net, 1000 observations were generated according to
$$\Xm\sim\N\left(\bm{0},
\begin{bmatrix}
1 & 0.7 & 0.7\\
0.7 & 1 & 0.7\\
0.7 & 0.7 & 1\\
\end{bmatrix}
\right),\ 
\yv=\Xm
\begin{bmatrix}
1\\ 0.1\\ 0\\
\end{bmatrix}.$$

The solution paths for four different values of $\alpha$ are shown in Figure \ref{fig:diff_demo}. For $\alpha=0$, due to the correlations in the data, initially, all parameter estimates aim toward values somewhere between 0 and 1. As $t$ increases ($\lambda$ decreases), the estimates start approaching their true values. Elastic gradient descent is more affected by the correlations, i.e.\ the parameter estimates move together for a larger fraction of the solution path, than the elastic net is, which is in line with the observations in Section \ref{sec:corrs}, and also tends to affect the model selection properties.
%The two middle columns, where $\alpha\in(0,1)$, suggest that elastic gradient descent is better at including true positives than the elastic net is, but that the elastic net is better at excluding true negatives. 
While elastic gradient descent includes the true positive $\beta_2$ for the entire solution paths for $\alpha=0.5$ and $\alpha=0.7$, this is not the case for the elastic net. On the other hand, for $\alpha=0.5$, elastic gradient descent erroneously includes true negative $\beta_3$ for a larger fraction of the solution path than the elastic net does.
As $\alpha$ increases, in both cases, the maximum values of the paths $\beta_2$ and $\beta_3$ are reduced, but while elastic gradient descent "cuts the peak" from above, the elastic net translates the entire path downward. The "peak cutting" behavior of elastic gradient descent comes from the fact that the gradient is the smallest just before changing sign, at the top of the peak. 

\begin{figure}
  \center
  \includegraphics[width=1.\textwidth]{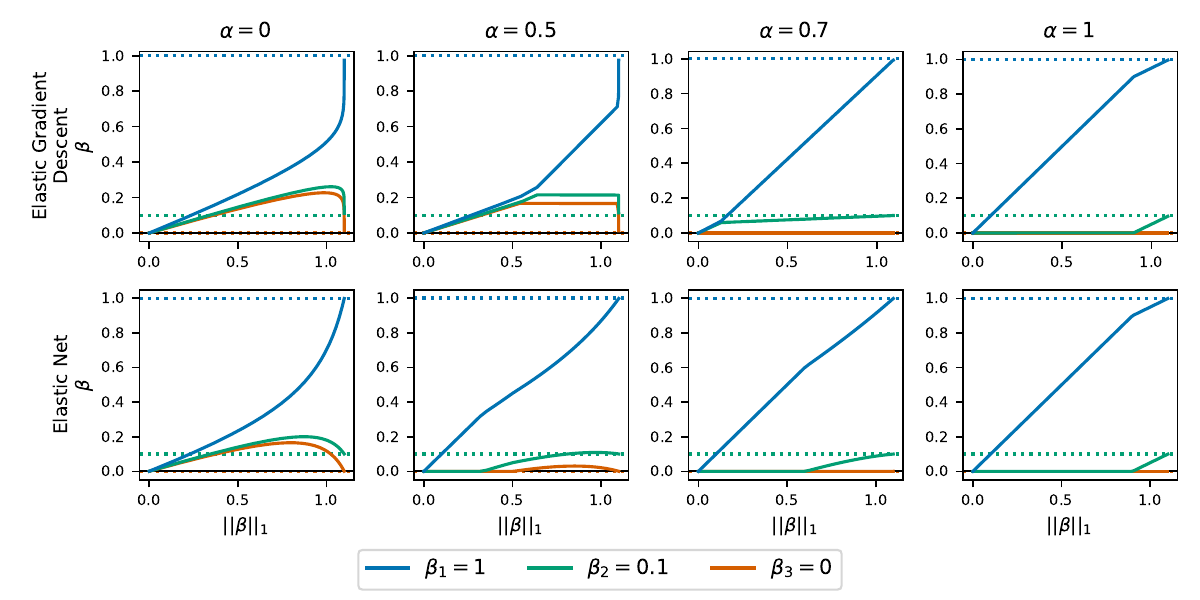}
  \caption{Comparison between elastic gradient descent (without momentum) and the elastic net on highly correlated data. In both cases, the drift toward 1 of the second parameter reduces with increasing $\alpha$, but while elastic gradient descent cuts the peak from above, the elastic net moves the entire path downward. In contrast to the elastic net, elastic gradient descent correctly includes $\beta_2$ for the whole solution paths when $\alpha=0.5$ and $\alpha=0.7$, but erroneously includes $\beta_3$ for a larger fraction of the solution path for $\alpha=0.5$.}
  \label{fig:diff_demo}
\end{figure}

\subsection{Solution Paths for the Diabetes Data}
\label{sec:diabetes}
The diabetes data set contains 442 observations, each consisting of 1 target value, which measures disease progression, and the 10 covariates \texttt{age}, \texttt{sex}, \texttt{bmi} (body mass index), \texttt{bp} (average blood pressure), \texttt{tc} (t-cells), \texttt{ld} (low-density lipoproteins), \texttt{hdl} (high-density lipoproteins), \texttt{tch} (thyroid stimulating hormone), \texttt{ltg} (lamotrigine) and \texttt{glu} (blood sugar level).

In Figure \ref{fig:diab_paths_zoom}, we show the solution paths of elastic gradient descent, without momentum, and the elastic net for two different values of $\alpha$. Similar to in Figure \ref{fig:diff_demo}, elastic gradient descent cuts peaks from above, while the elastic net translates them toward zero. It can be seen how this difference makes the algorithms behave differently for small values of $\|\betav\|_1$. While elastic gradient descent tends to include a subset of the parameters in the model immediately, the inclusion of the same set is more spread out for the elastic net. This contributes to elastic gradient descent proposing fewer models along the solution path than the elastic net does. Excluding the empty model, elastic gradient descent proposes 3 different models for $\alpha=0.3$ and 7 models for $\alpha=0.7$. The corresponding numbers for the elastic net are 10 and 11. If it were to be taken into account that the elastic net proposes the same model at different, non-adjacent sections along the path, its numbers would be even higher. This suggests that in terms of model selection, elastic gradient descent is more robust with respect to the degree of penalization than the elastic net is.

\begin{figure}
  \center
  \includegraphics[width=1.\textwidth]{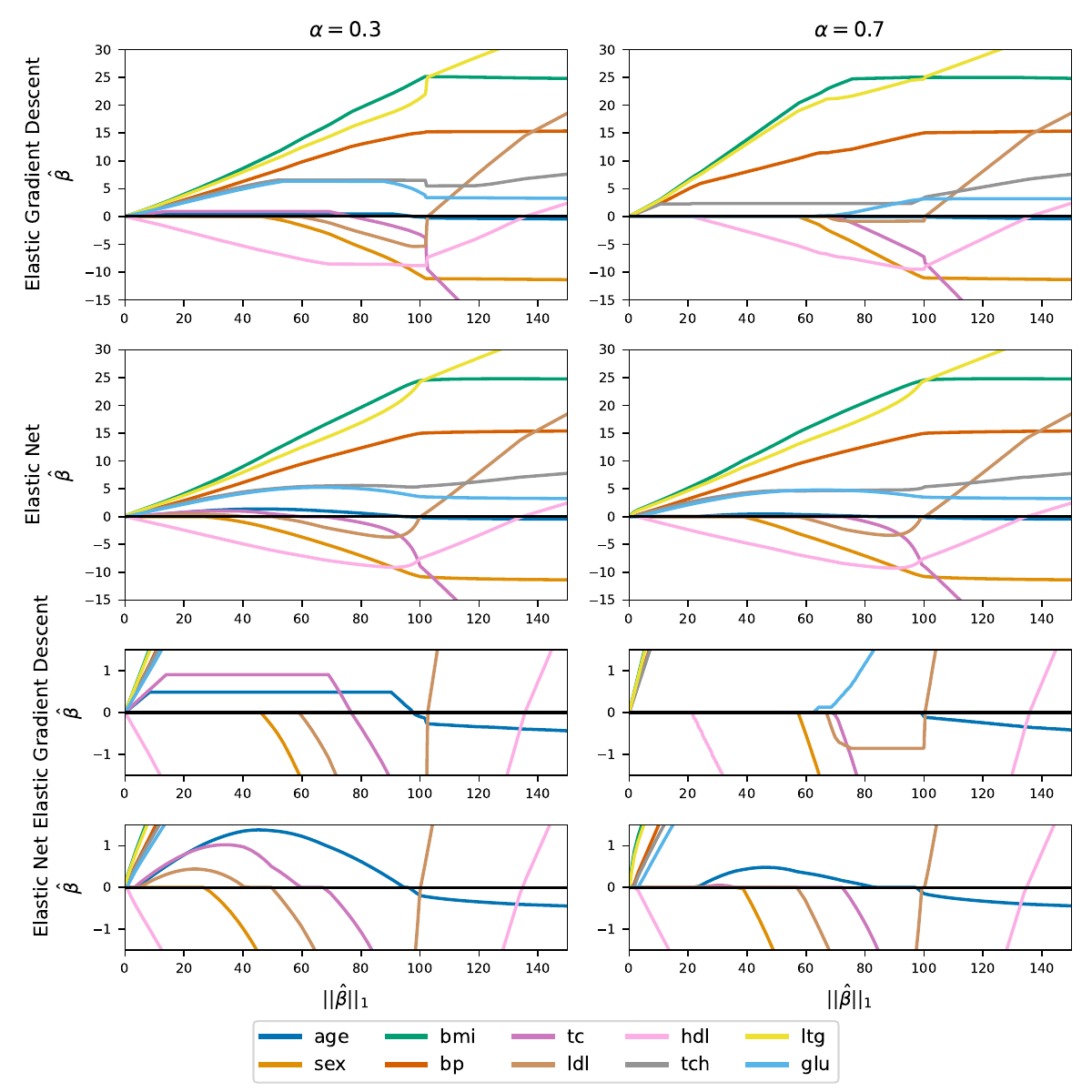}
  \caption{Solution paths for elastic gradient descent (without momentum) and the elastic net on the diabetes data. Rows three and four show the same things as rows one and two but on different y-scales. While elastic gradient descent cuts peaks from above, the elastic net translates them toward zero. Elastic gradient descent is more robust in terms of model selection with respect to the degree of penalization, proposing fewer different models along the solution path than the elastic net does.}
  \label{fig:diab_paths_zoom}
\end{figure}

In Figure \ref{fig:diab_paths_grad}, we compare the solution paths and the normalized gradients for elastic gradient descent with $\alpha=0.5$, coordinate descent, and the elastic net. Note that the bottom right panel does not show the gradients of the elastic net, since there are none, but instead the gradients of the elastic gradient flow solution. Compared to coordinate descent, elastic gradient descent includes more parameters earlier, which is in line with the motivation behind the elastic net to include correlated covariates together.
Studying the gradients, it can be seen how the parameters are split into three sets, which we refer to as the free, coupled, and inactive sets, see Appendix \ref{sec:flow} for details. The free parameters all have normalized gradient values, $|g_d|/\|\gv\|_\infty$, larger than $\alpha$, and are updated freely. This group includes the maximum gradient parameter with $|g_m|/\|\gv\|_\infty=1$. For the inactive parameters, $|g_d|/\|\gv\|_\infty<\alpha$ and these parameters are not updated as can be seen in the first column. For the coupled parameters, $|g_d|/\|\gv\|_\infty$ oscillates around (for the descent algorithms) or equals (for the flow algorithm) $\alpha$. The coupled parameters are still updated but at a slower pace than the free ones. For coordinate descent, there are no free parameters, only coupled and inactive.
Toward the end of the training, when $\|\gv\|_\infty$ is small, we see oscillations in the gradients for coordinate descent and elastic gradient descent, which is in line with the conclusions from Proposition \ref{thm:conv_anal}.
\begin{figure}
  \center
  \includegraphics[width=1.\textwidth]{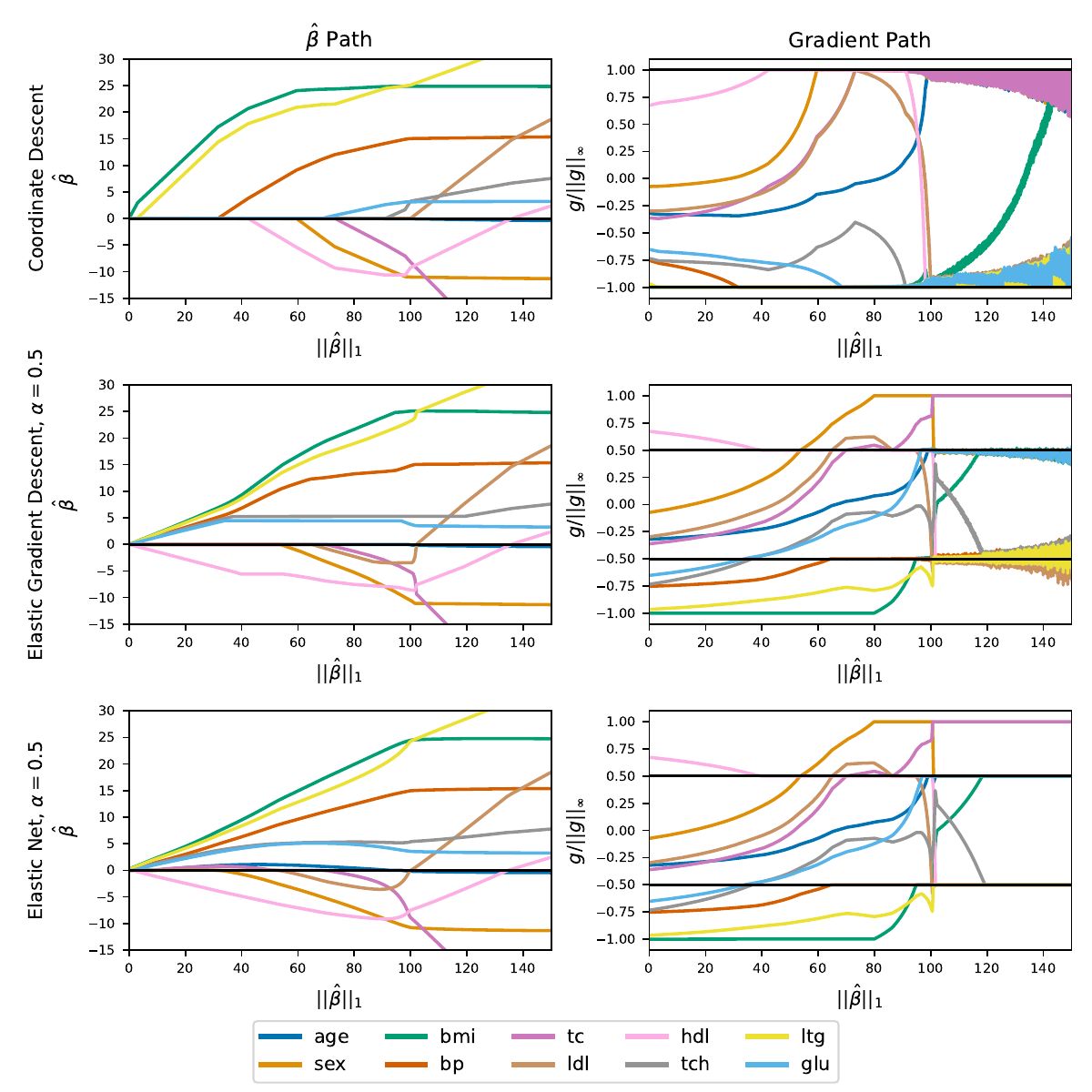}
  \caption{Solution paths and normalized gradients for coordinate descent, elastic gradient descent (without momentum), and the elastic net on the diabetes data. Note that the bottom right pane shows the gradients of elastic gradient flow.  We see how the parameters are split into the free, coupled, and inactive sets. Depending on whether $|g_d|$ is greater than, equal to, or smaller than $\alpha\cdot\|\gv\|_\infty$, the parameters update either freely, in a coupled fashion, or not at all, respectively. For instance, the \texttt{tch} parameter initially has an absolute normalized gradient value larger than $\alpha=0.5$ and updates freely. At $\|\betahv\|_1\approx 40$, the absolute normalized gradient becomes less than $\alpha$ and the parameter is not updated at all until $\|\betahv\|_1\approx 120$. Then it starts updating in a coupled fashion.}
  \label{fig:diab_paths_grad}
\end{figure}

\subsection{Synthetic Data for Model Selection}
\label{sec:synth_exp}
To compare model selection and performance, the following synthetic data set was created: The variables were split into two blocks of equal length, where the first block was included in the true model, and the second was not. The parameter values of the true positive variables were normally distributed with mean 2 and variance 1, the correlations within the two blocks were set to $\rho_1$, and between the two blocks to $\rho_2$:
\begin{equation*}
\begin{aligned}
\bm{\beta^*}&=\begin{bmatrix} \N(2,1)_{p/2}^\top & \bm{0}_{p/2}^\top\end{bmatrix}^\top\\
\bm{\Sigma_{11}}&=\bm{\Sigma_{22}}=\rho_1\cdot(\bm{1}\bm{1}^\top)_{p/2\times p/2}+(1-\rho_1)\cdot\Imm_{{p/2}\times {p/2}}\\
\bm{\Sigma_{12}}&=\bm{\Sigma_{12}}^\top=\rho_2\cdot(\bm{1}\bm{1}^\top)_{p/2\times p/2}\\
\bm{\Sigma}&= \begin{bmatrix} \bm{\Sigma_{11}} & \bm{\Sigma_{12}}\\ \bm{\Sigma_{12}}^\top & \bm{\Sigma_{11}} \end{bmatrix},
\end{aligned}
\end{equation*}
where $\Imm_{{p/2}\times {p/2}}$ denotes the $(p/2)\times (p/2)$ identity matrix, $(\bm{1}\bm{1}^\top)_{{p/2}\times {p/2}}$ denotes a $(p/2)\times (p/2)$ matrix of only ones, $\N(\cdot, \cdot)_{p/2}$ denotes an i.i.d.\ vector of length $p/2$ and $\bm{0}_{p/2}$ a vector of length $p/2$ with all zeros.
For $n=100/30/30$, where the three values of $n$ denote training, validation and testing sets, $n$ observations were sampled according to
\begin{equation*}
\begin{aligned}
\Xm&\sim\N(\bm{0},\bm{\Sigma}),\ \yv=\Xm\bm{\beta^*}+ \N(\bm{0},\sigma^2\Imm),
\end{aligned}
\end{equation*}
for $\sigma=10$, $\rho_1=0.7$, $\rho_2=0.3$ and $p\in[50,60,\dots, 200]$.
For each value of $p$, the experiment was repeated 5001 times for different data realizations. 
For elastic gradient descent and the elastic net, nine different values of $\alpha$ were considered, $\alpha\in\{0.1,\ 0.2,\dots\ 0.9\}$, and the combination of $(\alpha,t/\lambda)$ with the lowest mean squared error, MSE, on validation data was selected. For coordinate descent, where always $\alpha=1$, $t$ was selected by validation MSE.

The following test statistics were computed and compared between the three models:
\begin{itemize}
\item
Sensitivity (true positive rate).
\item
Specificity (true negative rate).
\item
Estimation error, $\frac1{n^*}\|\bm{X^*}\betahv-\bm{X^*}\bm{\beta^*}\|_2$, where $\bm{X^*}\in \R^{n^*\times p}$ is previously unseen data.
\item
Prediction error, $\frac1{p}\|\betahv-\bm{\beta^*}\|_2$.
\item
Execution time in seconds.
\end{itemize}

Figure \ref{fig:p_sweep} shows the median values together with the first and third quartiles across the 5001 realizations, for the different test statistics.

Elastic gradient descent and the elastic net perform similarly in all aspects except computational time, where elastic gradient descent performs significantly faster. The computational performance of elastic gradient descent improves with momentum. For high-dimensional data ($p>n=100$), where no unique solution exists, momentum also greatly improves the model specificity. These two results are in line with the discussion in Section \ref{sec:momentum}, according to which momentum increases the computational speed and promotes a sparser solution. The elastic net is more stable than elastic gradient descent in terms of specificity, at least in the absence of momentum, where elastic gradient descent, although performing well in general, sometimes includes all true negatives. In Appendix \ref{sec:more_exps}, we further examine the specificity properties by varying the experiment so that the number of non-zero parameters is constant when $p$ increases.

Compared to coordinate descent, elastic gradient descent performs better in all aspects except for specificity. The higher specificity of coordinate descent, however, comes at the cost of much worse sensitivity, and prediction and estimation errors.
The execution times of elastic gradient descent and the elastic net include testing for nine different values of $\alpha$, while for coordinate descent only one value of $\alpha$ is considered. Still, coordinated descent requires more computational time than elastic gradient descent. This can be attributed to the fact that coordinate descent updates only one parameter per iteration, something that becomes more apparent when $p$ is large.

The computational time of elastic gradient descent is less affected by the dimensionality than those of the elastic net and coordinate descent. Since the elastic net and coordinate descent algorithms only update one parameter per iteration, the dimensionality has quite a large impact on the execution time of these algorithms. On the other hand, elastic gradient descent may update multiple parameters per iteration, and the execution time is thus less affected by the dimensionality.

In Appendix \ref{sec:more_exps}, we extend the simulation, presenting results for all combinations of $\rho_1\in[0.5,\ 0.6,\ 0.7,\ 0.8,\ 0.9,]$, $\rho_2\in[0.0,\ 0.1,\ 0.2,\ 0.3,$ $0.4,\ 0.5]$ and $p\in[50,100,200]$. The conclusions are consistent with the ones presented here.

\begin{figure}
  \center
  \includegraphics[width=1.\textwidth]{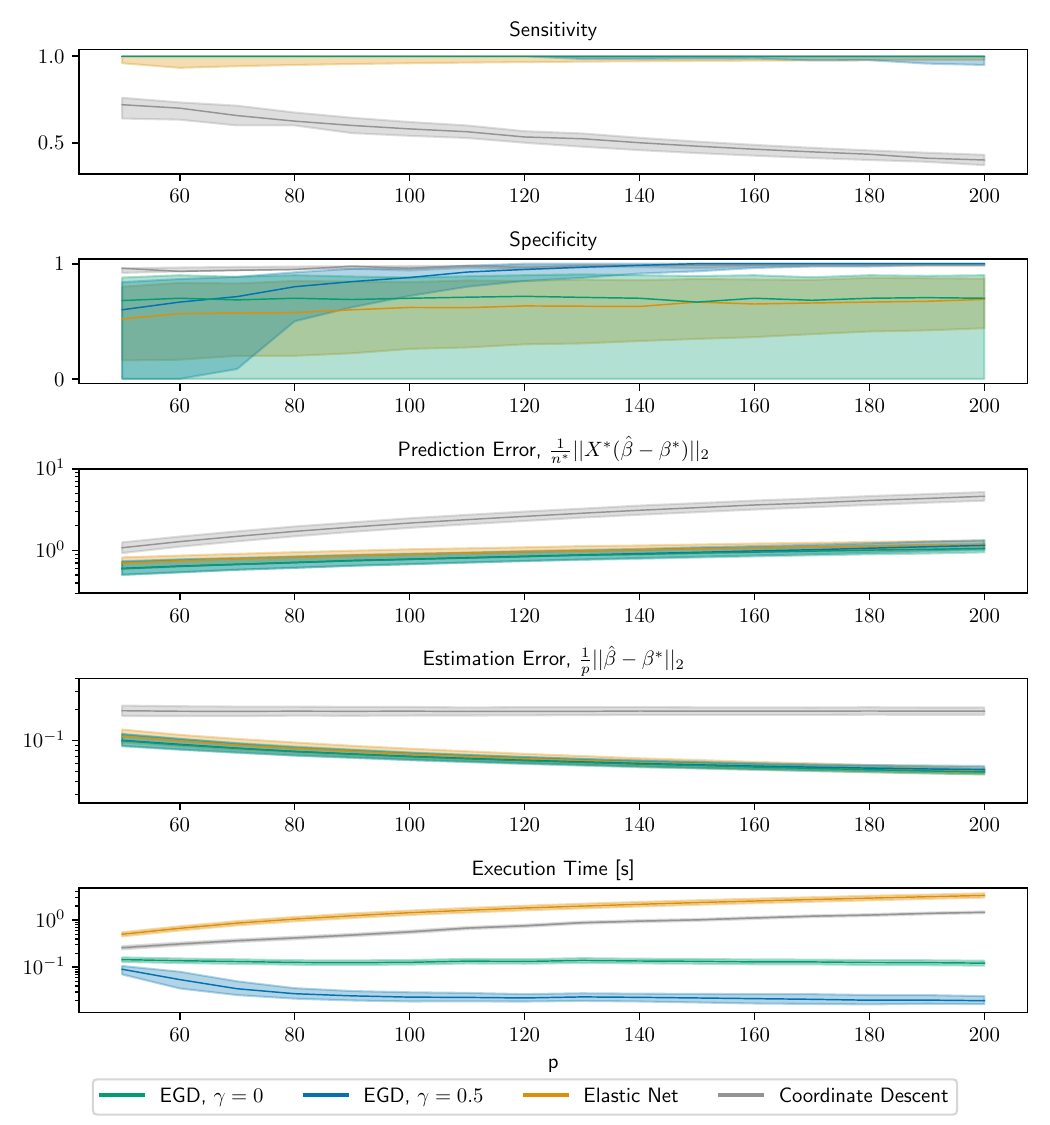}
  \caption{Median and first and third quartiles for the sensitivity, specificity, prediction and estimation errors, and execution time in seconds. Compared to the elastic net, elastic gradient descent performs similarly, except for execution time, where it is much faster. Compared to coordinate descent, elastic gradient descent performs better in all aspects except specificity. Elastic gradient descent performs faster, and has better specificity (especially when $p>n$), with momentum than without. The signal-to-noise ratio increases with the dimensionality and is, for some different values of $p$, $p=50$: 18.5, $p=100$: 72.5, $p=150$: 162, $p=200$: 287.}
  \label{fig:p_sweep}
\end{figure}

\subsection{Computational Efficiently on Real Data sets}
In this section, we compare elastic gradient descent with and without momentum to the elastic net and coordinate descent on the nine real data sets described in Table \ref{tab:datasets}. The data sets were selected to compare the algorithms on a diverse set of applications, although they all have in common that they are relatively large in terms of number of observations and/or dimensions.
The data was split 80\%/10\%/10\% into training, validation, and testing data for 5001 random splits. 
For elastic gradient descent and the elastic net, nine different values of $\alpha$ were considered, $\alpha\in\{0.1,\ 0.2,\dots\ 0.9\}$, and the combination of $(\alpha,t/\lambda)$ with the lowest mean squared error, MSE, on validation data was selected. For coordinate descent, where always $\alpha=1$, $t$ was selected by validation MSE.

\begin{table}
\centering
\begin{tabular}{|l|l|}
%\begin{tabularx}{\textwidth}{|X|>{\hsize=.2\hsize}X|}
\hline
Data set & Size, $n\times p$\\
\hline
\hline
Quality of aspen tree fibres\tablefootnote{Available at \url{https://openmv.net/info/wood-fibres}.} & $25165\times 5$\\
\hline
House values in California \citep{pace1997sparse}\tablefootnote{Available at \url{https://www.dcc.fc.up.pt/~ltorgo/Regression/cal_housing.html}.} & $20640\times  8$\\
\hline
\makecell[l]{Daily concentration of black smoke particles in the U.K.\\ in the year 2000 \citep{wood2017generalized}\tablefootnote{Available at \url{https://www.maths.ed.ac.uk/~swood34}.}} & $45568\times 10$\\
\hline
Results of the 2019 Portuguese Parliamentary Elections\tablefootnote{Available at \url{https://archive.ics.uci.edu/dataset/513/real+time+election+results+portugal+2019}.} & $21643\times 18$\\
\hline
\makecell[l]{Appliances energy use in a low energy building in\\Stambruges, Belgium \citep{candanedo2017data}\tablefootnote{Available at \url{https://github.com/LuisM78/Appliances-energy-prediction-data}.}} & $19735\times 27$\\
\hline
\makecell[l]{Protein structure as root-mean-square deviation of\\ atomic positions, taken from CASP\tablefootnote{\url{https://predictioncenter.org/}, Available at \url{https://archive.ics.uci.edu/ml/datasets/Physicochemical+Properties+of+Protein+Tertiary+Structure}.}} & $45730\times 9$\\
\hline
Critical temperature of superconductors\tablefootnote{Available at \url{https://archive.ics.uci.edu/dataset/464/superconductivty+data}.} & $21263\times 81$\\
\hline
Readability of texts used in English education\tablefootnote{Available at \url{https://www.kaggle.com/code/uocoeeds/building-a-regression-model-with-elastic-net/input}.} & $2834\times 768$\\
\hline
Topic popularity on Twitter \citep{kawala2016study}\tablefootnote{Available at \url{http://archive.ics.uci.edu/dataset/248/buzz+in+social+media}.} & $291624\times 77$\\
\hline
%\end{tabularx}
\end{tabular}
\caption{Real data sets used for comparing elastic gradient descent to the elastic net and coordinate descent.}
\label{tab:datasets}
\vspace{-.5cm}
\end{table}
The results are presented in Table \ref{tab:real_data}, where we compare execution time in seconds, R$^2$ (proportion of explained variation) on test data, and model size (number of non-zeros in $\betahv$), and are in line with those of Section \ref{sec:synth_exp}.

Elastic gradient descent and the elastic net perform similarly, apart from elastic gradient descent being up to three orders of magnitude faster. In contrast to coordinate descent, which is executed once, elastic gradient descent is executed nine times, for nine different values of $\alpha$. Still, for all data sets except the California housing and protein structure data sets it performs faster than nine times the speed coordinate descent (with momentum, it is faster for all data sets). When $p$ is large, elastic gradient descent tends to be faster than coordinate descent even in absolute numbers, with the exception of the English readability data.

Coordinate descent tends to select a sparser model, at the expense of a lower R$^2$. This is in line with the higher specificity of coordinate descent observed in Section \ref{sec:synth_exp}.
In addition to increasing the computational performance of elastic gradient descent, there is also a tendency for momentum to promote a sparser model, especially when $p$ is large (again with the exception of the English readability data).

\begin{table}
\centering
\small
\begin{tabular}{|l|l||l|l|l|}
\hline
\makecell{Data set\\$(n\times p)$} & Algorithm & \makecell{Execution Time\\in Seconds} & R$^2$ & Model Size \\
\hline
\hline
\multirow{4}{*}{\makecell{Aspen\\Fibres\\$(25165 \times 5)$}} & EGD, $\gamma=0$ & $0.16,\ (0.13, 0.18)$ & $0.49,\ (0.47, 0.51)$ & $5,\ (5, 5)$ \\
& EGD, $\gamma=0.5$                                                    & $0.11,\ (0.082, 0.13)$ & $0.49,\ (0.47, 0.51)$ & $5,\ (5, 5)$ \\
& Elastic Net                                                          & $6.2,\ (5.7, 6.4)$ & $0.49,\ (0.47, 0.51)$ & $5,\ (5, 5)$ \\
& CD                                                                   & $0.023,\ (0.013, 0.031)$ & $0.48,\ (0.46, 0.50)$ & $3,\ (3, 3)$ \\
\hline
\multirow{4}{*}{\makecell{California\\Housing\\$(20640 \times 8)$}} & EGD, $\gamma=0$ & $1.1,\ (1.1, 1.2)$ & $0.60,\ (0.59, 0.61)$ & $8,\ (8, 8)$ \\
& EGD, $\gamma=0.5$                                                    & $0.62,\ (0.57, 0.69)$ & $0.60,\ (0.59, 0.62)$ & $8,\ (7, 8)$ \\
& Elastic Net                                                          & $11,\ (9.7, 11)$ & $0.60,\ (0.59, 0.62)$ & $8,\ (7, 8)$ \\
& CD                                                                   & $0.083,\ (0.053, 0.11)$ & $0.57,\ (0.55, 0.59)$ & $5,\ (4, 6)$ \\
\hline
\multirow{4}{*}{\makecell{U.K.\ Black\\Smoke\\$(45568 \times 10)$}} & EGD, $\gamma=0$ & $0.26,\ (0.22, 0.30)$ & $0.14,\ (0.13, 0.14)$ & $10,\ (10, 10)$ \\
& EGD, $\gamma=0.5$                                                    & $0.20,\ (0.17, 0.23)$ & $0.14,\ (0.13, 0.14)$ & $10,\ (10, 10)$ \\
& Elastic Net                                                          & $27,\ (27, 28)$ & $0.14,\ (0.13, 0.15)$ & $10,\ (10, 10)$ \\
& CD                                                                   & $0.039,\ (0.022, 0.056)$ & $0.13,\ (0.13, 0.14)$ & $7,\ (7, 7)$ \\
\hline
\multirow{4}{*}{\makecell{Portugese\\Elections\\$(21643 \times 18)$}} & EGD, $\gamma=0$ & $0.61,\ (0.55, 0.67)$ & $0.11,\ (0.087, 0.12)$ & $14,\ (12, 15)$ \\
& EGD, $\gamma=0.5$                                                    & $0.35,\ (0.30, 0.40)$ & $0.11,\ (0.087, 0.12)$ & $12,\ (12, 15)$ \\
& Elastic Net                                                          & $74,\ (66, 83)$ & $0.10,\ (0.087, 0.12)$ & $12,\ (9, 14)$ \\
& CD                                                                   & $0.11,\ (0.083, 0.15)$ & $0.10,\ (0.086, 0.12)$ & $7,\ (5, 7)$ \\
\hline
\multirow{4}{*}{\makecell{Appliances\\Energy Use\\$(19735 \times 27)$}} & EGD, $\gamma=0$ & $4.2,\ (3.9, 4.6)$ & $0.16,\ (0.15, 0.17)$ & $27,\ (27, 27)$ \\
& EGD, $\gamma=0.5$                                                    & $2.9,\ (2.6, 3.1)$ & $0.16,\ (0.15, 0.17)$ & $27,\ (27, 27)$ \\
& Elastic Net                                                          & $83,\ (80, 85)$ & $0.16,\ (0.15, 0.18)$ & $27,\ (27, 27)$ \\
& CD                                                                   & $1.2,\ (1.1, 1.4)$ & $0.078,\ (0.073, 0.083)$ & $5,\ (4, 5)$ \\
\hline
\multirow{4}{*}{\makecell{Protein\\Structure\\$(45730 \times 9)$}} & EGD, $\gamma=0$ & $1.3,\ (1.2, 1.4)$ & $0.24,\ (0.24, 0.25)$ & $7,\ (7, 8)$ \\
& EGD, $\gamma=0.5$                                                    & $0.86,\ (0.73, 0.97)$ & $0.26,\ (0.26, 0.27)$ & $9,\ (9, 9)$ \\
& Elastic Net                                                          & $120,\ (120, 120)$ & $0.28,\ (0.27, 0.29)$ & $9,\ (9, 9)$ \\
& CD                                                                   & $0.11,\ (0.077, 0.14)$ & $0.15,\ (0.15, 0.16)$ & $2,\ (2, 2)$ \\
\hline
\multirow{4}{*}{\makecell{Super-\\conductors\\$(21263 \times 81)$}} & EGD, $\gamma=0$ & $2.6,\ (2.2, 2.8)$ & $0.70,\ (0.70, 0.71)$ & $81,\ (81, 81)$ \\
& EGD, $\gamma=0.5$                                                    & $1.1,\ (0.93, 1.2)$ & $0.66,\ (0.66, 0.67)$ & $77,\ (76, 77)$ \\
& Elastic Net                                                          & $340,\ (330, 360)$ & $0.72,\ (0.71, 0.73)$ & $64,\ (64, 65)$ \\
& CD                                                                   & $1.5,\ (1.2, 1.7)$ & $0.62,\ (0.62, 0.63)$ & $15,\ (14, 15)$ \\
\hline
\multirow{4}{*}{\makecell{English\\Readability\\$(2834 \times 768)$}} & EGD, $\gamma=0$ & $1.7,\ (1.5, 1.7)$ & $0.69,\ (0.67, 0.71)$ & $105,\ (89, 123)$ \\
& EGD, $\gamma=0.5$                                                    & $1.2,\ (1.0, 1.3)$ & $0.70,\ (0.68, 0.72)$ & $115,\ (97, 144)$ \\
& Elastic Net                                                          & $1400,\ (1300, 1500)$ & $0.73,\ (0.71, 0.74)$ & $304,\ (265, 374)$ \\
& CD                                                                   & $0.26,\ (0.22, 0.32)$ & $0.58,\ (0.56, 0.60)$ & $22,\ (21, 24)$ \\
\hline
\multirow{4}{*}{\makecell{Twitter\\Popularity\\$(291624 \times 77)$}} & EGD, $\gamma=0$ & $48,\ (41, 55)$ & $0.94,\ (0.89, 0.95)$ & $60,\ (57, 61)$ \\
& EGD, $\gamma=0.5$                                                    & $6.0,\ (5.3, 7.0)$ & $0.93,\ (0.89, 0.94)$ & $27,\ (24, 29)$ \\
& Elastic Net                                                          & $5100,\ (4800, 5300)$ & $0.94,\ (0.90, 0.94)$ & $49,\ (48, 50)$ \\
& CD                                                                   & $360,\ (350, 370)$ & $0.93,\ (0.92, 0.94)$ & $8,\ (7, 8)$ \\
\hline
\end{tabular}
\caption{Median and first and third quartiles (within parenthesis) of execution time (in seconds), R$^2$, and model size when applying the three algorithms on the nine real data sets. Elastic gradient descent performs significantly faster than the elastic net. Adding momentum further increases the computational speed. Coordinate descent tends to select a sparser model. While coordinate descent is evaluated only once (for $\alpha=1$), elastic gradient descent and the elastic net are evaluated nine times (for $\alpha\in\{0.1,\dots\ 0.9\}$).}
\label{tab:real_data}
\end{table}

\section{Conclusions}
We proposed elastic gradient descent, a simple-to-implement, iterative optimization method, which generalizes gradient descent and coordinate descent (forward stagewise regression). We also investigated the case of infinitesimal optimization step size, presenting a piecewise analytical solution for solving linear regression with elastic gradient flow. 

We compared elastic gradient descent with and without momentum to the elastic net and coordinate descent, both theoretically and on simulated and real data. Elastic gradient descent and the elastic net provided similar solutions, but with elastic gradient descent being up to three orders of magnitude faster on the investigated data. Compared to coordinate descent, elastic gradient descent selected a model with better performance, although still sparse. In addition to faster performance, adding momentum to elastic gradient descent promotes a sparser model for high dimensional data.

We used elastic gradient descent for standard linear regression. However, it would also be interesting to apply it for classification by extending it to logistic and multinomial regression. Furthermore, the optimization algorithm can be used instead of e.g.\ gradient descent on any optimization problem. For instance, it would be interesting to train a neural network with elastic gradient descent, obtaining a model that grows in complexity with optimization time.

Code is available at \url{https://github.com/allerbo/elastic_gradient_descent}.

\section*{Acknowledgments}
This research was supported by funding from the Swedish Research Council (VR), the Swedish Foundation for Strategic Research, the Wallenberg AI, Autonomous Systems and Software Program (WASP), and the Chalmers AI Research Center (CHAIR).\\
OA would like to thank the editor and an anonymous reviewer, whose suggestions helped to improve the manuscript substantially.

\clearpage
\bibliography{refs}

\begin{thebibliography}{}

\bibitem[Ali et~al., 2019]{ali2019continuous}
Ali, A., Kolter, J.~Z., and Tibshirani, R.~J. (2019).
\newblock A continuous-time view of early stopping for least squares
  regression.
\newblock In {\em The 22nd International Conference on Artificial Intelligence
  and Statistics}, pages 1370--1378. PMLR.

\bibitem[Boyd and Vandenberghe, 2004]{boyd2004convex}
Boyd, S.~P. and Vandenberghe, L. (2004).
\newblock {\em Convex Optimization}.
\newblock Cambridge university press.

\bibitem[Candanedo et~al., 2017]{candanedo2017data}
Candanedo, L.~M., Feldheim, V., and Deramaix, D. (2017).
\newblock Data driven prediction models of energy use of appliances in a
  low-energy house.
\newblock {\em Energy and Buildings}, 140:81--97.

\bibitem[Cand{\`e}s and Recht, 2009]{candes2009exact}
Cand{\`e}s, E.~J. and Recht, B. (2009).
\newblock Exact matrix completion via convex optimization.
\newblock {\em Foundations of Computational Mathematics}, 9(6):717--772.

\bibitem[Efron et~al., 2004]{efron2004least}
Efron, B., Hastie, T., Johnstone, I., Tibshirani, R., et~al. (2004).
\newblock Least angle regression.
\newblock {\em Annals of Statistics}, 32(2):407--499.

\bibitem[Friedman et~al., 2008]{friedman2008sparse}
Friedman, J., Hastie, T., and Tibshirani, R. (2008).
\newblock Sparse inverse covariance estimation with the graphical lasso.
\newblock {\em Biostatistics}, 9(3):432--441.

\bibitem[Friedman et~al., 2010]{friedman2010regularization}
Friedman, J., Hastie, T., and Tibshirani, R. (2010).
\newblock Regularization paths for generalized linear models via coordinate
  descent.
\newblock {\em Journal of Statistical Software}, 33(1):1.

\bibitem[Friedman and Popescu, 2004]{friedman2004gradient}
Friedman, J. and Popescu, B.~E. (2004).
\newblock Gradient directed regularization.
\newblock {\em Unpublished manuscript, http://www-stat. stanford. edu/\~{}
  jhf/ftp/pathlite. pdf}.

\bibitem[Hastie et~al., 2007]{hastie2007forward}
Hastie, T., Taylor, J., Tibshirani, R., Walther, G., et~al. (2007).
\newblock Forward stagewise regression and the monotone lasso.
\newblock {\em Electronic Journal of Statistics}, 1:1--29.

\bibitem[Kawala et~al., 2016]{kawala2016study}
Kawala, F., Gaussier, {\'E}., Douzal-Chouakria, A., and Diemert, E. (2016).
\newblock A study of different keyword activity prediction problems in social
  media.
\newblock {\em International Journal of Social Network Mining}, 2(3):224--255.

\bibitem[Magnus, 1954]{magnus1954exponential}
Magnus, W. (1954).
\newblock On the exponential solution of differential equations for a linear
  operator.
\newblock {\em Communications on Pure and Applied Mathematics}, 7(4):649--673.

\bibitem[Pace and Barry, 1997]{pace1997sparse}
Pace, R.~K. and Barry, R. (1997).
\newblock Sparse spatial autoregressions.
\newblock {\em Statistics \& Probability Letters}, 33(3):291--297.

\bibitem[Pedregosa et~al., 2011]{scikit-learn}
Pedregosa, F., Varoquaux, G., Gramfort, A., Michel, V., Thirion, B., Grisel,
  O., Blondel, M., Prettenhofer, P., Weiss, R., Dubourg, V., Vanderplas, J.,
  Passos, A., Cournapeau, D., Brucher, M., Perrot, M., and Duchesnay, E.
  (2011).
\newblock Scikit-learn: Machine learning in {P}ython.
\newblock {\em Journal of Machine Learning Research}, 12:2825--2830.

\bibitem[Polyak, 1964]{polyak1964some}
Polyak, B.~T. (1964).
\newblock Some methods of speeding up the convergence of iteration methods.
\newblock {\em USSR Computational Mathematics and Mathematical Physics},
  4(5):1--17.

\bibitem[Rosset et~al., 2004]{rosset2004boosting}
Rosset, S., Zhu, J., and Hastie, T. (2004).
\newblock Boosting as a regularized path to a maximum margin classifier.
\newblock {\em Journal of Machine Learning Research}, 5:941--973.

\bibitem[Simon et~al., 2013]{simon2013sparse}
Simon, N., Friedman, J., Hastie, T., and Tibshirani, R. (2013).
\newblock A sparse-group lasso.
\newblock {\em Journal of Computational and Graphical Statistics},
  22(2):231--245.

\bibitem[Tibshirani, 1996]{tibshirani1996regression}
Tibshirani, R. (1996).
\newblock Regression shrinkage and selection via the lasso.
\newblock {\em Journal of the Royal Statistical Society: Series B (Statistical
  Methodology)}, 58(1):267--288.

\bibitem[Tibshirani et~al., 2005]{tibshirani2005sparsity}
Tibshirani, R., Saunders, M., Rosset, S., Zhu, J., and Knight, K. (2005).
\newblock Sparsity and smoothness via the fused lasso.
\newblock {\em Journal of the Royal Statistical Society: Series B (Statistical
  Methodology)}, 67(1):91--108.

\bibitem[Tibshirani, 2015]{tibshirani2015general}
Tibshirani, R.~J. (2015).
\newblock A general framework for fast stagewise algorithms.
\newblock {\em Journal of Machine Learning Research}, 16(1):2543--2588.

\bibitem[Vaughan et~al., 2017]{vaughan2017stagewise}
Vaughan, G., Aseltine, R., Chen, K., and Yan, J. (2017).
\newblock Stagewise generalized estimating equations with grouped variables.
\newblock {\em Biometrics}, 73(4):1332--1342.

\bibitem[Wood et~al., 2017]{wood2017generalized}
Wood, S.~N., Li, Z., Shaddick, G., and Augustin, N.~H. (2017).
\newblock Generalized additive models for gigadata: Modeling the uk black smoke
  network daily data.
\newblock {\em Journal of the American Statistical Association},
  112(519):1199--1210.

\bibitem[Yuan and Lin, 2006]{yuan2006model}
Yuan, M. and Lin, Y. (2006).
\newblock Model selection and estimation in regression with grouped variables.
\newblock {\em Journal of the Royal Statistical Society: Series B (Statistical
  Methodology)}, 68(1):49--67.

\bibitem[Zhang, 2019]{zhang2019forward}
Zhang, M. (2019).
\newblock Forward-stagewise clustering: An algorithm for convex clustering.
\newblock {\em Pattern Recognition Letters}, 128:283--289.

\bibitem[Zou, 2006]{zou2006adaptive}
Zou, H. (2006).
\newblock The adaptive lasso and its oracle properties.
\newblock {\em Journal of the American Statistical Association},
  101(476):1418--1429.

\bibitem[Zou and Hastie, 2005]{zou2005regularization}
Zou, H. and Hastie, T. (2005).
\newblock Regularization and variable selection via the elastic net.
\newblock {\em Journal of the Royal Statistical Society: Series B (Statistical
  Methodology)}, 67(2):301--320.

\end{thebibliography}
\bibliographystyle{apalike}
\clearpage

\appendix

\section{Connection to Steepest Descent and the General Stagewise Procedure}
\label{sec:steep_general}
In this section, we redefine elastic gradient descent within the frameworks of steepest descent \citep{boyd2004convex} and the general stagewise procedure \citep{tibshirani2015general}, obtaining two related, but slightly different, flavors of Equation \ref{eq:egd}.

\subsection{Steepest Descent}
Steepest descent \citep{boyd2004convex}, generalizes coordinate and gradient descent. For a given norm $\|\cdot\|$, $\dbetahv$ is given by

\begin{equation}
\label{eq:steepest_desc}
\begin{aligned}
\dbetahvs(t)&=\argmax_{\vv:\ \|\vv\|=1}\gv(t)^\top\vv\\
\betahatv(t+\Delta t)&=\betahatv(t)-\Delta t\cdot \dbetahvs(t).
\end{aligned}
\end{equation}
For the $\ell_2$ norm, steepest descent becomes normalized gradient descent, 
$$\dbetahvgds(t)=\frac{\gv(t)}{\|\gv(t)\|_2}=\frac{\Igd\cdot\gv(t)}{\|\Igd\cdot\gv(t)\|_2},$$
while the $\ell_1$ norm corresponds to coordinate descent, 
$$\dbetahvcds(t)=\dbetahvcd(t)=\Icd(t)\cdot\sgn(\gv(t))=\frac{\Icd(t)\cdot\gv(t)}{\|\Icd(t)\cdot\gv(t)\|_1}.$$
When formulating elastic gradient descent, inspired by Equation \ref{eq:en}, we would like to use $\alpha \|\vv\|_1 +(1-\alpha)\|\vv\|_2^2=1$ in Equation \ref{eq:steepest_desc}, however then there is no analytical solution to the equation. Instead, we use the following strategy to obtain an approximate solution:
\begin{enumerate}
\item Define $\dbetahvegds$ as a generalization of both $\dbetahvcds$ and $\dbetahvgds$, such that
\begin{enumerate}
\item the model selection property of the elastic net is obtained,
\item $\alpha \|\vv\|_1 +(1-\alpha)\|\vv\|_2^2=1$.
\end{enumerate}
\item Within the freedom remaining after step 1, tune $\dbetahvegds$ to, approximately, maximize $\gv^\top\dbetahvegds$.
\end{enumerate}

Combining $\dbetahvgds$ and $\dbetahvcds$ in the same way as was done in Equation \ref{eq:egd}, we define 
\begin{equation}
\label{eq:egds}
\dbetahvegds(t):=\Iegds(\alpha,t)\cdot\gv(t)\left(\frac{\alpha}{\left\|\Iegds(\alpha,t)\cdot\gv(t)\right\|_1}+\frac{1-\alpha}{\left\|\Iegds(\alpha,t)\cdot\gv(t)\right\|_2}\right),
\end{equation}
where $\Iegds$ is a diagonal matrix with zeros and ones on the diagonal, such that $\Iegds(0,t)=\Igd=\bm{I}$ and $\Iegds(1,t)=\Icd(t)$. However, $\Iegds$ is not necessarily identical to $\Iegd$.

We define $p_1(t)\in[1,p]$ to be the number of ones in $\Iegds$, i.e.\ the number of parameters that are updated at time $t$:
\begin{definition}[$p_1$]
\label{dfn:p1}
$$p_1(\alpha,t):=\sum_{d=1}^p\left(\Iegds(\alpha,t)\right)_{dd}.$$
\end{definition}
As optimization proceeds toward convergence, all gradient components approach zero, and thus each other. This means that $p_1$ increases (i.e.\ more parameters are updated), but not necessarily monotonically, toward $p$. However, for $\alpha>0$, some absolute gradient components may oscillate around $\alpha\cdot|g_m|$, being updated in one time step, but not in the next. In that case, we may have $p_1<p$ during the entire training. Also, note that $p_1$ is not explicitly defined; its value is a consequence of Definition \ref{dfn:p1}. 

Now, to obtain $\alpha \left\|\dbetahvegds\right\|_1 +(1-\alpha)\left\|\dbetahvegds\right\|_2^2=1$, $\dbetahvegds$ needs to be scaled, as specified in Proposition \ref{thm:h_alpha_sd}.
\begin{proposition}~\\
\label{thm:h_alpha_sd}
\begin{equation*}
\begin{aligned}
&\text{Let } q_1(t):=\left(\frac{\|\Iegds(\alpha,t)\cdot\gv(t)\|_1}{\|\Iegds(\alpha,t)\cdot\gv(t)\|_2}\right)^2 \textrm{ and let }\\
&c_\alpha(t):=\frac{\sqrt{q_1(t)\cdot(\alpha^2q_1(t)+4(1-\alpha))}-\alpha \cdot q_1(t)}{2(1-\alpha)\left(\sqrt{q_1(t)}\cdot(1-\alpha)+\alpha\right)}.\\
&\text{Then }\alpha \cdot \left\|c_\alpha(t)\dbetahvegds(t)\right\|_1 +(1-\alpha)\cdot \left\|c_\alpha(t)\dbetahvegds(t)\right\|_2^2=1.
\end{aligned}
\end{equation*}
\end{proposition}
$c_\alpha$ depends both on $\alpha$ and the quotient between the $\ell_1$ and $\ell_2$ norms in a quite complicated form. However, according to Proposition \ref{thm:eg_norm_sd}, in the absence of $c_\alpha$ the distance from 1 is still bounded:

\begin{proposition}~\\
\label{thm:eg_norm_sd}
For $\alpha\in [0,1]$, $1\leq p_1\leq p$
\begin{equation*}
\begin{aligned}
0.61&<1-\alpha(1-\alpha)(2-\alpha)\cdot\left(1-\frac 1{p_1(t)}\right)\\
&\leq \alpha\left\|\dbetahvegds(t)\right\|_1 +(1-\alpha)\left\|\dbetahvegds(t)\right\|_2^2\\
&\leq1+\alpha(1-\alpha)\cdot\left(\sqrt{p_1(t)}-1\right)\leq1+\frac{\sqrt{p_1(t)}-1}4.
\end{aligned}
\end{equation*}
\end{proposition}

What remains to do, is to select $\Iegds$ to maximize $c_\alpha\gv^\top\dbetahvegds$. Since $\gv^\top\Iegds\gv=\gv^\top\Iegds\Iegds\gv=\left\|\Iegds\gv\right\|_2^2$, and since $c_\alpha\geq 0$ for $q_1\geq 0$ and $\alpha \in[0,1]$, maximizing $c_\alpha\gv^\top \dbetahvegds$ amounts to maximizing
\begin{equation*}
\label{eq:max_sd}
\gv(t)^\top\dbetahvegds(t)=\alpha\frac{\left\|\Iegds(\alpha,t)\cdot\gv(t)\right\|_2^2}{\left\|\Iegds(\alpha,t)\cdot\gv(t)\right\|_1}+(1-\alpha)\left\|\Iegds(\alpha,t)\cdot\gv(t)\right\|_2.
\end{equation*}
The second term trivially increases with $p_1$ while, according to Lemma \ref{thm:decr_p1}, the first term decreases with $p_1$.
\begin{lemma}~\\
\label{thm:decr_p1}
$$\frac{\|\Iegds\cdot\gv\|_2^2}{\|\Iegds\cdot\gv\|_1}\textrm{ is a decreasing function in } p_1.$$
\end{lemma}
The exact trade-off between the two terms depends on the gradient at the specific time step and has no general solution. However, when $\alpha$ is large we want $\left\|\Iegds\right\|_2^2/\left\|\Iegds\right\|_1$ to be large, i.e. we want $p_1$ to be small. When $\alpha$ is small, we want $\left\|\Iegds\right\|_2$ to be large, i.e. we want $p_1$ to be large. This desire is consistent with how we defined $\Iegd$ in Definition \ref{dfn:iegd}, and we thus define $\Iegds:=\Iegd$. This means that we use the same gradient selection matrix as in the original formulation of elastic gradient descent.

\subsection{The General Stagewise Procedure}
The general stagewise procedure \citep{tibshirani2015general} is formulated similarly to steepest descent, but while the purpose of steepest descent just is to find the optimal solution, in the general stagewise procedure, the entire solution path is of interest. Here, the norm in the constraint is replaced by any convex function, $h$, and the optimization step size, $\Delta t$, is incorporated into $\dbetahv$:
\begin{equation*}
\label{eq:general_stagewise}
\begin{aligned}
\dbetahvg(t)&=\argmax_{h(\vv)\leq\Delta t}\gv(t)^\top\vv\\
\betahatv(t+\Delta t)&=\betahatv(t)-\dbetahvg(t).
\end{aligned}
\end{equation*}

In this framework we obtain
\begin{equation*}
\begin{aligned}
&\dbetahvcdg(t)=\Delta t\cdot \frac{\Icd(t)\cdot\gv(t)}{\|\Icd(t)\cdot\gv(t)\|_1}=\Delta t\cdot\dbetahvcd(t)\\
&\dbetahvgdg(t)=\sqrt{\Delta t}\cdot\frac{\gv(t)}{\|\gv(t)\|_2}=\sqrt{\Delta t}\cdot\frac{\Igd\cdot\gv(t)}{\|\Igd\cdot\gv(t)\|_2}=\sqrt{\Delta t}\cdot\dbetahvgds(t),
\end{aligned}
\end{equation*}
which suggests
\begin{equation}
\label{eq:egdg}
\dbetahvegdg(t):=\Iegd(\alpha,t)\cdot\gv(t)\left(\frac{\alpha\Delta t}{\left\|\Iegd(\alpha,t)\cdot\gv(t)\right\|_1}+\frac{(1-\alpha)\sqrt{\Delta t}}{\left\|\Iegd(\alpha,t)\cdot\gv(t)\right\|_2}\right).
\end{equation}
The analogs of Propositions \ref{thm:h_alpha_sd} and \ref{thm:eg_norm_sd} in this framework are presented in Propositions \ref{thm:h_alpha_gs} and \ref{thm:eg_norm_gs}.

\begin{proposition}~\\
\label{thm:h_alpha_gs}
\begin{equation*}
\begin{aligned}
&\text{Let } \dbetahvegdgc(t) :=\Iegd(\alpha,t)\cdot\gv(t)\left(\frac{\alpha\cdot c_{\alpha,\Delta t}(t)\cdot\Delta t}{\left\|\Iegd(\alpha,t)\cdot\gv(t)\right\|_1}+\frac{(1-\alpha)\sqrt{c_{\alpha,\Delta t}(t)\cdot\Delta t}}{\left\|\Iegd(\alpha)\cdot\gv(t)\right\|_2}\right)\\
&\text{for } c_{\alpha,\Delta t}(t):=\\
&\left(\frac{\sqrt{2\alpha\sqrt{q_1\cdot(\alpha^2q_1 + 4\Delta t(1-\alpha))} + q_1\cdot((1-\alpha)^3-2\alpha^2)} - (1-\alpha)\sqrt{q_1\cdot(1-\alpha)}}{\alpha\sqrt{4\Delta t(1-\alpha)}}\right)^2,\\
&\text{where } q_1=q_1(t):=\left(\frac{\|\Iegd(\alpha,t)\cdot\gv(t)\|_1}{\|\Iegd(\alpha,t)\cdot\gv(t)\|_2}\right)^2.\\
&\text{Then }\alpha \cdot\left\|\dbetahvegdgc(t)\right\|_1 +(1-\alpha)\cdot\left\|\dbetahvegdgc(t)\right\|_2^2=\Delta t.
\end{aligned}
\end{equation*}
\end{proposition}

\begin{proposition}~\\
\label{thm:eg_norm_gs}
For $\alpha\in [0,1]$, $1\leq p_1\leq p$
\begin{equation*}
\begin{aligned}
0.61\cdot\Delta t&<\Delta t\left(1-\alpha(1-\alpha)(2-\alpha)\cdot\left(1-\frac {\Delta t}{p_1(t)}\right)\right)\\
&\leq \alpha \left\|\dbetahvegdg(t)\right\|_1 +(1-\alpha)\left\|\dbetahvegdg(t)\right\|_2^2\\
&\leq\Delta t\left(1+\alpha(1-\alpha)\cdot\left(\sqrt{\frac{p_1(t)}{\Delta t}}-1\right)\right)\leq\Delta t\left(1+\frac{\sqrt{p_1(t)/\Delta t}-1}4\right).
\end{aligned}
\end{equation*}
\end{proposition}

\subsection{Comparing the Formulations}
Compared to the original formulation of elastic gradient descent in Equation \ref{eq:egd}, both the steepest descent and general stagewise formulations differ in the normalization of the second term. Apart from that, the general stagewise formulation uses different step sizes for the coordinate and gradient descent contributions, where the difference grows with smaller $\Delta t$ (assuming $\Delta t<1$).

We also note that compared to Proposition \ref{thm:eg_norm_sd}, in Proposition \ref{thm:eg_norm_gs}, $p_1$ is replaced by $p_1/\Delta t$. This means that if $\Delta t$ is small, $\left\|\dbetahvegdg\right\|_1 +(1-\alpha)\left\|\dbetahvegdg\right\|_2^2$ might deviate quite much from $\Delta t$.

According to our empirical experience, however, all flavors of elastic gradient descent, i.e. Equation \ref{eq:egd}, and Equations \ref{eq:egds} and \ref{eq:egdg} with and without scaling, provide virtually identical solution paths. 
In Figure \ref{fig:diab_path_flavs}, we compare the solution paths for the diabetes data for the different flavors of elastic gradient descent with $\alpha=0.5$. The paths, displayed in the first column, are hardly, if at all, distinguishable. The second column shows the normalized gradients. Just as for the solution paths, the gradients evolve very similarly between the five versions, even though some differences are visible.

The third column shows how $\alpha\left\|\dbetahv\right\|_1+(1-\alpha)\left\|\dbetahv\right\|_2^2=: h_\alpha(\dbetahv)$ deviates from 1 (or $\Delta t$), together with the two bounds provided by Propositions \ref{thm:eg_norm_sd} and \ref{thm:eg_norm_gs} and with $p_1$, which is displayed on a different y-scale. It can be seen how $h_\alpha$ is exactly 1 ($\Delta t$) in the scaled case and how it stays within the bounds in the unscaled case. As expected, $h_\alpha$ deviates more from $\Delta t$ in the general stagewise framework than it does from 1 in the steepest descent framework. For standard gradient descent, $h_\alpha$ deviates a lot from 1. It can also be noted how the effective value of $p_1$ is always strictly lower than $p=10$, which can be attributed to the oscillations around $\pm\alpha\cdot|g_m|$, i.e.\ the coupled parameters are included at some time steps and excluded at other.

\begin{figure}
  \center
  \includegraphics[width=1.\textwidth]{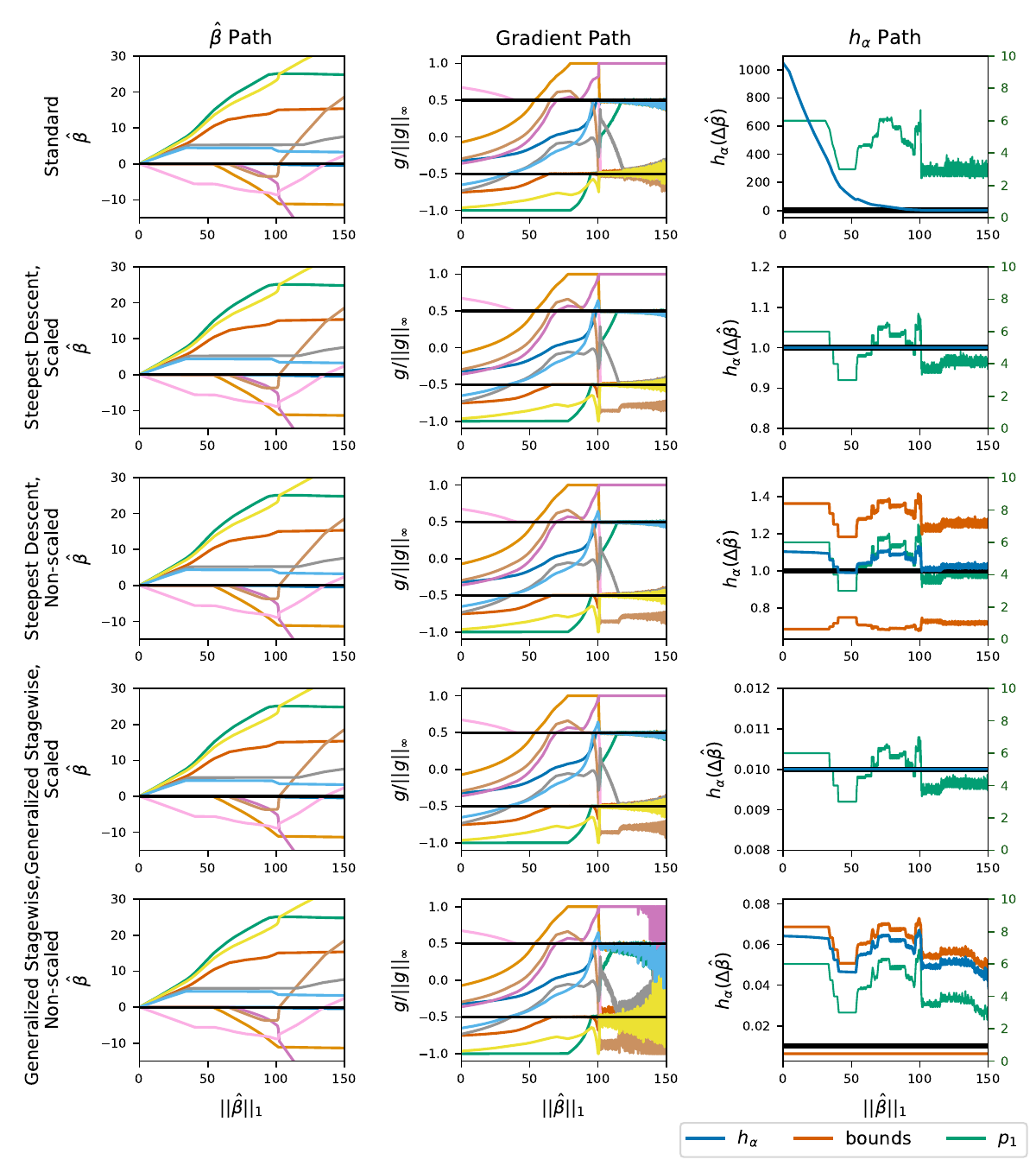}
  \caption{Solution paths for elastic gradient flow and the different flavors of elastic gradient descent with step size 0.01, without momentum, at $\alpha=0.5$ on the diabetes data. The second column shows the normed gradients, with black lines at $\pm \alpha$. The third column shows how $h_\alpha(\dbetahv)$ deviates from 1 (or $\Delta t$), together with bounds from Propositions \ref{thm:eg_norm_sd} and \ref{thm:eg_norm_gs}. On the right y-axis, $p_1$ is plotted. To increase readability a moving average with width 9 was applied to the graphs in the third column.}
  \label{fig:diab_path_flavs}
\end{figure}

\section{Additional Experiments}
\label{sec:more_exps}
In this section, we provide additional experiments.
\subsection{Illustration of Momentum Induced Sparsity}
In Figure \ref{fig:path_demo_momentum}, we compare the solution paths of elastic gradient descent with and without momentum, for a step size of 0.01. The inertia introduced by the momentum causes $\beta_2$ to be zero for a longer time, (it is included into the model at a smaller reconstruction error) than when no momentum is used.
\begin{figure}
  \center
  \includegraphics[width=.5\textwidth]{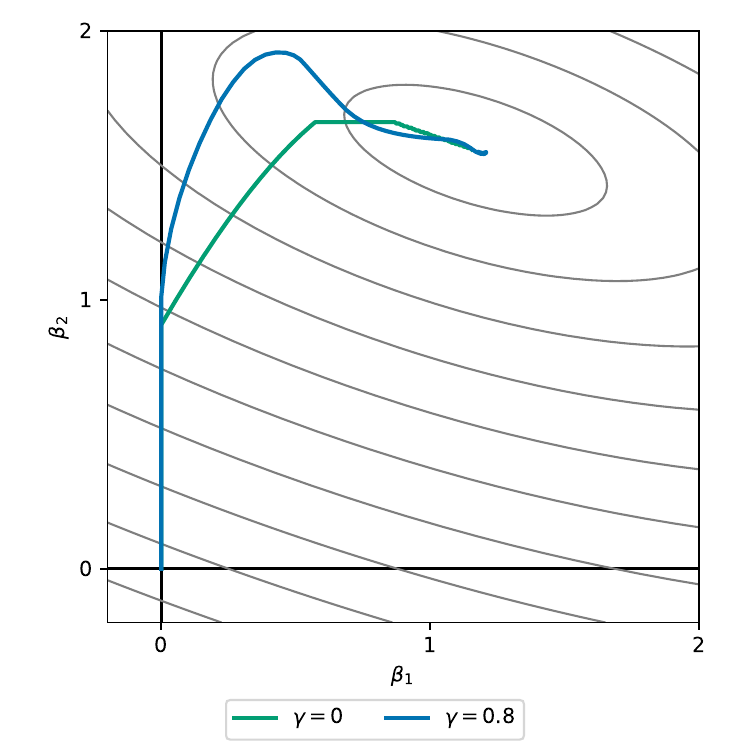}
  \caption{Solution paths of elastic gradient descent with and without momentum. With momentum, $\beta_2$ is included into the model later, i.e.\ for a smaller reconstruction error, than without momentum.}
  \label{fig:path_demo_momentum}
\end{figure}

\subsection{Synthetic Data for Model Selection}
In this section, we extend the simulation of Section \ref{sec:synth_exp}. 

To further investigate the specificity of the algorithms, we changed the experiment setup, so that we always used 40 non-zero parameters, with the remaining $p-40$ parameters being zero, keeping all other aspects the same. The results are presented in Figure \ref{fig:p_sweep2}. This time, momentum no longer provides the same advantage in specificity as before. 
This can probably be attributed to the fact that in the first case, when the number of non-zero parameters grows, so does the value of the maximum initial gradient, which makes it harder for a zero parameter to be erroneously included, especially with momentum when early stages of training matter more, and thus leads to a better specificity. On the other hand, in this second case, when the number of non-zero parameters, and thus the maximum initial gradient, is constant while the number of parameters to possibly erroneously include increases, this advantage of momentum is less prominent.

\begin{figure}
  \center
  \includegraphics[width=1.\textwidth]{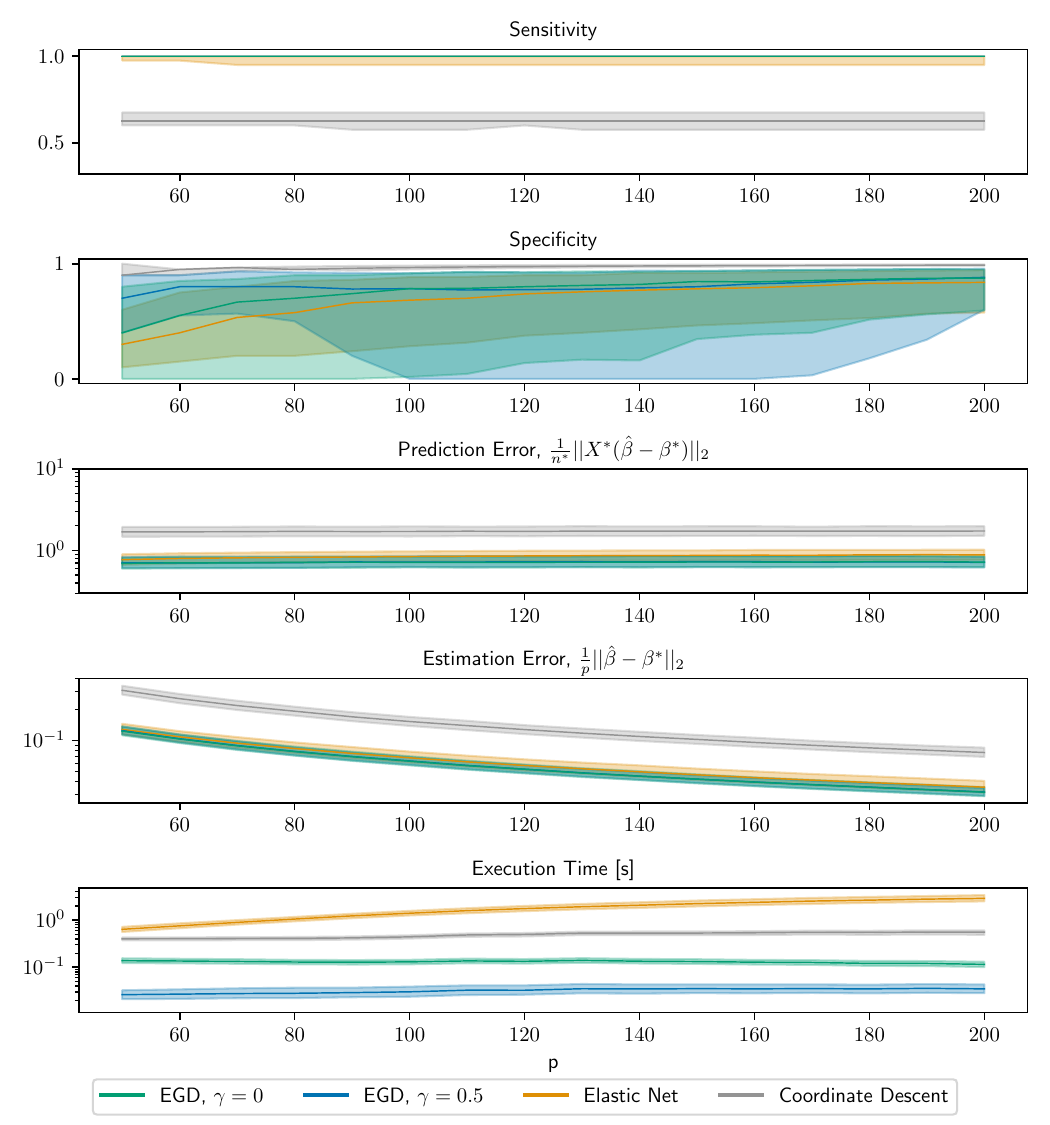}
  \caption{Median and first and third quartiles for the sensitivity, specificity, prediction and estimation errors, and execution time in seconds. Compared to in Figure \ref{fig:p_sweep}, momentum no longer provides the same advantage in specificity as before. The signal-to-noise ratio is always between $46.4$ and $46.9$.}
  \label{fig:p_sweep2}
\end{figure}

We also extended the experiments of Section \ref{sec:synth_exp} by using the same setup as in the original experiments, except using $\rho_1\in[0.5,\ 0.6,\ 0.7,\ 0.8,\ 0.9,]$, $\rho_2\in[0.0,\ 0.1,\ 0.2,\ 0.3,$ $0.4,\ 0.5]$ and $p\in[50,\ 100,\ 200]$.

The results are presented in Figures \ref{fig:rho_sweep_50}, \ref{fig:rho_sweep_100} and \ref{fig:rho_sweep_200}. The signal-to-noise ratios of the problems, which were essentially constant with respect to $\rho_2$, are stated for the different values of $\rho_1$.

The conclusions are consistent with those from Section \ref{sec:synth_exp}: Elastic gradient descent and the elastic net perform similarly in all aspects except computational time, where elastic gradient descent performs significantly faster. The computational performance of elastic gradient descent improves with momentum. For high-dimensional data ($p>n=100$), where no unique solution exists, momentum also greatly improves the model specificity. 
Compared to coordinate descent, elastic gradient descent performs better in all aspects except for specificity. The higher specificity of coordinate descent, however, comes at the cost of much worse sensitivity, and prediction and estimation errors.

\begin{figure}
  \center
  \includegraphics[width=1.\textwidth]{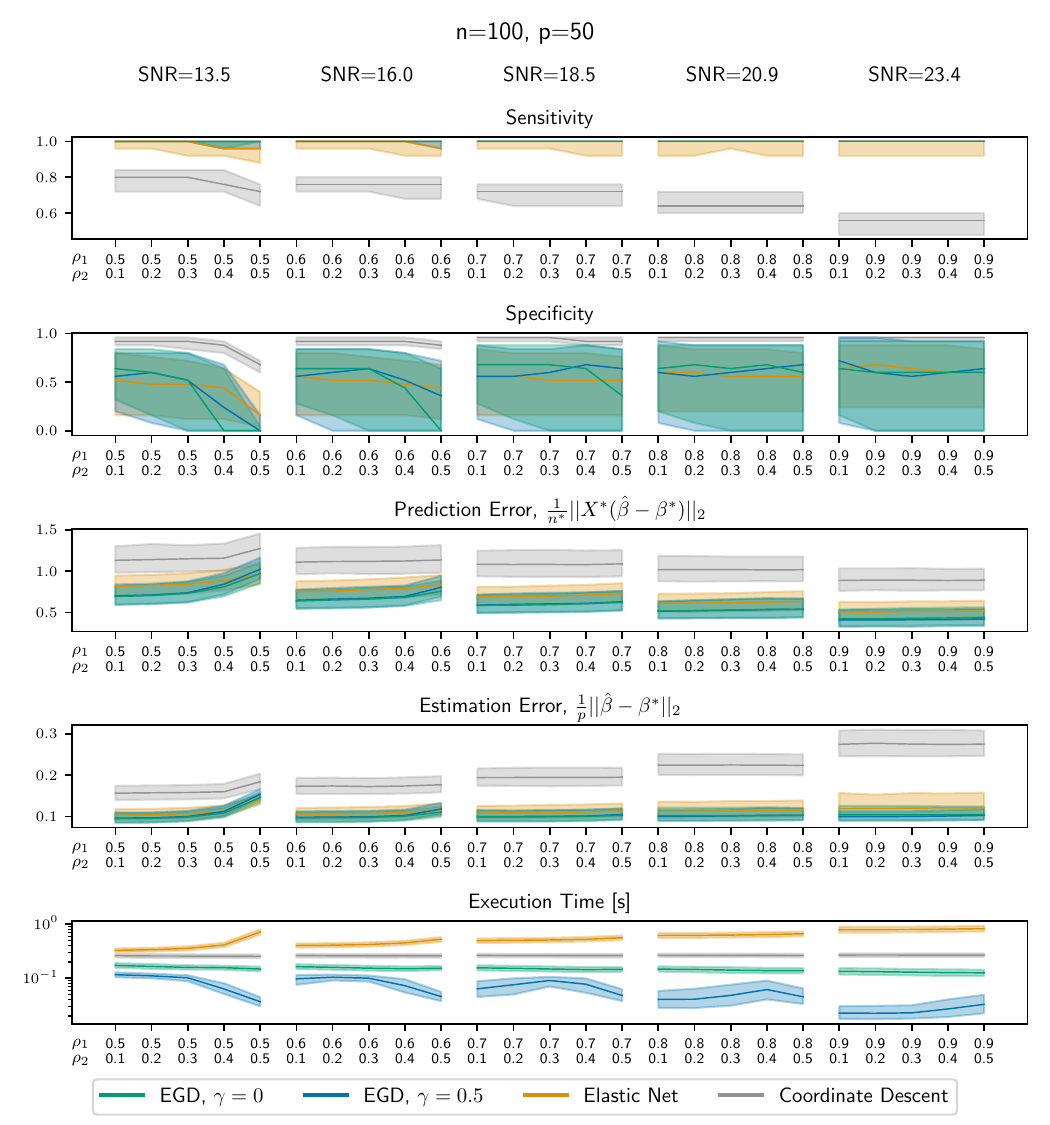}
  \caption{Median and first and third quartiles for the sensitivity, specificity, prediction and estimation errors, and execution time in seconds, in the low-dimensional case. The value of $\rho_1$ is constant within each of the five panels in each row, while $\rho_2$ varies. The signal-to-noise ratio of the problem, which is essentially constant with respect to $\rho_2$, is stated for the different values of $\rho_1$. Compared to the elastic net, elastic gradient descent performs similarly, except for execution time, where it is much faster. Compared to coordinate descent, elastic gradient descent performs better in all aspects except specificity. Elastic gradient descent performs faster with momentum than without.}
  \label{fig:rho_sweep_50}
\end{figure}

\begin{figure}
  \center
  \includegraphics[width=1.\textwidth]{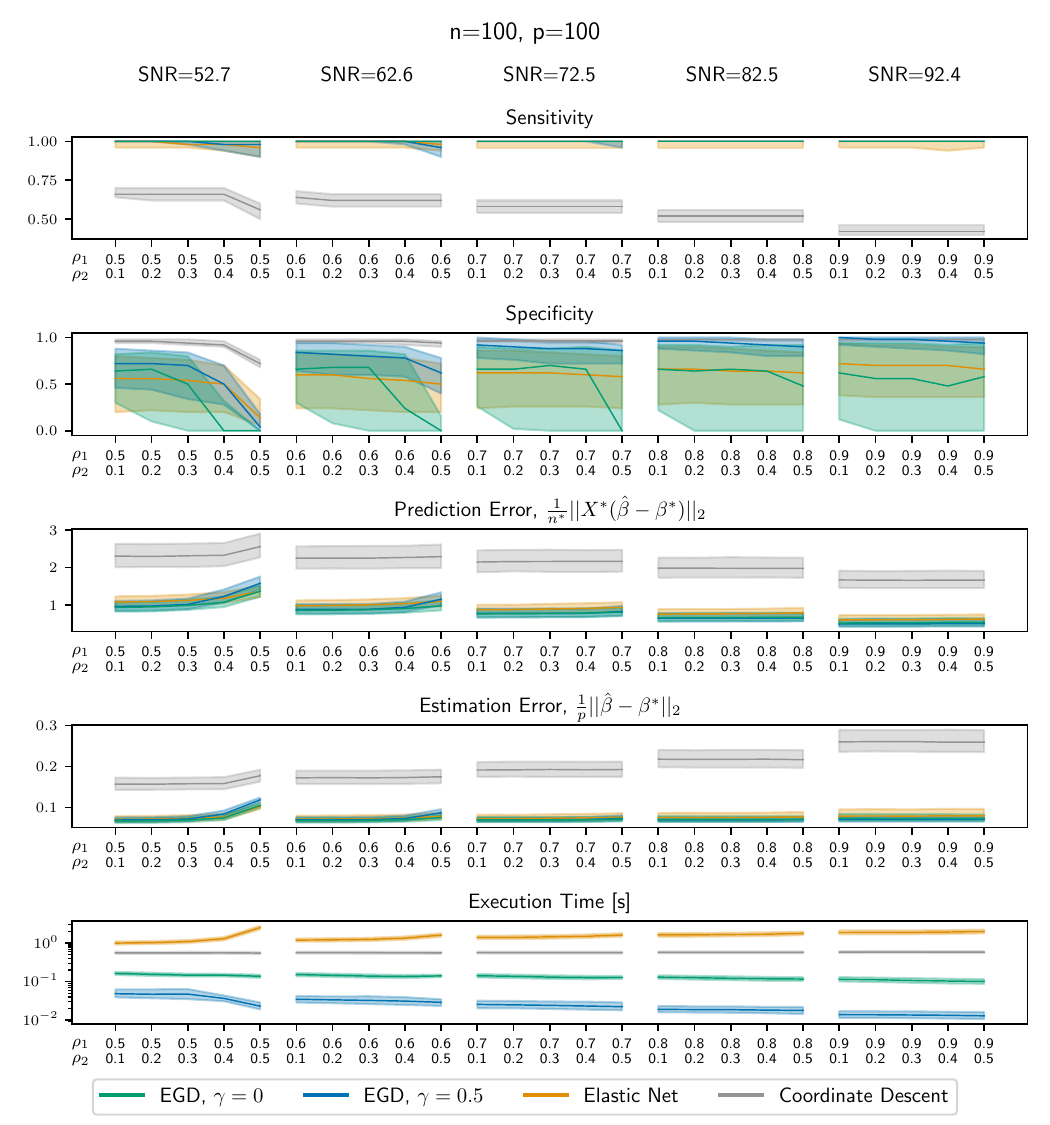}
  \caption{Median and first and third quartiles for the sensitivity, specificity, prediction and estimation errors, and execution time in seconds, when $n=p$. The value of $\rho_1$ is constant within each of the five panels in each row, while $\rho_2$ varies. The signal-to-noise ratio of the problem, which is essentially constant with respect to $\rho_2$, is stated for the different values of $\rho_1$. Compared to the elastic net, elastic gradient descent performs similarly, except for execution time, where it is much faster. Compared to coordinate descent, elastic gradient descent performs better in all aspects except specificity. Elastic gradient descent performs faster, and has better specificity, with momentum than without.}
  \label{fig:rho_sweep_100}
\end{figure}

\begin{figure}
  \center
  \includegraphics[width=1.\textwidth]{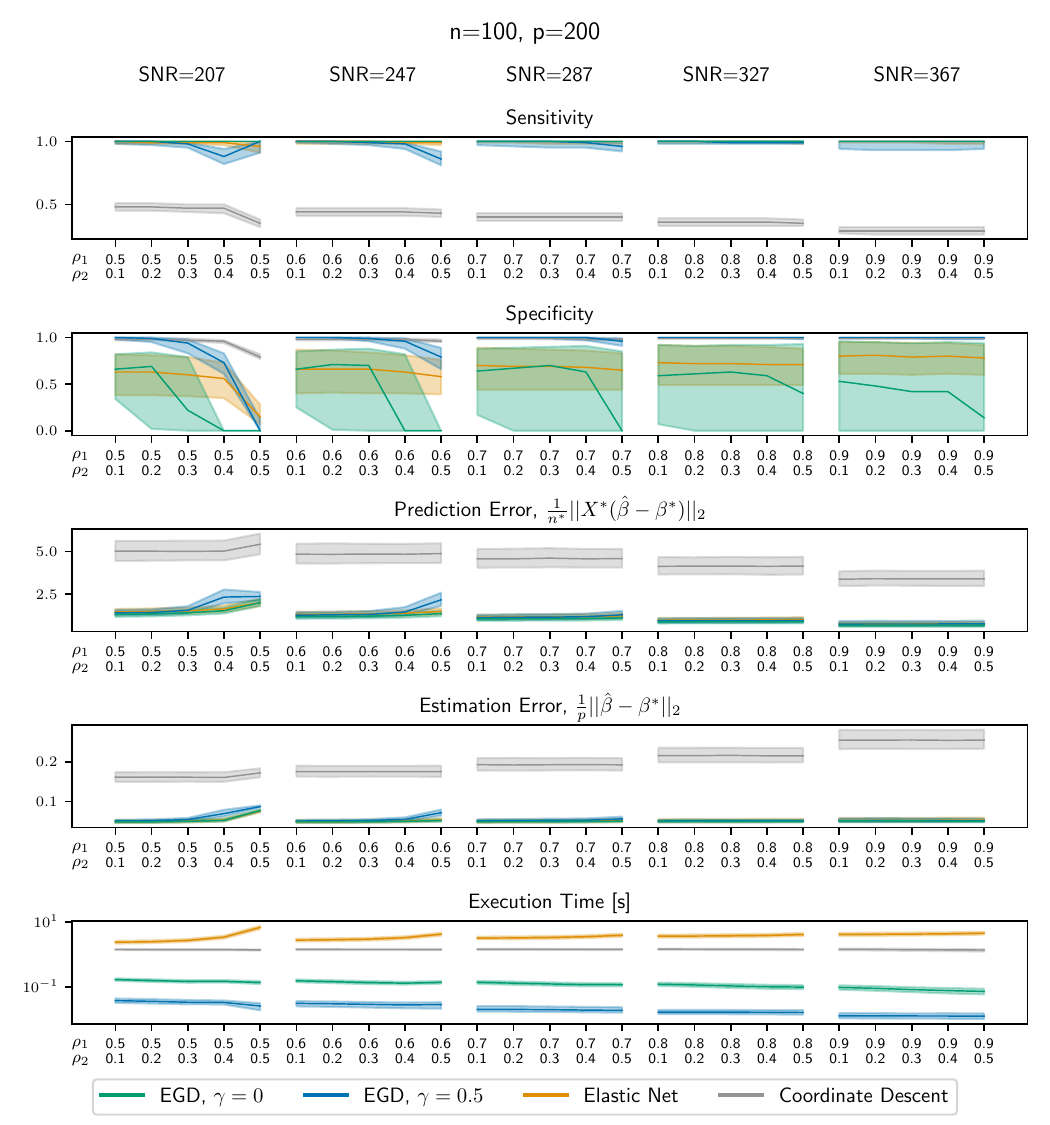}
  \caption{Median and first and third quartiles for the sensitivity, specificity, prediction and estimation errors, and execution time in seconds, in the high-dimensional case. The value of $\rho_1$ is constant within each of the five panels in each row, while $\rho_2$ varies. The signal-to-noise ratio of the problem, which is essentially constant with respect to $\rho_2$, is stated for the different values of $\rho_1$. Compared to the elastic net, elastic gradient descent performs similarly, or better, except for execution time, where it is much faster. Compared to coordinate descent, elastic gradient descent performs better in all aspects except specificity. Elastic gradient descent performs faster, and has better specificity, with momentum than without.}
  \label{fig:rho_sweep_200}
\end{figure}

\section{Elastic Gradient Flow}
\label{sec:flow}
In this section, we investigate the limits as the step size, $\Delta t$, goes to zero when solving linear least squares with gradient, coordinate, and elastic gradient descent with momentum. 
Gradient descent with infinitesimal step size is known as gradient flow, and analogously we use the terms coordinate flow and elastic gradient flow. 
We start by reviewing gradient flow and then consider coordinate flow and elastic gradient flow, where the latter generalizes both gradient and coordinate flow. Since it is not obvious that the limits of $\Icd(t)$ and $\Iegd(t)$ exist as $\Delta t\to 0$, coordinate and elastic gradient flow are presented as well-motivated definitions rather than as theorems, with motivations in Section \ref{constr:flows}. 

In the following, subscript with respect to a set denotes the sub-matrix (or sub-vector) specified by the indices in the set, $(\cdot)^{(k)}$ denotes the time derivative of order $k$, and $\odot$ denotes element-wise multiplication.

For linear least squares, the gradient at time $t$ is
\begin{equation}
\label{eq:grad}
\gv(t):=\nabla_{\betahv(t)}\left(\frac1{2n}\|\yv-\Xm\betahv(t)\|_2^2\right)=\frac 1n\Xm^\top(\Xm\betahv(t)-\yv)=-\Ss(\betahatolsv-\betahv(t)),
\end{equation}
where $\Ss:=\frac 1n \Xm^\top \Xm$ is the empirical covariance matrix and $\betahatolsv:=(\Xm^\top\Xm)^+\Xm^\top\yv$ (where $(\cdot)^+$ denotes the Moore-Penrose pseudoinverse) is the minimum norm ordinary least squares solution for $t=\infty$. The last equality in Equation \ref{eq:grad} follows from $\Xm^\top=(\Xm^\top\Xm)(\Xm^\top\Xm)^{+}\Xm^\top$.

\subsection{Gradient Flow}
\label{sec:grad_flow}
When linear regression is solved using gradient descent with momentum, $\betahv$ is updated iteratively according to
\begin{equation}
\label{eq:grad_desc}
\begin{aligned}
\betahv(t+\Delta t)=&\betahv(t)+\gamma\left(\betahv(t)-\betahv(t-\Delta t)\right)-\Delta t\cdot \gv(t)\\
=&\betahv(t)+\gamma\left(\betahv(t)-\betahv(t-\Delta t)\right)+\Delta t\cdot \Ss(\betahatolsv-\betahv(t)).
\end{aligned}
\end{equation}
Moving all but the last term to the left-hand side, dividing by $\Delta t$ and then letting $\Delta t\to 0$ results in the differential equation
$$(1-\gamma)\cdot\frac{\partial\betahv(t)}{\partial t}=\Ss(\betahatolsv-\betahv(t)),$$
which has the solution
$$\betahv(t)=\betahatolsv-\exp\left(-\frac t{1-\gamma}\Ss\right)(\betahatolsv-\betahnollv),$$
where $\exp$ denotes the matrix exponential.

For $\betahnollv=\bm{0}$ the gradient flow estimate becomes
\begin{equation}
\label{eq:beta_gf}
\betahvgf(t)=\left(\Imm-\exp\left(-\frac t{1-\gamma}\Ss\right)\right)\betahatolsv.
\end{equation}

\subsection{Coordinate Flow}
\label{sec:coord_flow}
When linear regression is solved using coordinate descent with momentum, $\betahv$ is updated iteratively according to
\begin{equation}
\label{eq:coord_desc}
\begin{aligned}
\betahv(t+\Delta t)=&\betahv(t)+\gamma\left(\betahv(t)-\betahv(t-\Delta t)\right)-\Delta t \cdot \Icd(t)\cdot\sgn(\gv(t))\\
=&\betahv(t)+\gamma\left(\betahv(t)-\betahv(t-\Delta t)\right)+\Delta t \cdot\Icd(t)\cdot\sgn(\Ss(\betahatolsv-\betahv(t)).
\end{aligned}
\end{equation}
When the magnitudes of two or more gradient components are all close to the maximum gradient value, the parameter to update changes in almost every time step, that is the 1 in $\Icd$ changes position in almost every time step. As $\Delta t \to 0$, we would like to replace $\Icd$ with a coordinate flow version, $\Icf$. In contrast to $\Icd$, where only one diagonal element is 1 and the rest are 0, for $\Icf$ we allow multiple diagonal elements to be non-zero, as long as the sum of the diagonal elements is 1. The only 1 in $\Icd$ is now distributed along the diagonal of $\Icf$, with $(\Icf)_{dd}>0$ if and only if $d$ belongs to the active set, i.e.\ the set of indices between which the 1 in $\Icd$ alters.
This leads us to define coordinate flow according to Definition \ref{dfn:coord_flow}, where details behind the definition are presented in Section \ref{constr:coord_flow}.

The coordinate flow estimate, $\betahvcf(t)$, changes linearly in time with a slope controlled by the piece-wise constant matrix $\Icf$. At certain times, $t_i$, $\Icf$ is updated to a new constant matrix. These times correspond to changes in the active set, $S_A$, which specifies which parameters are updated, namely the ones with non-zero slopes.
Since $\betahv$ being linear in $t$ implies that also $\|\betahv\|_1$ is linear in $t$, coordinate flow can be thought of as a dual formulation of the forward stagewise version LARS algorithm, providing $\betahv$ as a function of $t$ rather than of $\|\betahv\|_1$.

\begin{definition}[Coordinate Flow]~
\label{dfn:coord_flow}
\begin{itemize}
\item $\betahvcf(0)=\bm{0}$, $t_0=0$, $t_\imax$ is the time of convergence, i.e.\ $\betahvcf(t_\imax)=\betahatolsv$.
\item For $0<i<\imax$,
\begin{equation}
\label{eq:beta_cf}
\betahvcf(t)=\betahvcf(t_i)+\frac{t-t_i}{1-\gamma}\Icfi\cdot\sgn\left(\Ss(\betahatolsv-\betahvcf(t_i))\right),\quad t\in [t_i,t_{i+1}),
\end{equation}
where 
$\{\Icfi\}_{i=0}^\imax$ are constant diagonal matrices, with non-zero diagonal components given by
$$\left(\Icfi\right)_{S^i_A,S^i_A}=\textrm{diag}\left((\bm{B}_{:,1})^{-1}\odot\sgn\left(\Ss(\betahatolsv-\betahvcf(t_i))\right)_{S^i_A}\right),$$ 
where $\bm{B}$ is a square matrix, stated in the construction (Section \ref{constr:coord_flow}), that depends on $\sgn\left(\Ss(\betahatolsv-\betahvcf(t_i))\right)_{S^i_A}$ and $\Ss_{S^i_A,S^i_A}$.
\item $S^i_A:=\left\{d: \left(\Icfi\right)_{dd}>0\right\}$.
\item
The times $\{t_i\}_{i=0}^\imax$, when $\Icfi$ is updated are given by
\begin{equation*}
\begin{aligned}
&t_0=0\\
&\Delta t_{i,d,1}=\frac{\left(\Ss_{d,:}-\Ss_{m_1,:}\right)(\betahatolsv-\betahvcf(t_i))}{\left(\Ss_{d,:}-\Ss_{m_1,:}\right)\Icfi\cdot\sgn\left(\Ss(\betahatolsv-\betahvcf(t_i))\right)}\\
&\Delta t_{i,d,2}=\frac{\left(\Ss_{d,:}+\Ss_{m_1,:}\right)(\betahatolsv-\betahvcf(t_i))}{\left(\Ss_{d,:}+\Ss_{m_1,:}\right)\Icfi\cdot\sgn\left(\Ss(\betahatolsv-\betahvcf(t_i))\right)}\\
&t_{i+1}=t_i+\min_{\substack{d\notin S^i_A,\ k=1,2\\ \Delta t_{i,d,k}>0}}\Delta t_{i,d,k}.
\end{aligned}
\end{equation*}
\end{itemize}
\end{definition}

\subsection{Elastic Gradient Flow}
\label{sec:el_flow}
When linear regression is solved using elastic gradient descent, $\betahv$ is updated iteratively according to
\begin{equation}
\label{eq:en_desc}
\begin{aligned}
\betahv(t+\Delta t)=&\betahv(t)+\gamma\left(\betahv(t)-\betahv(t-\Delta t)\right)- \Delta t \cdot\Iegd(\alpha,t)\left(\alpha\cdot\sgn(\gv(t))+(1-\alpha)\cdot\gv(t)\right)\\
=&\betahv(t)+\gamma\left(\betahv(t)-\betahv(t-\Delta t)\right)\\
&+ \Delta t \cdot\Iegd(\alpha,t)\left(\alpha\cdot\sgn\left(\Ss(\betahatolsv-\betahv(t))\right)+(1-\alpha)\cdot\Ss(\betahatolsv-\betahv(t))\right).
\end{aligned}
\end{equation}
Compared to coordinate descent, where only the parameter with maximum gradient is updated, this time the parameters with large enough, but not necessarily the largest, gradients are updated. Just as for coordinate descent and flow, we replace $\Iegd(\alpha,t)$ with $\Iegf(\alpha,t)$. Again, $\Iegf(\alpha,t)$ is a diagonal matrix with $\left(\Iegf(\alpha,t)\right)_{dd}\in [0,1]$, but in contrast to $\Icf(t)$ it is not piece-wise constant.

We define elastic gradient flow according to Definition \ref{dfn:en_flow}, where details behind the definition are presented in Section \ref{constr:en_flow}. The active and inactive sets of coordinate flow are now generalized to the free, coupled, and inactive sets for elastic gradient flow.
Similar to coordinate flow, $\Iegf$ is recalculated at certain times, $t_i$. However, between the times of recalculation only the entries corresponding to the free and inactive sets, $(\Iegf)_{S_F\cup S_0,S_F\cup S_0}$, remain constant, while the entries corresponding to the coupled set, $(\Iegf)_{S_C,S_C}$, change with time on the interval $(0,1)$.
\begin{definition}[Elastic Gradient Flow]~
\label{dfn:en_flow}
~
\begin{itemize}
\item $\betahvegf(0)=\bm{0}$, $t_0=0$, $t_\imax$ is the time of convergence, i.e.\ $\betahvegf(t_\imax)=\betahatolsv$.
\item For $0<i<\imax$,
\begin{equation}
\begin{aligned}
\label{eq:beta_ef}
\betahvegf(t)=& \betahvegf(t_i)+ \left((1-\alpha)\Ss\right)^{-1}\left(\Imm-\exp\left(\Omi(t_i,t)\right)\right)\\
&\cdot\left(\alpha\cdot\sgn\left(\Ss(\betahatolsv-\betahvegf(t_i))\right)+(1-\alpha)\cdot\Ss\left(\betahatolsv-\betahvegf(t_i)\right)\right),\\
&t\in [t_i,t_{i+1}),
\end{aligned}
\end{equation}
where 
\begin{itemize}
\item 
$\Omi(t_i,t)$ is the Magnus expansion \citep{magnus1954exponential} of $-\frac{1-\alpha}{1-\gamma}\Ss\Iegfi(\alpha,t)$, 
\item
$\left\{\Iegfi(\alpha,t)\right\}_{i=0}^\imax$ are diagonal matrices with elements in $[0,1]$, such that
\begin{itemize}
\item
$\left(\Iegfi\right)_{S^i_F,S^i_F}(\alpha,t)=\Imm$.
\item
$\left(\Iegfi\right)_{S^i_0,S^i_0}(\alpha,t)=\bm{0}$.
\item
$\left(\Iegfi\right)_{S^i_C,S^i_C}(\alpha,t)$ is defined through its Taylor expansion:
$$\left(\Iegfi\right)_{S^i_C,S^i_C}(\alpha,t):=\sum_{k=0}^\infty \left(\left(\Iegfi\right)_{S^i_C,S^i_C}\right)^{(k)}(t_i)\frac{(t-t_i)^k}{k!},$$
where
\begin{equation}
\begin{aligned}
\label{eq:en_taylor}
\left(\left(\Iegfi\right)_{S^i_C,S^i_C}\right)^{(k)}(t_i)&=\text{\emph{diag}}\left(\bm{A}^{-1}\bm{b}(k)\odot\frac 1{\alpha\cdot\sgn(\gv(t_i))+(1-\alpha)\cdot\gv(t_i)}\right),\\
k&=0,1,\dots,
\end{aligned}
\end{equation}
where matrix $\bm{A}$ and vectors $\bm{b}(k)$, both stated in the construction (Section \ref{constr:en_flow}), depend on $\gv(t_i)$, $\Ss$ and $\left(\Iegfi\right)^{(l)}(\alpha,t_i)$, $l<k$.
\end{itemize}
\end{itemize}
\item
{\makebox[5.7cm]{$S^i_F(t):=\{d: \left(\Iegf\right)_{dd}(t)=1\}$.\hfill}{The free set.}}
\item
{\makebox[5.7cm]{$S^i_0(t):=\{d: \left(\Iegf\right)_{dd}(t)=0\}$.\hfill}{The inactive set.}}
\item
{\makebox[5.7cm]{$S^i_C(t):=\{d: \left(\Iegf\right)_{dd}(t)\in(0,1)\}$.\hfill}{The coupled set.}}
\end{itemize}
\end{definition}
\noindent\textbf{Remark 1:} Note that Equation \ref{eq:beta_ef} is defined even when $\Ss$ is not invertible. By Taylor expanding $\Ss^{-1}\left(\Imm-\exp\left(\Omi(t_i,t)\right)\right)$, it can be seen that $\Ss^{-1}$ is canceled by (at least) one $\Ss$ from $\left(\Imm-\exp\left(\Omi(t_i,t)\right)\right)$.\\
\textbf{Remark 2:} Since implementing elastic gradient flow is quite computationally heavy, the formulation should be seen as a tool for providing a deeper understanding of elastic gradient descent, rather than as a substitute.\\
\textbf{Remark 3:} For coordinate flow, $\left\{\Icf(\alpha,t)\right\}_{i=0}^{i_{\textrm{max}}}$ and $\{t_i\}_{i=1}^\imax$ could be calculated analytically. Due to the exponential function in Equation \ref{eq:beta_ef},
this is not the case for elastic gradient flow. However, $\Iegf(\alpha,t)$ can be expressed by its Taylor expansion of arbitrary order, using the derivatives from Equation \ref{eq:en_taylor}. The second order expansion of $\Omi(t_i,t)$ is presented in Section \ref{sec:magnus}.

To calculate $\{t_i\}_{i=1}^{i_{\textrm{max}}}$, the following criteria have to be evaluated numerically, selecting $t_i$ as the one that occurs first.
\begin{enumerate}
\item
{\makebox[6.1cm]{$|g_d(t)|=\alpha\cdot|g_m(t)|$ for $d\in S_0$.\hfill}{A parameter leaves the inactive set.}}
\item
{\makebox[6.1cm]{$|g_d(t)|=\alpha\cdot|g_m(t)|$ for $d\in S_F$.\hfill}{A parameter leaves the free set.}}
\item
{\makebox[6.1cm]{$\left(\Iegf\right)_{dd}(\alpha,t)\in \{0,1\}$ for $d\in S_C$.\hfill}{A parameter leaves the coupled set.}}
\item
{\makebox[6.1cm]{$|g_d(t)|=|g_m(t)|$ for $d\in S_F$, $d\neq m$.\hfill}{The maximum gradient component changes.}}
\end{enumerate}
\textbf{Remark 4:} Similar to how lasso and ridge regression are special cases of the elastic net, coordinate and gradient descent (flow) are special cases of elastic gradient descent (flow). Remembering that $\Iegd(0,t)=\Imm$, it is trivial that Equations \ref{eq:grad_desc} and \ref{eq:coord_desc} are special cases of Equation \ref{eq:en_desc}. The flow versions require slightly more work: When $\alpha=0$, all variables belong to the free set at all times, which means that $\Iegf(0,t)=\Imm$ and there are no update times $t_i$. Furthermore, since $\Ss\Iegf(0,t)=\Ss$ is independent of $t$, $\exp(-\Om(0,t))$ reduces to $\exp\left(-\frac1{1-\gamma}\int_0^t\Ss dt\right)=\exp\left(-\frac t{1-\gamma}\Ss\right)$ (see Section \ref{sec:magnus} for details), which commutes with $\Ss^{-1}$, and for $\betahnollv=\bm{0}$ Equation \ref{eq:beta_ef} simplifies to Equation \ref{eq:beta_gf}. When $\alpha=1$, all parameters are either in the inactive or in the coupled set; except when only one parameter is non-zero, then the free set consists of that single parameter and the coupled set is empty. That is, with $S_A:=S_F\cup S_C$ the definition of the $t_i$'s for coordinate flow and elastic gradient flow coincide. Letting $(1-\alpha)\to 0$ and using $\lim_{x\to 0}\frac{\Imm - \exp(-x(\bm{A}+x\bm{B}))}x=\bm{A}$, Equation \ref{eq:beta_ef} simplifies to Equation \ref{eq:beta_cf}.

In gradient flow (and descent), all parameters may update freely according to their gradient values, while both Definitions \ref{dfn:coord_flow} and \ref{dfn:en_flow} split the parameters into groups, with different update rules. For coordinate flow (and descent), some parameters are not allowed to update at all, while others update, but in a coupled fashion, making sure that the gradients are always equal. 
Elastic gradient flow (and descent) combines all of these three update properties. The free set contains the indices of the parameters for which $|g_d|>\alpha\cdot|g_m|$, which are updated according to their gradient value. The inactive set contains the indices of the parameters for which $|g_d|<\alpha\cdot|g_m|$, which are not updated. The coupled set contains the indices of the parameters for which $|g_d|=\alpha\cdot|g_m|$. In the discrete case this corresponds to $\left(\Iegd\right)_{dd}$ fluctuating between 0 and 1, and $|g_d|$ oscillating around $\alpha\cdot|g_m|$, while in the continuous case $\left(\Iegf\right)\in (0,1)$ and $|g_d|=\alpha\cdot|g_m|$. Since $\left(\Iegf\right)_{dd}$, and hence the update speed of these parameters, depends on the value of $|g_m|$, we refer to them as coupled. These three sets are illustrated in Section \ref{sec:diabetes}.

\subsection{Construction of Coordinate and Elastic Gradient Flow}
\label{constr:flows}
In this section, we present the details behind the definitions of coordinate and elastic gradient flow. The following notation is used: Uppercase boldface letters are used for matrices and lowercase boldface letters for vectors. Slices of matrices and vectors are marked by subscripts, which might be either a single index, a set of indices, or a colon that denotes an entire row/column. Complements are denoted with a minus sign. We will give two examples for $\bm{A}\in \R^{m\times n}$ and $S_A=\{3,5\}$: $\bm{A}_{S_A,-1}$, denotes a $2\times(n-1)$ matrix consisting of rows 3 and 5, and all but the first column of matrix $\bm{A}$. $\bm{A}_{:,-S_A}$ denotes an $m\times (n-2)$ matrix consisting of all rows and all columns except 3 and 5. Time derivatives of order $k$ are denoted interchangeably with $\frac{\partial^k}{\partial t^k}(\cdot)$ and $(\cdot)^{(k)}$ and $\odot$ denotes element-wise multiplication.

\subsubsection{Construction of Coordinate Flow}
\label{constr:coord_flow}
If at some time interval, $[t_1,t_2]$, the magnitudes of two or more gradient components are all close to the maximum gradient value, the index $m$ (where $m=\argmax_d|g_d|$ and $(\Icd)_{mm}=1$, as defined above) might alternate very frequently between these elements (or equivalently, the one in $\Icd$ frequently changes position along the diagonal) during the interval; we denote the corresponding set of indices $S_A(t_1,t_2)$. If we could look at an even finer time scale than $\Delta t$, we would observe the same behavior on the sub-interval $[t,t+\Delta t]$. To consider a finer time scale than $\Delta t$, we split the time step $\Delta t$ into $K$ sub-steps rewriting Equation \ref{eq:coord_desc} as
\begin{equation}
\label{eq:coord_flow1}
\begin{aligned}
\betahv(t+\Delta t)=&\betahv(t)+\gamma\left(\betahv(t)-\betahv(t-\Delta t)\right)\\
&+\sum_{k=0}^{K-1}\frac{\Delta t} K\cdot \Icf\left(t+k\frac{\Delta t} K\right)\sgn\left(\Ss\left(\betahatolsv-\betahv\left(t+k\frac{\Delta t}K\right)\right)\right).
\end{aligned}
\end{equation}

If $d\in S_A(t,t+\Delta t)$, then $|g_d(\tau)|> 0$ for $\tau \in [t,t+\Delta t]$, i.e.\ $g_d$ does not change sign, which means that $\sgn(g_d)$ remains constant on the interval. If, on the other hand, $d\notin S_A(t,t+\Delta t)$, then $g_d$ might change sign on the interval, but then $(\Icd)_{dd}=0$, and the value of $\sgn(g_d)$ is not considered. This means that Equation \ref{eq:coord_flow1} can be written as 
\begin{equation*}
%\label{eq:coord_flow2}
\betahv(t+\Delta t)=\betahv(t)+\gamma\left(\betahv(t)-\betahv(t-\Delta t)\right)+\frac{\Delta t}K\sum_{k=0}^{K-1}\Icf\left(t+k\frac{\Delta t} K\right)\sgn(\Ss(\betahatolsv-\betahv(t))).
\end{equation*}
Rearranging and letting first $K\to\infty$, then $\Delta t\to 0$, assuming that the limits exist, we obtain
$$(1-\gamma)\cdot\frac{\partial \betahv(t)}{\partial t}=\Icf^\infty(t)\sgn(\Ss(\betahatolsv-\betahv(t))),$$
where
$$\Icf^\infty(t):=\lim_{\Delta t\to 0}\lim_{K\to\infty}\frac 1K\sum_{k=0}^{K-1}\Icd\left(t+k\frac {\Delta t}K\right)$$
is the limit of averages of matrices where one diagonal element equals one and the remaining elements equal zero.
Since it is not obvious that this limit exists, we instead define $\Icf$ as an average of matrices of type $\Icd$, i.e.,
\begin{itemize}
\item $(\Icf)_{dd}\in [0,1]$
\item $\sum_d(\Icf)_{dd}=1$
\item $(\Icf)_{dd}(t)>0\iff d\in S_A(t)$
\end{itemize}
and obtain% Equation \ref{eq:coord_flow2}.
\begin{equation}
\label{eq:coord_flow2}
(1-\gamma)\cdot\frac{\partial \betahv(t)}{\partial t}=\Icf(t)\sgn(\Ss(\betahatolsv-\betahv(t))).
\end{equation}

Repeating the reasoning just after Equation \ref{eq:coord_flow1}, we can say that if $d\in S^i_A:=S_A(t_i,t_{i+1})$, then $g_d$ does not change sign for $t \in [t_i,t_{i+1})$ and 
$$-\sgn(g_d(t))=-\sgn(g_d(t_i))=\sgn(\Ss(\betahatolsv-\betahv(t_i)))_d=:(\svi)_d.$$
If, on the other hand, $d\notin S^i_A$, then $\left(\Icfi\right)_{dd}(t)=0$ and the value of $(\svi)_d$ is not considered. This means that Equation \ref{eq:coord_flow2} can be written as 
\begin{equation*}
%\label{eq:coord_flow4}
(1-\gamma)\cdot\frac{\partial\betahv(t)}{\partial t}=\Icf(t)\svi,\quad t\in[t_i,t_{i+1}),
\end{equation*}
to which the solution is given by
$$\betahv(t)=\betahv(t_i)+\frac1{1-\gamma}\int_{t_i}^t\Icf(\tau)d\tau\svi.$$

We now show that for $t\in[t_i,t_{i+1})$, $\int_{t_i}^t\Icf(\tau)d\tau=(t-t_i)\Icfi$, and calculate $\Icfi$.

Assume $S^i_A=\{m_1,m_2,\dots, m_{p_m}\}$ at $t=t_i$, which implies $|g_{m_1}(t_i)|=|g_{m_2}(t_i)|=\dots=|g_{m_{p_m}}(t_i)|$. If $d\notin S^i_A$ at time $t$, then $\left(\Icfi\right)_{dd}(t)=0$, so we focus on the sub-matrix $\left(\Icf\right)_{S^i_A,S^i_A}(t)$, which is a $p_m\times p_m$ matrix containing only the rows and columns for which $d\in S^i_A$ at time $t_i$.

We want to construct $\left(\Icf\right)_{S^i_A,S^i_A}(t)$ such that the elements of $S^i_A$ do not change for $t \in [t_i,t_{i+1})$, which implies $|g_{m_1}(t)|=|g_{m_2}(t)|=\dots=|g_{m_{p_m}}(t)|$.
Let's start with $|g_{m_1}(t)|=|g_{m_2}(t)|$.
\begin{equation*}
\begin{aligned}
0=&|g_{m_2}(t)|-|g_{m_1}(t)|=|\Ss(\betahatolsv-\betahv(t))|_{m_2}-|\Ss(\betahatolsv-\betahv(t))|_{m_1}\\
=&\si_{m_2}\cdot(\Ss(\betahatolsv-\betahv(t)))_{m_2}-\si_{m_1}\cdot(\Ss(\betahatolsv-\betahv(t)))_{m_1}\\
=&\underbrace{\cancel{\si_{m_2}\cdot(\Ss(\betahatolsv-\betahv(t_i)))_{m_2}}}_{=|g_{m_2}(t_i)|}-\si_{m_2}\cdot\left(\Ss\frac1{1-\gamma}\int_{t_i}^t\Icf(\tau)d\tau\svi\right)_{m_2}\\
&-\underbrace{\cancel{\si_{m_1}\cdot(\Ss(\betahatolsv-\betahv(t_i)))_{m_1}}}_{=|g_{m_1}(t_i)|}+\si_{m_1}\cdot\left(\Ss\frac1{1-\gamma}\int_{t_i}^t\Icf(\tau)d\tau\svi\right)_{m_1}\\
=&\frac1{1-\gamma}\left(\si_{m_1}\cdot\Ss_{m_1,:}-\si_{m_2}\cdot\Ss_{m_2,:}\right)\int_{t_i}^t\Icf(\tau)d\tau\svi\\
\stackrel{(a)}{=}&\frac1{1-\gamma}\left(\si_{m_1}\cdot\Ss_{m_1,S^i_A}-\si_{m_2}\cdot\Ss_{m_2,S^i_A}\right)\int_{t_i}^t\left(\Icf\right)_{S^i_A,S^i_A}(\tau)d\tau\svi_{S^i_A}\\
\iff 0=&\left(\si_{m_1}\cdot\Ss_{m_1,S^i_A}-\si_{m_2}\cdot\Ss_{m_2,S^i_A}\right)\int_{t_i}^t\left(\Icf\right)_{S^i_A,S^i_A}(\tau)d\tau\svi_{S^i_A},\\
\end{aligned}
\end{equation*}
where $(a)$ follows from the fact that $\left(\Icfi\right)_{dd}(t)=0$ for $d\notin S^i_A$.
Repeating the same calculations for all combinations of indices in $S^i_A$ gives us $p_m-1$ independent equations. Together with 
\begin{equation*}
\begin{aligned}
1&=\sum_{m\in S^i_A} \left(\Icf\right)_{mm}(t)\iff t-t_i=\sum_{m\in S^i_A}\int_{t_i}^t\left(\Icf\right)_{mm}(\tau)d\tau=\bm{1}^\top \int_{t_i}^t\left(\Icf\right)_{S^i_A,S^i_A}(\tau)d\tau\bm{1}\\
&=(\svi_{S^i_A})^\top \int_{t_i}^t\left(\Icf\right)_{S^i_A,S^i_A}(\tau)d\tau\svi_{S^i_A},
\end{aligned}
\end{equation*}
we obtain the following system of linear equations
$$
\underbrace{
\begin{bmatrix}
\horzbar (\svi_{S^i_A})^\top \horzbar\\
\horzbar \left(\si_{m_1}\cdot\Ss_{m_1,S^i_A}-\si_{m_2}\cdot\Ss_{m_2,S^i_A}\right) \horzbar \\
\horzbar \left(\si_{m_2}\cdot\Ss_{m_2,S^i_A}-\si_{m_3}\cdot\Ss_{m_3,S^i_A}\right) \horzbar \\
\vdots& \\
\horzbar \left(\si_{m_{p_m-1}}\cdot\Ss_{m_{p_m-1},S^i_A}-\si_{m_{p_m}}\cdot\Ss_{m_{p_m},S^i_A}\right) \horzbar\\
\end{bmatrix}
}_{=:\bm{B}}
\begin{bmatrix}
\vline\\
\int_{t_i}^t\left(\Icf\right)_{S^i_A,S^i_A}(\tau)d\tau\svi_{S^i_A}\\
\vline\\
\end{bmatrix}
=
\underbrace{
\begin{bmatrix}
t-t_i\\
0\\
\vdots\\
0\\
\end{bmatrix}
}_{=:\bm{b}}.
$$
Solving this system, we obtain
\begin{equation*}
\begin{aligned}
&\int_{t_i}^t\left(\Icf\right)_{S^i_A,S^i_A}(\tau)d\tau\svi_{S^i_A}=\bm{B}^{-1}\bm{b}=(t-t_i)\cdot (\bm{B}_{:,1})^{-1}\\
\iff&\int_{t_i}^t\left(\Icf\right)_{S^i_A,S^i_A}(\tau)d\tau=(t-t_i)\cdot \textrm{diag}\left((\bm{B}_{:,1})^{-1}\odot\svi_{S^i_A}\right)=:(t-t_i)\left(\Icfi\right)_{S^i_A,S^i_A},
\end{aligned}
\end{equation*}
where $\odot$ denotes element-wise multiplication.

For some combinations of $\Ss$ and $\betahatolsv-\betav(t)$, we might obtain a solution where $\left(\Icfi\right)_{dd}<0$ for some $d$. Since this is obviously not feasible (since it would correspond to a negative time step), in these cases, we remove the corresponding parameter from $S^i_A$, which implies $\left(\Icfi\right)_{dd}=0$. Then $\left(\Icfi\right)_{S^i_A,S^i_A}$ is recalculated for the new $S^i_A$. If $\left(\Icfi\right)_{dd}<0$ for more than one parameter simultaneously, only the parameter with the largest (absolute valued) negative value is removed. This procedure is repeated until all $\left(\Icfi\right)_{dd} \in [0,1]$.

What is left to do is to compute the times when some new $d$ enters $S^i_A$, i.e.\ when $|g_{d}(t)|=|g_{m}(t)|$ for $m\in S^i_A$, $d\notin S^i_A$. Since the absolute gradient value is identical for all $m\in S^i_A$, we use $m_1$. First assume $\sgn(g_d(t))=\sgn(g_{m_1}(t))$. Then
\begin{equation*}
\label{eq:delta_t1}
\begin{aligned}
0&=g_{m_1}(t)-g_{d}(t)=-\Ss_{m_1,:}(\betahatolsv-\betahv(t))+\Ss_{d,:}(\betahatolsv-\betahv(t))\\
&=-\Ss_{m_1,:}(\betahatolsv-\betahv(t_i))+(t-t_i)\Ss_{m_1,:}\Icfi\svi+\Ss_{d,:}(\betahatolsv-\betahv(t_i))-(t-t_i)\Ss_{d,:}\Icfi\svi\\
&=\left(\Ss_{d,:}-\Ss_{m_1,:}\right)(\betahatolsv-\betahv(t_i))-(t-t_i)\left(\Ss_{d,:}-\Ss_{m_1,:}\right)\Icfi\svi\\
&\iff t=t_i+\underbrace{\frac{\left(\Ss_{d,:}-\Ss_{m_1,:}\right)(\betahatolsv-\betahv(t_i))}{\left(\Ss_{d,:}-\Ss_{m_1,:}\right)\Icfi\svi}}_{=:\Delta t_{i,d,1}}.
\end{aligned}
\end{equation*}
Repeating the same calculations when $\sgn(g_d(t))=-\sgn(g_{m_1}(t))$ results in
\begin{equation*}
\label{eq:delta_t2}
t=t_i+\underbrace{\frac{\left(\Ss_{d,:}+\Ss_{m_1,:}\right)(\betahatolsv-\betahv(t_i))}{\left(\Ss_{d,:}+\Ss_{m_1,:}\right)\Icfi\svi}}_{=:\Delta t_{i,d,2}}.
\end{equation*}

Putting this together, we obtain
$$t_{i+1}=t_i+\min_{\substack{d\notin S^i_A,\ k=1,2\\ \Delta t_{i,d,k}>0}}\Delta t_{i,d,k}.$$%=t_i+\min_{\substack{d\notin S^i_M\\ ??>0}}
%\frac{\left(\Ss_{d,:}\pm\Ss_{m_1,:}\right)(\betahatolsv-\betavti)}{\left(\Ss_{d,:}\pm\Ss_{m_1,:}\right)\Ii\svi}.$$

\subsubsection{Construction of Elastic Gradient Flow}
\label{constr:en_flow}
This time the differential equation of interest becomes
\begin{equation}
\label{eq:en_flow}
(1-\gamma)\cdot\frac{\partial\betahv(t)}{\partial t}=\Iegf(\alpha,t)\left(\alpha\cdot\sgn\left(\Ss(\betahatolsv-\betahv(t))\right)+(1-\alpha)\cdot\Ss(\betahatolsv-\betahv(t))\right).
\end{equation}

Since the vector $\sgn(\Ss(\betahatolsv-\betahv(t)))$ is constant for $d\notin S_0$, Equation \ref{eq:en_flow} simplifies to
\begin{equation}
\label{eq:en_flow1}
(1-\gamma)\cdot\frac{\partial\betahvegf(t)}{\partial t}=\Iegfi(\alpha,t)\left(\alpha\cdot\svi+(1-\alpha)\cdot\Ss(\betahatolsv-\betahvegf(t))\right),\quad t\in[t_i,t_{i+1}),
\end{equation}
where $\svi:=-\sgn(g_d(t))=-\sgn(g_d(t_i))=\sgn(\Ss(\betahatolsv-\betahvegf(t_i)))$.

Let 
\begin{equation*}
%\label{eq:eta}
\etav(t):=\alpha\cdot\svi+(1-\alpha)\cdot\Ss(\betahatolsv-\betahvegf(t)).
\end{equation*}
Then
\begin{equation}
\label{eq:betaeta}
\betahvegf(t)=-\left((1-\alpha)\Ss\right)^{-1}\left(\etav(t)-\left(\alpha\cdot\svi+(1-\alpha)\cdot\Ss(\betahatolsv-\betahvegf(t_i))\right)\right)+\betahvegf(t_i)\\
\end{equation}
and
\begin{equation*}
\frac{\partial \etav(t)}{\partial t}=-(1-\alpha)\Ss\frac{\partial\betahvegf(t)}{\partial t}\iff\frac{\partial \betahvegf(t)}{\partial t}=-\left((1-\alpha)\Ss\right)^{-1}\frac{\partial\etav(t)}{\partial t}.
\end{equation*}

We can now rewrite Equation \ref{eq:en_flow1} in terms of $\etav$:
$$\frac{\partial\etav(t)}{\partial t}=-\frac{1-\alpha}{1-\gamma}\Ss\Iegfi(\alpha,t)\etav(t),\quad t\in[t_i,t_{i+1}),$$
which gives us
\begin{equation}
\label{eq:etasol}
\etav(t)=\exp\left(\Omi(t_i,t)\right)\etav(t_i),\quad \in [t_i,t_{i+1}),
\end{equation}
where $\exp$ is the matrix exponential and $\Omi(t_i,t)$ is the Magnus expansion \citep{magnus1954exponential} of $-\frac{1-\alpha}{1-\gamma}\Ss\Iegfi(\alpha,t)$.
Plugging Equation \ref{eq:etasol} into Equation \ref{eq:betaeta}, we obtain
\begin{equation*}
\begin{aligned}
\betahvegf(t)=& \betahvegf(t_i)\\
&+ \left((1-\alpha)\Ss\right)^{-1}\left(\Imm-\exp\left(\Omi(t_i,t)\right)\right)
\cdot\left(\alpha\cdot\svi+(1-\alpha)\cdot\Ss\left(\betahatolsv-\betahvegf(t_i)\right)\right),\\
&t\in [t_i,t_{i+1}).
\end{aligned}
\end{equation*}

We now calculate the time derivatives of $\left(\Ii\right)_{S^i_C,S^i_C}(t_i)$. Let $m:=\textrm{argmax}_d|g_d(t_i)|$. If $c\in S^i_C$, then $(\Ii)_{cc}(t)\notin \{0,1\}$ and $|g_c(t)|=\alpha|g_{m}(t)|$.
This means that for $t \in [t_i,t_{i+1})$
\begin{equation}
\label{eq:abs_eq_flow}
\left|\frac{g_c(t)}{g_m(t)}\right|=\alpha,
\end{equation}
which we want to solve for $(\Ii)_{cc}(\alpha,t):=(\I)_{cc}(\alpha,t),$ $t\in[t_i,t_{i+1})$.\\
Since $\left|\frac{g_c(t_i)}{g_m(t_i)}\right|=\alpha$, Equation \ref{eq:abs_eq_flow} holds if 
\begin{equation}
\label{eq:grad_quot}
\left(\left|\frac{g_c(t_i)}{g_m(t_i)}\right|\right)^{(k+1)}=0,\ k=0,1,\dots,
\end{equation}
because requiring that the derivative is 0 for $t_i$ implies that the function remains constant for $t>t_i$, at least if the derivative remains zero, which is why we require the second derivative to be 0, and so on.
Using Lemma \ref{thm:abs_der}, we obtain
\begin{equation}
\label{eq:der0}
\begin{aligned}
\left(\left|\frac{g_c(t_i)}{g_m(t_i)}\right|\right)^{(k+1)}&=0\\
\iff&\\
 g^{(k+1)}_c(t_i)\cdot g_m(t_i)-g_c(t_i)\cdot g^{(k+1)}_m(t_i)&=-g^2_m(t_i)\cdot\sgn\left(\frac {g_c(t_i)}{g_m(t_i)}\right)\cdot\Oo\left(\left(\frac {g_c(t_i)}{g_m(t_i)} \right)^{(k)}\right).
\end{aligned}
\end{equation}

If we solve Equation \ref{eq:grad_quot} for $k$'s in increasing order, when solving for $k+1$,
$$\left(\left|\frac{g_c(t_i)}{g_m(t_i)}\right|\right)^{(l)}=0,\ l=1,2,\dots k$$
and Equation \ref{eq:der0} simplifies to
\begin{equation}
\label{eq:der0_1}
\left(\left|\frac{g_c(t_i)}{g_m(t_i)}\right|\right)^{(k+1)}=0\iff g^{(k+1)}_c(t_i)\cdot g_m(t_i)-g_c(t_i)\cdot g^{(k+1)}_m(t_i)=0.
\end{equation}

Using Lemma \ref{thm:gk0}, abbreviating $\gsvi:=\left(\alpha\cdot\sgn(\gv(t_i))+(1-\alpha)\cdot\gv(t_i)\right)$ and $\ckv:=\Oo((\Ii)^{(k-1)}(t_i)$, we obtain
\begin{equation}
\label{eq:gk0_1}
g_d^{(k+1)}(t_i)=\frac1{1-\gamma}\left(-\Ss_{d,:}(\Ii)^{(k)}(t_i)\gsvi + \ck_d\right),
\end{equation}
where $\ckv$ can be calculated by first setting $(\Ii)^{(k)}(t_i)=\bm{0}$, to remove the first term in Equation \ref{eq:gk0_1}, and then evaluating $\gv^{(k+1)}(t_i)$ using Equations \ref{eq:grad} and \ref{eq:beta_ef}:
\begin{equation*}
\begin{aligned}
\ckv&=(1-\gamma)\gv^{(k+1)}(t_i) =(1-\gamma)\Ss\betav^{(k+1)}(t_i)\\
&=-(1-\gamma)\Ss\left.\frac{d^{k+1}\exp\left(\Omi(t_i,t)\right)}{dt^k}\right|_{t=t_i}\left(\frac {\alpha}{1-\alpha}\Ss^{-1}\svi+\betahatolsv-\betavti\right),
\end{aligned}
\end{equation*}
where the derivative of the matrix exponential can be calculated using Lemma \ref{thm:mat_exp_der}.

Now, inserting Equation \ref{eq:gk0_1} into Equation \ref{eq:der0_1}, we obtain.
\begin{equation*}
\label{eq:ii_ab}
\begin{aligned}
&\frac1{1-\gamma}\left(-\Ss_{c,:}(\Ii)^{(k)}(t_i)\gsvi+\ck_c\right)\cdot g_m(t_i)\\
&-g_c(t_i)\cdot\frac1{1-\gamma}\left(-\Ss_{m,:}(\Ii)^{(k)}(t_i)\gsvi+\ck_m\right)=0\\
\iff&(-\Ss_{c,:}(\Ii)^{(k)}(t_i)\gsvi)+\ck_c)\cdot g_m(t_i)-g_c(t_i)\cdot(-\Ss_{m,:}(\Ii)^{(k)}(t_i)\gsvi)+\ck_m)=0\\
\iff&\ck_m\cdot g_c(t_i)-\ck_c\cdot g_m(t_i)=\left(g_c(t_i)\Ss_{m,:}-g_m(t_i)\cdot\Ss_{c,:}\right)(\Ii)^{(k)}(t_i)\gsvi\\
&\stackrel{(a)}{=}
\begin{bmatrix}
g_c(t_i)\Ss_{m,c}-g_m(t_i)\cdot\Ss_{c,c} & g_{c}(t_i)\Ss_{m,-c}-g_m(t_i)\cdot\Ss_{c,-c}
\end{bmatrix}
\begin{bmatrix}
\left((\Ii)^{(k)}(t_i)\gsvi\right)_c\\
\left((\Ii)^{(k)}(t_i)\gsvi\right)_{-c}
\end{bmatrix}\\
\iff&= \underbrace{\ck_m\cdot g_c(t_i)-g_m(t_i)\cdot\ck_c-\left(g_c(t_i)\Ss_{m,-c}-g_m(t_i)\cdot\Ss_{c,-c}\right)\left((\Ii)^{(k)}(t_i)\gsvi\right)_{-c}}_{=:b}\\
&=\underbrace{\left(g_c(t_i)\Ss_{m,c}-g_m(t_i)\cdot\Ss_{c,c}\right)}_{=:a}\underbrace{\left((\Ii)^{(k)}(t_i)\gsvi\right)_c}_{=(\Ii)_{cc}^{(k)}(t_i)\gsi_c}\\
\iff&(\Ii)_{cc}^{(k)}(t_i)=\frac{b}{a\cdot \gsi_c},\\
\end{aligned}
\end{equation*}
where in $(a)$, the columns of $\Ss$ and the rows of $(\Ii)^{(k)}(t_i)\gsvi$ are split into the two sets $c$ and $-c$.

Writing this as a system of linear equations, we obtain
\begin{equation*}
\begin{aligned}
\bm{A}:=&\left(\gv_{S^i_C}(t_i)\Ss_{m,S^i_C}-g_m(t_i)\cdot\Ss_{S^i_C,S^i_C}\right)\\
\bm{b}:=&\ck_m\cdot\gv_{S^i_C}(t_i)-g_m(t_i)\cdot\ckv_{S^i_C}\\
&-\left(\gv_{S^i_C}(t_i)\Ss_{m,S^i_F\cup S^i_0}-g_m(t_i)\cdot\Ss_{S^i_C,S^i_F\cup S^i_0}\right)\left((\Ii)^{(k)}(t_i)\gsvi\right)_{S^i_F\cup S^i_0}\\
\left(\Ii\right)_{S^i_C,S^i_C}^{(k)}(t_i)=&\textrm{diag}(\bm{A}^{-1}\bm{b}/\gsvi_{S^i_C}),
\end{aligned}
\end{equation*}
where the division is element-wise.

For $k=0$, we have $\ckv=(\Ii)^{(-1)}(t_i)=\bm{0}$, $(\Ii)_{S^i_0,S^i_0}(t_i)=\bm{0}$ and $(\Ii)_{S^i_F,S^i_F}(t_i)=\Imm$, which means that $\bm{b}$ simplifies to
$$\bm{b}_{k=0}=\left(g_m(t_i)\cdot\Ss_{S^i_C,S^i_F}-\gv_{S^i_C}(t_i)\Ss_{m,S^i_F}\right)\gsvi_{S^i_F}.$$
For $k\geq1$, we have $(\Ii)_{S^i_F\cup S^i_0,S^i_F\cup S^i_0}^{(k)}(t_i)=\bm{0}$, which means that $\bm{b}$ simplifies to
$$\bm{b}_{k\geq1}=\ck_m\cdot\gv_{S^i_C}(t_i)-g_m(t_i)\cdot\ckv_{S^i_C}.$$

Similar to for coordinate flow, $(\I)_{cc}(\alpha,t)\in[0,1]$ might not always hold and the corresponding modification, in this case, amounts to: If at any time $t_i$, $(\Ii)_{cc}\leq0$ for some parameter in $S^i_C$, then that parameter is moved from $S^i_C$ to $S^i_0$; if $(\Ii)_{dd}\geq1$ for some parameter in $S^i_C$, then that parameter is moved from $S^i_C$ to $S^i_F$. If $(\Ii)_{dd}\notin[0,1]$ for more than one parameter simultaneously, only the parameter that deviates most from $[0,1]$ is removed. This procedure is repeated until all $(\Ii)_{dd} \in [0,1]$, which means that $(\Ii)_{cc}\in (0,1)$ for $c\in S^i_C$.

\subsection{Magnus Expansion}
\label{sec:magnus}
According to \cite{magnus1954exponential},
\begin{equation}
\label{eq:magnus_exp}
\begin{aligned}
\Omi(t_i,t)=&\int_{t_i}^t\bm{A}(\tau_1)d\tau_1+\frac12\int_{t_i}^t\int_{t_i}^{\tau_1}\left[\bm{A}(\tau_1),\bm{A}(\tau_2)\right]d\tau_2d\tau_1\\
&+\frac14\int_{t_i}^t\int_{t_i}^{\tau_1}\int_{t_i}^{\tau_2}\left[\bm{A}(\tau_1),\left[\bm{A}(\tau_2),\bm{A}(\tau_3)\right]\right]d\tau_3d\tau_2d\tau_1+\dots,
\end{aligned}
\end{equation}
where the commutator is defined according to $[\bm{A},\bm{B}]:=\bm{A}\bm{B}-\bm{B}\bm{A}$. For $\bm{A}(t)=-(1-\alpha)\Ii(\alpha,t)\Ss$, the first two terms in Equation \ref{eq:magnus_exp}, together with its time derivatives, are calculated below.
\subsubsection{First Term}
\begin{equation*}
\begin{aligned}
\Omi_1(t_i,t)&=\int_{t_i}^t-\frac{1-\alpha}\Ss{1-\gamma}\Ii(\alpha,\tau) d\tau\stackrel{(a)}=-\frac{1-\alpha}{1-\gamma}\Ss\int_{t_i}^t\sum_{l=0}^\infty\frac{(\tau-t_i)^l}{l!}(\Ii)^{(l)}(\alpha,t_i)d\tau\\
&=-\frac{1-\alpha}{1-\gamma}\Ss\sum_{l=0}^\infty\frac{(t-t_i)^{l+1}}{(l+1)!}(\Ii)^{(l)}(\alpha,t_i).\\
(\Omi_1)^{(k)}(t_i,t)&=-\frac{1-\alpha}{1-\gamma}\Ss\sum_{l=0}^\infty\frac{(t-t_i)^{l+1-k}}{(l+1-k)!}(\Ii)^{(l)}(\alpha,t_i)\\
(\Omi_1)^{(k)}(t_i,t_i)&=-\frac{1-\alpha}{1-\gamma}\Ss(\Ii)^{(k-1)}(\alpha,t_i),
\end{aligned}
\end{equation*}
where in $(a)$, we use the Taylor expansion of $\Ii(\alpha,t)$ around $t_i$, and the last step uses that $(t_i-t_i)^{l+1-k}=\begin{cases} 1\textrm{ if } l=k-1\\0 \textrm{ else.}\end{cases}$

\subsubsection{Second Term}
\begin{equation*}
\begin{aligned}
\Omi_2(t_i,t)&=\frac12\int_{t_i}^t\int_{t_i}^{\tau_1}\left[-\frac{1-\alpha}{1-\gamma}\Ss\Ii(\alpha,\tau_1),\ -\frac{1-\alpha}{1-\gamma}\Ss\Ii(\alpha,\tau_2)\right] d\tau_2d\tau_1\\
&=\frac12\left(\frac{1-\alpha}{1-\gamma}\right)^2\int_{t_i}^t\int_{t_i}^{\tau_1}\left[\Ss\Ii(\alpha,\tau_1),\ \Ss\Ii(\alpha,\tau_2)\right] d\tau_2d\tau_1.
\end{aligned}
\end{equation*}
Focusing on the commutator, and writing $\Ii(\alpha,t)$ as its Taylor expansion, we obtain
\begin{equation*}
\begin{aligned}
&\left[\Ss\Ii(\alpha,\tau_1),\ \Ss\Ii(\alpha,\tau_2)\right]\\
=&\left[\Ss\sum_{l_1=0}^\infty\frac{(\tau_1-t_i)^{l_1}}{l_1!}(\Ii)^{(l_1)}(\alpha,t_i),\ \Ss\sum_{l_2=0}^\infty\frac{(\tau_2-t_i)^{l_2}}{l_2!}(\Ii)^{(l_2)}(\alpha,t_i)\right]\\
=&\sum_{l_1=0}^\infty\sum_{l_2=0}^\infty\frac{(\tau_1-t_i)^{l_1}}{l_1!}\frac{(\tau_2-t_i)^{l_2}}{l_2!}\left[\Ss(\Ii)^{(l_1)}(\alpha,t_i),\ \Ss(\Ii)^{(l_2)}(\alpha,t_i)\right]\\
\stackrel{(a)}=&\sum_{l_1=1}^\infty\sum_{l_2=0}^{l_1-1}\frac1{l_1!l_2!}\left((\tau_1-t_i)^{l_1}(\tau_2-t_i)^{l_2}-(\tau_2-t_i)^{l_2}(\tau_1-t_i)^{l_1}\right)\\
&\cdot\left[\Ss(\Ii)^{(l_1)}(\alpha,t_i),\ \Ss(\Ii)^{(l_2)}(\alpha,t_i)\right],
\end{aligned}
\end{equation*}
where in $(a)$, we use
\begin{equation*}
\begin{aligned}
&[\Ss(\Ii)^{(l)}(\alpha,t_i),\ \Ss(\Ii)^{(l)}(\alpha,t_i)]=0\\
&[\Ss(\Ii)^{(l_2)}(\alpha,t_i),\ \Ss(\Ii)^{(l_1)}(\alpha,t_i)]=-[\Ss(\Ii)^{(l_1)}(\alpha,t_i),\ \Ss(\Ii)^{(l_2)}(\alpha,t_i)].
\end{aligned}
\end{equation*}
We now solve the integral with respect to $\tau_1$ and $\tau_2$,
$$\int_{t_i}^t\int_{t_i}^{\tau_1}\left((\tau_1-t_i)^{l_1}(\tau_2-t_i)^{l_2}-(\tau_2-t_i)^{l_2}(\tau_1-t_i)^{l_1}\right)d\tau_2d\tau_1$$
$$=\frac{l_1-l_2}{(l_1+1)!(l_2+1)!(l_1+l_2+2)}(t-t_i)^{l_1+l_2+2},$$
and putting it together, we obtain
\begin{equation*}
\begin{aligned}
&\Omi_2(t_i,t)\\
=&\frac12\left(\frac{1-\alpha}{1-\gamma}\right)^2\sum_{l_1=1}^\infty\sum_{l_2=0}^{l_1-1}\frac{(l_1-l_2)(t-t_i)^{l_1+l_2+2}}{(l_1+1)!(l_2+1)!(l_1+l_2+2)}\left[\Ss(\Ii)^{(l_1)}(\alpha,t_i),\ \Ss(\Ii)^{(l_2)}(\alpha,t_i)\right]\\
=&\frac12\left(\frac{1-\alpha}{1-\gamma}\right)^2\sum_{l_1=2}^\infty\sum_{l_2=1}^{l_1-1}\frac{(l_1-l_2)(t-t_i)^{l_1+l_2}}{l_1!l_2!(l_1+l_2)}\left[\Ss(\Ii)^{(l_1-1)}(\alpha,t_i),\ \Ss(\Ii)^{(l_2-1)}(\alpha,t_i)\right].
\end{aligned}
\end{equation*}
Differentiating $k$ times with respect to $t$ yields
\begin{equation*}
\begin{aligned}
(\Omi_2)^{(k)}(t_i,t)=&\frac12\left(\frac{1-\alpha}{1-\gamma}\right)^2\sum_{l_1=2}^\infty\sum_{l_2=1}^{l_1-1}\frac{(l_1-l_2)(l_1+l_2)!\cdot(t-t_i)^{l_1+l_2-k}}{l_1!l_2!(l_1+l_2)(l_1+l_2-k)!}\\
&\cdot\left[\Ss(\Ii)^{(l_1-1)}(\alpha,t_i),\ \Ss(\Ii)^{(l_2-1)}(\alpha,t_i)\right]\\
(\Omi_2)^{(k)}(t_i,t_i)=&\frac12\left(\frac{1-\alpha}{1-\gamma}\right)^2\sum_{l_2=1}^{\lfloor\frac{k-1}{2}\rfloor}\frac{(k-2l_2)(k-1)!}{l_2!(k-l_2)!}\\
&\cdot\left[\Ss(\Ii)^{(k-l_2-1)}(\alpha,t_i),\ \Ss(\Ii)^{(l_2-1)}(\alpha,t_i)\right],
\end{aligned}
\end{equation*}
where the last step uses that the only surviving term in the first sum is when $l_1+l_2-k=0 \iff l_1=k-l_2$.

\section{Proofs}
\label{sec:proofs}

\begin{proof}[Proof of Proposition \ref{thm:conv_anal}]~\\
Since $L$ is assumed to be strongly convex, with Hessian bounded by $M$, 

\begin{equation*}
\begin{aligned}
&L(\betahv-\Delta t\cdot \dbetahvegd)=L(\betahv)-\Delta t\cdot \nabla L(\betahv)^\top\dbetahvegd+\frac{\Delta t^2}2\dbetahvegd^\top\nabla^2L(\betahv)\dbetahvegd\\
\leq& L(\betahv)-\Delta t\cdot \nabla L(\betahv)^\top\dbetahvegd+\frac{M \Delta t^2}2\dbetahvegd^\top\dbetahvegd\\
=&L(\betahv)-\Delta t\cdot \gv^\top\Iegd\cdot(\alpha\cdot\sgn(\gv)+(1-\alpha)\gv)\\
&+\frac{M \Delta t^2}2 (\alpha\cdot\sgn(\gv)+(1-\alpha)\gv)^\top\underbrace{\Iegd^\top\Iegd}_{=\Iegd}\cdot(\alpha\cdot\sgn(\gv)+(1-\alpha)\gv).\\
=&L(\betahv)-\Delta t\cdot\left(\alpha \cdot \underbrace{\gv^\top\Iegd\sgn(\gv)}_{=\|\Iegd\gv\|_1}+(1-\alpha)\cdot \underbrace{\gv^\top\Iegd\gv}_{=\|\Iegd\gv\|^2_2}\right)\\
&+\frac{M \Delta t^2}2 \left(\alpha^2\cdot\underbrace{\sgn(\gv)\Iegd\sgn(\gv)}_{=\|\Iegd\gv\|_0}+2\alpha(1-\alpha)\cdot\underbrace{\gv^\top\Iegd\sgn(\gv)}_{=\|\Iegd\gv\|_1}+(1-\alpha)^2\cdot\underbrace{\gv^\top\Iegd\gv}_{\|\Iegd\gv\|_2^2}\right).
\end{aligned}
\end{equation*}
Using the inequalities
\begin{equation*}
\begin{aligned}
\|\Iegd\gv\|_1\geq\frac1\gmax\cdot\|\Iegd\gv\|_2^2\\
\|\Iegd\gv\|_0\leq\frac1{\gmin^2}\cdot\|\Iegd\gv\|_2^2\\
\|\Iegd\gv\|_1\leq\frac1\gmin\cdot\|\Iegd\gv\|_2^2,\\
\end{aligned}
\end{equation*}
we obtain
\begin{equation*}
\begin{aligned}
&L(\betahv-\Delta t\cdot \dbetahvegd)-L(\betahv)\leq
-\Delta t\cdot\left(\alpha \cdot \frac1{\gmax}\cdot\|\Iegd\gv\|_2^2+(1-\alpha)\cdot \|\Iegd\gv\|^2_2\right)\\
&+\frac{M \Delta t^2}2 \left(\alpha^2\cdot\frac1{\gmin^2}\cdot\|\Iegd\gv\|_2^2+2\alpha(1-\alpha)\cdot\frac1{\gmin}\cdot\|\Iegd\gv\|_2^2+(1-\alpha)^2\cdot\|\Iegd\gv\|_2^2\right).\\
=&-\Delta t\cdot\|\Iegd\gv\|_2^2\cdot\left(\frac\alpha{\gmax}+(1-\alpha)-\frac{M \Delta t}2 \left(\frac{\alpha^2}{\gmin^2}+\frac{2\alpha(1-\alpha)}{\gmin}+(1-\alpha)^2\right)\right).\\
=&-\Delta t\cdot\|\Iegd\gv\|_2^2\cdot\left(\frac\alpha{\gmax}+(1-\alpha)-\frac{M \Delta t}2 \left(\frac{\alpha}{\gmin}+(1-\alpha)\right)^2\right).
\end{aligned}
\end{equation*}
Thus, $L(\betahv-\Delta t\cdot \dbetahvegd)-L(\betahv)\leq 0$ if  

\begin{equation}
\label{eq:bound2a}
\begin{aligned}
&\frac\alpha{\gmax}+(1-\alpha)-\frac{M \Delta t}2 \left(\frac{\alpha}{\gmin}+(1-\alpha)\right)^2\geq 0\\
\iff &\Delta t\leq\frac2M\cdot\frac{\alpha\gmax^{-1}+(1-\alpha)}{(\alpha\gmin^{-1}+(1-\alpha))^2}=\frac2M\cdot\frac{\gmin^2}{\gmax}\cdot \frac{\alpha+(1-\alpha)\gmax}{(\alpha+(1-\alpha)\gmin)^2}.
\end{aligned}
\end{equation}

Using $\gmin\geq\alpha\cdot \gmax$ and $\gmin>0$, we can remove $\gmin$ from Equation \ref{eq:bound2a}, lower the expression. From the two bounds on $\gmin$ we obtain
$$\frac\alpha\gmin\geq \frac{1_{\alpha> 0}}{\gmax},$$
where $\gmin\geq\alpha$ is used for $\alpha>0$, and $\gmin>0$ is used for $\alpha=0$. Thus, we obtain
\begin{equation*}
\Delta t\leq \frac2M\cdot\frac{\alpha\gmax^{-1}+(1-\alpha)}{\left(1_{\alpha>0}\cdot\gmax^{-1}+(1-\alpha)\right)^2}=\frac2M\cdot\gmax\cdot\frac{\alpha+(1-\alpha)\gmax}{\left(1_{\alpha>0}+(1-\alpha)\gmax\right)^2}.
\end{equation*}
\end{proof}

\begin{proof}[Proof of Lemma \ref{thm:ridge_rewrite}]~\\
The ridge estimate is defined as
$$\betahatv(\lambda):=(\Xm^\top\Xm+n\lambda \Imm)^{-1}\Xm^\top\yv,$$
which can be rewritten as
\begin{equation*}
\begin{aligned}
(\Xm^\top\Xm+n\lambda \Imm)^{-1}\Xm^\top\yv
&\stackrel{(a)}=(\underbrace{\Xm^\top\Xm}_{n\Ss}+n\lambda \Imm)^{-1}\underbrace{(\Xm^\top\Xm)}_{n\Ss}\underbrace{(\Xm^\top\Xm)^{+}\Xm^\top\yv}_{\betahatolsv}\\
&=\frac1n(\Ss+\lambda \Imm)^{-1}n\Ss\betahatolsv
=(\Ss+\lambda \Imm)^{-1}(\Ss+\lambda\Imm-\lambda\Imm)\betahatolsv\\
&=(\underbrace{(\Ss+\lambda \Imm)^{-1}(\Ss+\lambda \Imm)}_{\Imm}-\lambda(\Ss+\lambda \Imm)^{-1})\betahatolsv\\
&=\left(\Imm-\left(\Imm+\frac 1 \lambda \Ss\right)^{-1}\right)\betahatolsv,
\end{aligned}
\end{equation*}
where $(a)$ follows from $\Xm^\top=(\Xm^\top\Xm)(\Xm^\top\Xm)^{+}\Xm^\top$.
\end{proof}

\begin{proof}[Proof of Proposition \ref{thm:t_lbda}]~\\
We first note that when $\Ss=\Imm$, $\Ii(\alpha,t_1)\Imm$ and $\Ii(\alpha,t_2)\Imm$ are diagonal and thus commute, which means that $\Omi(t_i,t)$ reduces to $-\frac{1-\alpha}{1-\gamma}\int_{t_i}^t\Ii(\alpha,\tau)d\tau$. 

Since the data is uncorrelated and $\betahv(t_0)=\bm{0}$, $|\betahat^{\text{\normalfont egf}}_d(t)|\leq|\betahatols_d|$ and for $\betahat^{\text{\normalfont egf}}_d(t)\neq0$, $\sgn(\betahat^{\text{\normalfont egf}}_d(t))=\sgn(\betahatols_d)=\sgn(\betahatols_d-\betahat^{\text{\normalfont egf}}_d(t))$. This means that Equation \ref{eq:beta_ef} can be written as
\begin{equation*}
\begin{aligned}
\betahat^{\text{\normalfont egf}}_d(t)=&\sgn(\betahatols_d)\cdot\min\left(|\betahat^{\text{\normalfont egf}}_d(t_i)|+\frac {1}{1-\alpha}\left(1-\exp\left(-\frac{1-\alpha}{1-\gamma}\int_{t_i}^t(\Ii)_{dd}(\alpha,\tau)d\tau\right)\right)\right.\\
&\left.\cdot\left(\alpha+(1-\alpha)\left(|\betahatols_d|-|\betahat^{\text{\normalfont egf}}_d(t_i)|\right)\right), |\betahatols_d|\right),
\end{aligned}
\end{equation*}
where the minimum, which is included to assure that $|\betahat^{\text{\normalfont egf}}_d(t)|\leq|\betahatols_d|$, becomes active once the OLS solution is reached.

We now compare this the closed form solution of the elastic net,
$$\betahat^{\text{\normalfont en}}_d(\lambda)=\sgn(\betahatols_d)\frac{\max(|\betahatols_d|-\alpha\lambda,0)}{1+(1-\alpha)\lambda}=\sgn(\betahatols_d)\frac{\left(|\betahatols_d|-\min(\alpha\lambda,|\betahatols_d|)\right)}{1+(1-\alpha)\lambda}.$$

Requiring $\betahat^{\text{\normalfont en}}_d(\lambda)=\betahat^{\text{\normalfont egf}}_d(t)$, we obtain
\begin{equation*}
\begin{aligned}
\betahat^{\text{\normalfont en}}_d(\lambda)=\betahat^{\text{\normalfont egf}}_d(t)\iff
\lambda_d=\max\left(\frac{|\betahatols_d|-|\betahat^{\text{\normalfont egf}}_d(t_i)|-\vv}
{\alpha +(1-\alpha) \left(|\betahat^{\text{\normalfont egf}}_d(t_i)|+\vv\right)},0\right),
\end{aligned}
\end{equation*}
where $$\vv=\frac1{1-\alpha}\left(1-\exp\left(-\frac{1-\alpha}{1-\gamma}\int_{t_i}^t(\Ii)_{dd}(\alpha,\tau)d\tau\right)\right)\left(\alpha+(1-\alpha)\left(|\betahatols_d|-|\betahat^{\text{\normalfont egf}}_d(t_i)|\right)\right),$$
and where the equivalence is tedious but straightforward to show.

Letting $(1-\alpha)\to 0$, using $\lim_{x\to0}\frac{1-e^{-ax}}{x}=a$, we obtain
$$\lambda_d=\max\left(|\betahatols_d|-|\betahat^{\text{\normalfont cf}}_d(t_i)|-\frac{t-t_i}{1-\gamma}(\Icfi)_{dd},0\right).$$

To calculate $\sgn\left(\frac{\partial \lambda_d(t)}{\partial t}\right)$, we note first that
\begin{equation*}
\begin{aligned}
\frac{\partial \vv(t)}{\partial t} = & \frac1{1-\gamma}(\Ii)_{dd}(\alpha,t)\exp\left(-\frac{1-\alpha}{1-\gamma}\int_{t_i}^t(\Ii)_{dd}(\alpha,\tau)d\tau\right)\\
&\cdot\left(\alpha+(1-\alpha)\underbrace{\left(|\betahatols_d|-|\betahat^{\text{\normalfont egf}}_d(t_i)|\right)}_{\geq 0}\right)\geq 0,
\end{aligned}
\end{equation*}
which implies
\begin{equation*}
\begin{aligned}
\frac{\partial \lambda_d(t)}{\partial t}=&\underbrace{\frac{\partial(\max(0,\lambda_d(t)))}{\partial \lambda_d}}_{=:f_1(t)\geq 0}\cdot \underbrace{\frac1{\left(\alpha +(1-\alpha) \left(|\betahat^{\text{\normalfont egf}}_d(t_i)|+\vv(t)\right)\right)^2}}_{=:f_2(t)\geq 0}\\
&\cdot\left(\frac{\partial}{\partial t}\left(|\betahatols_d|-|\betahat^{\text{\normalfont egf}}_d(t_i)|-\vv(t)\right)\cdot\left(\alpha +(1-\alpha) \left(|\betahat^{\text{\normalfont egf}}_d(t_i)|+\vv(t)\right)\right)\right.\\
&-\left.\left(|\betahatols_d|-|\betahat^{\text{\normalfont egf}}_d(t_i)|-\vv(t)\right)\cdot\frac{\partial}{\partial t}\left(\alpha +(1-\alpha) \left(|\betahat^{\text{\normalfont egf}}_d(t_i)|+\vv(t)\right)\right)\right)\\
=&f_1(t)\cdot f_2(t)\cdot\left(-\frac{\partial\vv(t)}{\partial t}\cdot\left(\alpha +\cancel{(1-\alpha) \left(|\betahat^{\text{\normalfont egf}}_d(t_i)|+\vv(t)\right)}\right)\right.\\
&-\left.\left(|\betahatols_d|-\cancel{|\betahat^{\text{\normalfont egf}}_d(t_i)|-\vv(t)}\right)\cdot(1-\alpha)\frac{\partial\vv(t)}{\partial t}\right)\\
=&-\underbrace{f_1(t)}_{\geq 0}\cdot \underbrace{f_2(t)}_{\geq 0}\cdot\underbrace{\frac{\partial\vv(t)}{\partial t}}_{\geq 0}\cdot\underbrace{\left(\alpha+(1-\alpha)\cdot|\betahatols_d|\right)}_{\geq 0}\leq 0.
\end{aligned}
\end{equation*}
\end{proof}

\begin{proof}[Proof of Proposition \ref{thm:h_alpha_sd}]~\\
For $c_\alpha\geq 0$,
\begin{equation*}
\begin{aligned}
\left\|c_\alpha\dbetahvegds\right\|_1
=& c_\alpha\left\|\dbetahvegds\right\|_1= c_\alpha\|\Iegds\gv\|_1\left(\frac{\alpha}{\left\|\Iegds\gv\right\|_1}+\frac{1-\alpha}{\left\|\Iegds\gv\right\|_2}\right)\\
=&c_\alpha\left(\alpha+(1-\alpha)\frac{\left\|\Iegds\gv\right\|_1}{\left\|\Iegds\gv\right\|_2}\right)
=c_\alpha\left(\alpha+(1-\alpha)\sqrt{q_1}\right).\\
\end{aligned}
\end{equation*}
\begin{equation*}
\begin{aligned}
\left\|c_\alpha\dbetahvegds\right\|_2^2
=&c_\alpha^2\|\Iegds\gv\|_2^2\left(\frac{\alpha}{\left\|\Iegds\gv\right\|_1}+\frac{1-\alpha}{\left\|\Iegds\gv\right\|_2}\right)^2\\
=&c_\alpha^2\left(\alpha\frac{\left\|\Iegds\gv\right\|_2}{\left\|\Iegds\gv\right\|_1}+(1-\alpha)\right)^2 
=c_\alpha^2\left(\frac\alpha{\sqrt{q_1}}+1-\alpha\right)^2.\\
\end{aligned}
\end{equation*}
Solving 
\begin{equation*}
\begin{aligned}
1=&\alpha \left\|c_\alpha\dbetahvegds\right\|_1 +(1-\alpha)\left\|c_\alpha\dbetahvegds\right\|_2^2\\
=&c_\alpha\alpha\left(\alpha +(1-\alpha)\cdot\sqrt{q_1}\right) +c_\alpha^2(1-\alpha)\left(\alpha\cdot\frac{1}{\sqrt{q_1}} +1-\alpha\right)^2
\end{aligned}
\end{equation*}
for $c_\alpha$, the non-negative root is
$$c_\alpha=\frac{\sqrt{q_1(\alpha^2q_1+4(1-\alpha))}-\alpha q_1}{2(1-\alpha)(\sqrt{q_1}(1-\alpha)+\alpha)}.$$
\end{proof}

\begin{proof}[Proof of Proposition \ref{thm:eg_norm_sd}]
\begin{equation*}
\begin{aligned}
\left\|\dbetahvegds\right\|_1
=& \|\Iegds\gv\|_1\left(\frac{\alpha}{\left\|\Iegds\gv\right\|_1}+\frac{1-\alpha}{\left\|\Iegds\gv\right\|_2}\right)\\
=&\left(\alpha+(1-\alpha)\frac{\left\|\Iegds\gv\right\|_1}{\left\|\Iegds\gv\right\|_2}\right).
\end{aligned}
\end{equation*}
\begin{equation*}
\begin{aligned}
\left\|\dbetahvegds\right\|_2^2
=&\|\Iegds\gv\|_2^2\left(\frac{\alpha}{\left\|\Iegds\gv\right\|_1}+\frac{1-\alpha}{\left\|\Iegds\gv\right\|_2}\right)^2\\
=&\left(\alpha\frac{\left\|\Iegds\gv\right\|_2}{\left\|\Iegds\gv\right\|_1}+(1-\alpha)\right)^2.
\end{aligned}
\end{equation*}

\begin{equation*}
\begin{aligned}
&\alpha \left\|\dbetahvegds\right\|_1 +(1-\alpha)\left\|\dbetahvegds\right\|_2^2\\
=&\alpha\left(\alpha +(1-\alpha)\cdot\frac{\left\|\Iegds\gv\right\|_1}{\left\|\Iegds\gv\right\|_2}\right) +(1-\alpha)\left(\alpha\cdot\frac{\left\|\Iegds\gv\right\|_2}{\left\|\Iegds\gv\right\|_1} +1-\alpha\right)^2.
\end{aligned}
\end{equation*}

Using the inequalities $1\leq \frac{\|\Iegds\gv\|_1}{\|\Iegds\gv\|_2}\leq \sqrt{p_1}$ and the equalities $\alpha+(1-\alpha)^3=1-\alpha(1-\alpha)(2-\alpha)$ and $\alpha^2+1-\alpha=1-\alpha(1-\alpha)$, we obtain the lower bound 
$$1-\alpha(1-\alpha)\left(2-\alpha-\frac\alpha{p_1}-\frac{2(1-\alpha)}{\sqrt{p_1}}\right)\stackrel{(a)}{\geq}1-\alpha(1-\alpha)(2-\alpha)\cdot\left(1-\frac 1{p_1}\right),$$
where $(a)$ follows from $\sqrt{p_1}\leq p_1$ for $p_1\geq 1$, and the upper bound
$$1+\alpha(1-\alpha)\cdot(\sqrt{p_1}-1).$$

Noting that $\alpha(1-\alpha)(2-\alpha)<0.39$ for $\alpha \in [0,1]$, and $\alpha(1-\alpha)\leq \frac14$ completes the proof.
\end{proof}

\begin{proof}[Proof of Proposition \ref{thm:h_alpha_gs}]~\\
For $c_{\alpha,\Delta t}\geq 0$,
\begin{equation*}
\begin{aligned}
\left\|\dbetahvegdgc\right\|_1
=& \|\Iegd\gv\|_1\left(\frac{\alpha\cdot c_{\alpha,\Delta t}\cdot \Delta t}{\left\|\Iegd\gv\right\|_1}+\frac{(1-\alpha)\sqrt{c_{\alpha,\Delta t}\cdot \Delta t}}{\left\|\Iegd\gv\right\|_2}\right)\\
=&\left(\alpha\cdot c_{\alpha,\Delta t}\cdot\Delta t+(1-\alpha)\sqrt{c_{\alpha,\Delta t}\cdot\Delta t}\frac{\left\|\Iegd\gv\right\|_1}{\left\|\Iegd\gv\right\|_2}\right)\\
=&\left(\alpha\cdot c_{\alpha,\Delta t}\cdot \Delta t+(1-\alpha)\sqrt{c_{\alpha,\Delta t}\cdot \Delta t\cdot q_1}\right).\\
\end{aligned}
\end{equation*}
\begin{equation*}
\begin{aligned}
\left\|\dbetahvegdgc\right\|_2^2
=& \|\Iegd\gv\|_2^2\left(\frac{\alpha\cdot c_{\alpha,\Delta t}\cdot \Delta t}{\left\|\Iegd\gv\right\|_1}+\frac{(1-\alpha)\sqrt{c_{\alpha,\Delta t}\cdot \Delta t}}{\left\|\Iegd\gv\right\|_2}\right)^2\\
=&\left(\alpha\cdot c_{\alpha,\Delta t}\cdot\Delta t\frac{\left\|\Iegd\gv\right\|_2}{\left\|\Iegd\gv\right\|_1}+(1-\alpha)\sqrt{c_{\alpha,\Delta t}\cdot\Delta t}\right)^2\\
=&\left(\alpha\cdot c_{\alpha,\Delta t}\cdot \Delta t\cdot\frac1{\sqrt{q_1}}+(1-\alpha)\sqrt{c_{\alpha,\Delta t}\cdot \Delta t}\right)^2.\\
\end{aligned}
\end{equation*}
Solving 
\begin{equation*}
\begin{aligned}
\Delta t=&\alpha \left\|\dbetahvegdgc\right\|_1 +(1-\alpha)\left\|\dbetahvegdgc\right\|_2^2\\
&\alpha\cdot\left(\alpha\cdot c_{\alpha,\Delta t}\cdot \Delta t+(1-\alpha)\sqrt{c_{\alpha,\Delta t}\cdot \Delta t\cdot q_1}\right)\\
&+(1-\alpha)\left(\alpha\cdot c_{\alpha,\Delta t}\cdot \Delta t\cdot\frac1{\sqrt{q_1}}+(1-\alpha)\sqrt{c_{\alpha,\Delta t}\cdot \Delta t}\right)^2\\
\end{aligned}
\end{equation*}
for $c_{\alpha,\Delta t}$, we obtain
$$c_{\alpha,\Delta t}=\left(\frac{\sqrt{2\alpha\sqrt{q_1(\alpha^2q_1 + 4\Delta t(1-\alpha))} + q_1((1-\alpha)^3-2\alpha^2)} - (1-\alpha)\sqrt{q_1(1-\alpha)}}{\alpha\sqrt{4\Delta t(1-\alpha)}}\right)^2.$$
\end{proof}

\begin{proof}[Proof of Proposition \ref{thm:eg_norm_gs}]
\begin{equation*}
\begin{aligned}
\left\|\dbetahvegdg\right\|_1
=& \|\Iegd\gv\|_1\left(\frac{\alpha\cdot \Delta t}{\left\|\Iegd\gv\right\|_1}+\frac{(1-\alpha)\sqrt{\Delta t}}{\left\|\Iegd\gv\right\|_2}\right)\\
=&\left(\alpha\cdot \Delta t+(1-\alpha)\sqrt{\Delta t}\frac{\left\|\Iegds\gv\right\|_1}{\left\|\Iegds\gv\right\|_2}\right).
\end{aligned}
\end{equation*}
\begin{equation*}
\begin{aligned}
\left\|\dbetahvegdg\right\|_2^2
=&\|\Iegd\gv\|_2^2\left(\frac{\alpha\cdot \Delta t}{\left\|\Iegd\gv\right\|_1}+\frac{(1-\alpha)\sqrt{\Delta t}}{\left\|\Iegd\gv\right\|_2}\right)^2\\
=&\left(\alpha\cdot \Delta t\frac{\left\|\Iegd\gv\right\|_2}{\left\|\Iegd\gv\right\|_1}+(1-\alpha)\sqrt{\Delta t}\right)^2.
\end{aligned}
\end{equation*}

\begin{equation*}
\begin{aligned}
&\alpha \left\|\dbetahvegdg\right\|_1 +(1-\alpha)\left\|\dbetahvegdg\right\|_2^2\\
=&\alpha\left(\alpha\cdot \Delta t +(1-\alpha)\sqrt{\Delta t}\cdot\frac{\left\|\Iegd\gv\right\|_1}{\left\|\Iegd\gv\right\|_2}\right) +(1-\alpha)\left(\alpha\cdot \Delta t\cdot\frac{\left\|\Iegd\gv\right\|_2}{\left\|\Iegds\gv\right\|_1} +(1-\alpha)\sqrt{\Delta t}\right)^2.
\end{aligned}
\end{equation*}

Using the inequalities $\frac{\|\Iegd\gv\|_1}{\|\Iegd\gv\|_2}\leq \sqrt{p_1}$, $\frac{\|\Iegd\gv\|_2}{\|\Iegd\gv\|_1}\leq 1$ we obtain upper bound
\begin{equation*}
\begin{aligned}
&\alpha\cdot\left(\alpha\Delta t+(1-\alpha)\sqrt{\Delta t}\cdot\frac{\|\Iegd\gv\|_1}{\|\Iegd\gv\|_2}\right)+(1-\alpha)\cdot\left(\alpha\Delta t\cdot\frac{\|\Iegd\gv\|_2}{\|\Iegd\gv\|_1}+(1-\alpha)\sqrt{\Delta t}\right)^2\\
\leq&\alpha\cdot\left(\alpha\Delta t+(1-\alpha)\sqrt{\Delta t}\cdot\sqrt{p_1}\right)+(1-\alpha)\cdot\left(\alpha\Delta t+(1-\alpha)\sqrt{\Delta t}\right)^2\\
\stackrel{(a)}{=}&\Delta t+\alpha(1-\alpha)\left((\alpha-3)\Delta t+\sqrt{\Delta t}\sqrt{p_1}+\alpha(\Delta t)^2+2(1-\alpha)\Delta t\sqrt{\Delta t}\right)\\
\stackrel{(b)}{\leq}&\Delta t+\alpha(1-\alpha)\left((\alpha-3)\Delta t+\sqrt{\Delta t}\sqrt{p_1}+\alpha\Delta t+2(1-\alpha)\Delta t\right)\\
=&\Delta t\left(1+\alpha(1-\alpha)\left(\sqrt{\frac{p_1}{\Delta t}}-1\right)\right),\\
\end{aligned}
\end{equation*}
where in $(a)$, we use $\alpha^2+(1-\alpha)^3=1+\alpha(1-\alpha)(\alpha-3)$ and $(b)$ uses $(\Delta t)^2\leq \Delta t$, $\sqrt{\Delta t}\leq 1$ for $\Delta t\leq 1$.

The lower bound is obtain by using $\frac{\|\Iegd\gv\|_1}{\|\Iegd\gv\|_2}\geq 1$, $\frac{\|\Iegd\gv\|_2}{\|\Iegd\gv\|_1}\geq \frac 1{\sqrt{p_1}}$:
\begin{equation*}
\begin{aligned}
&\alpha\cdot\left(\alpha\Delta t+(1-\alpha)\sqrt{\Delta t}\cdot\frac{\|\Iegd\gv\|_1}{\|\Iegd\gv\|_2}\right)+(1-\alpha)\cdot\left(\alpha\Delta t\cdot\frac{\|\Iegd\gv\|_2}{\|\Iegd\gv\|_1}+(1-\alpha)\sqrt{\Delta t}\right)^2\\
\geq&\alpha\cdot\left(\alpha\Delta t+(1-\alpha)\sqrt{\Delta t}\right)+(1-\alpha)\cdot\left(\frac{\alpha\Delta t}{\sqrt{p_1}}+(1-\alpha)\sqrt{\Delta t}\right)^2\\
\stackrel{(a)}{=}&\Delta t+\alpha(1-\alpha)\left((\alpha-3)\Delta t+\sqrt{\Delta t}+\alpha\frac{(\Delta t)^2}{p_1}+2(1-\alpha)\frac{\Delta t\sqrt{\Delta t}}{\sqrt{p_1}}\right)\\
\stackrel{(b)}{\geq}&\Delta t+\alpha(1-\alpha)\left((\alpha-3)\Delta t+\Delta t+\alpha\frac{(\Delta t)^2}{p_1}+2(1-\alpha)\frac{(\Delta t)^2}{p_1}\right)\\
=&\Delta t\left(1-\alpha(1-\alpha)(2-\alpha)\left(1-\frac{\Delta t}{p_1}\right)\right),
\end{aligned}
\end{equation*}
where in $(a)$, we use $\alpha^2+(1-\alpha)^3=1+\alpha(1-\alpha)(\alpha-3)$ and $(b)$ uses $\sqrt{\Delta t}\geq \Delta t$ for $\Delta t\leq 1$, and $\sqrt{p_1}\leq p_1$ for $p_1\geq 1$.

Noting that $\alpha(1-\alpha)(2-\alpha)<0.39$ for $\alpha \in [0,1]$, and $\alpha(1-\alpha)\leq \frac14$ completes the proof.
\end{proof}

\begin{proof}[Proof of Lemma \ref{thm:decr_p1}]~\\
We assume without loss of generality that $\gv$ is sorted, so that $|g_1|\geq|g_2|\geq \dots$.

$$\frac{\|\Iegds\gv\|_2^2}{\|\Iegds\gv\|_1}=\frac{g_1^2\cdot\left(1+\sum_{i=2}^{p_1}\left(\frac{g_i}{g_1}\right)^2\right)}{|g_1|\cdot\left(1+\sum_{i=2}^{p_1}\left|\frac{g_i}{g_1}\right|\right)},$$
Since $(g_i/g_1)^2\leq |g_i/g_1|$ and $|g_{i}|\geq |g_{i+1}|$, $\left(1+\sum_{i=2}^{p_1}\left(\frac{g_i}{g_1}\right)^2\right)/\left(1+\sum_{i=2}^{p_1}\left|\frac{g_i}{g_1}\right|\right)$ is a decreasing function of $p_1$.
\end{proof}

\begin{lemma}
\label{thm:abs_der}
$$\left(\left|\frac f g\right|\right)^{(k)}=\sgn\left(\frac f g\right)\cdot\left(\frac{f^{(k)}\cdot g-f\cdot g^{(k)}}{g^2}\right)+\Oo\left(\left(\frac f g\right)^{(k-1)}\right),$$
where $\Oo\left(\left(\frac f g\right)^{(k-1)}\right)$ denotes derivatives of $\left(\frac f g\right)$ of orders strictly lower than $k$.
\end{lemma}
\begin{proof}~\\
We first show that
\begin{equation}
\label{eq:quotient_higher}
\left(\frac fg\right)^{(k)}=\frac{f^{(k)}g -f g^{(k)}}{g^2}-\sum_{i=1}^{k-1} \binom{k}{i} \left(\frac fg\right)^{(k-i)}\frac{g^{(i)}}{g}:
\end{equation}
\begin{equation*}
\begin{aligned}
&\frac{f^{(k)}g -f g^{(k)}}{g^2}-\sum_{i=1}^{k-1} \binom{k}{i} \left(\frac fg\right)^{(k-i)}\frac{g^{(i)}}{g}=\frac{1}{g} \left(f^{(k)} -\left(\frac fg\right) g^{(k)}-\sum_{i=1}^{k-1} \binom{k}{i} \left(\frac fg\right)^{(k-i)}g^{(i)}\right)\\
=&\frac{1}{g} \left(f^{(k)} -\left(\sum_{i=0}^{k} \binom{k}{i} \left(\frac fg\right)^{(k-i)}g^{(i)}-\left(\frac fg\right)^{(k)}g\right)\right)\\
\stackrel{(a)}{=}&\frac{1}{g} \left(f^{(k)} -\left(\left(\frac fg g\right)^{(k)}-\left(\frac fg\right)^{(k)}g\right)\right)\\
=&\frac{1}{g} \left(f^{(k)} -f^{(k)}+\left(\frac fg\right)^{(k)}g\right)=\left(\frac fg\right)^{(k)},\\
\end{aligned}
\end{equation*}
where $(a)$ follows from the general Leibniz rule, $(fg)^{(k)}=\sum_{i=0}^{k} \binom{k}{i} f^{(k-i)}g^{(i)}$.
Now, according to the chain rule for higher order derivatives, Fa\`{a} di Bruno's formula,
$$\frac {df(g(x))} {dx^k} =f'(g(x))\cdot g^{(k)}(x)+\Oo(g^{(k-1)}),$$
where $\Oo(g^{(k-1)})$ denotes terms with derivatives of $g$ of orders strictly lower than $k$. Applying this, we obtain
$$\left(\left|\frac f g\right|\right)^{(k)}=\sgn\left(\frac f g\right)\cdot\left(\frac f g\right)^{(k)}+\Oo\left(\left(\frac f g\right)^{(k-1)}\right).$$
Applying Equation \ref{eq:quotient_higher} completes the proof.
\end{proof}

\begin{lemma}~\\
\label{thm:gk0}
With $\gv(t)=-\Ss(\betahatolsv-\betav(t))$ for $\betav(t)$ given by Equation \ref{eq:beta_ef},
$$\gv^{(k)}(t_i) =\frac1{1-\gamma}\left(-\Ss(\Ii)^{(k-1)}(t_i)(\alpha\cdot\sgn(\gv(t_i))+(1-\alpha)\cdot\gv(t_i)) + \Oo((\Ii)^{(k-2)}(t_i))\right),$$
where $\Oo((\Ii)^{(k-2)}(t_i))$ depends only on derivatives of order $k-2$ and lower, and\\$\Oo((\Ii)^{(k)}(t_i))=\bm{0}$ for $k<0$.
\end{lemma}
\begin{proof}~\\
We begin by showing that
\begin{equation}
\label{eq:omega_of_i}
\Om^{(k)}(t_i,t_i)=-\frac{1-\alpha}{1-\gamma}\Ss\cdot(\Ii)^{(k-1)}(t_i)+\frac1{1-\gamma}\Oo((\Ii)^{(k-2)}(t_i)),
\end{equation}
where $\Oo((\Ii)^{(k-2)}(t_i))$ depends only on derivatives of order $k-2$ and lower, and\\$\Oo((\Ii)^{(k)}(t_i))=\bm{0}$ for $k<0$.

According to \cite{magnus1954exponential}
\begin{equation*}
\begin{aligned}
\Omi(t_i,t)=&\int_{t_i}^t\bm{A}(\tau_1)d\tau_1+\frac12\int_{t_i}^t\int_{t_i}^{\tau_1}\left[\bm{A}(\tau_1),\bm{A}(\tau_2)\right]d\tau_2d\tau_1\\
&+\frac14\int_{t_i}^t\int_{t_i}^{\tau_1}\int_{t_i}^{\tau_2}\left[\bm{A}(\tau_1),\left[\bm{A}(\tau_2),\bm{A}(\tau_3)\right]\right]d\tau_3d\tau_2d\tau_1+\dots,
\end{aligned}
\end{equation*}
where the commutator is defined according to $[\bm{A},\bm{B}]:=\bm{A}\bm{B}-\bm{B}\bm{A}$. 

Multiple applications of the fundamental theorem of calculus result in
$$(\Omi)^{(k)}(t_i,t_i)=\bm{A}^{(k-1)}(t_i)+\frac12\Oo(\bm{A}^{(k-2)}(t_i))+\frac14\Oo(\bm{A}^{(k-3)}(t_i))+\dots$$
where $\Oo(\bm{A}^{(k)}(t_i))$ depends only on derivatives of order $k$ and lower, and $\Oo(\bm{A}^{(k)}(t_i))=\bm{0}$ for $k<0$. Setting $\bm{A}(t)=-\frac{1-\alpha}{1-\gamma}\Ss\Ii(t)$ results in Equation \ref{eq:omega_of_i}.

Next, we note that derivatives of order $k$ only appear in terms where $i=0$ or $i=k$ in Equation \ref{eq:mat_exp_der} and obtain
\begin{equation}
\label{eq:exp_x_oo}
\frac{d^k \exp(\Xm(t))}{dt^k}=\sum_{n=1}^\infty \frac 1{n!}\left(\Xm(t)\frac{d^k\Xm(t)^{n-1}}{dt^k}+\frac{d^k\Xm(t)}{dt^k}\Xm(t)^{n-1}+\Oo(\Xm^{(k-1)}(t))\right),
\end{equation}
where $\Oo(\Xm^{(k-1)}(t))$ depends only on derivatives of order $k-1$ and lower, and $\Oo(\Xm^{(0)}(t))=\bm{0}$.\\
Since $\Om(t_i,t_i)=\bm{0}$, for $n\in \mathbb{N}_0$,
$$\Om^n(t_i,t_i)=
\begin{cases}
\Imm,\ n=0\\
\bm{0},\ n>0,
\end{cases}$$
and inserting $\Om(t_i,t_i)$ into Equation \ref{eq:exp_x_oo} we obtain
\begin{equation}
\label{eq:dexpom_dt}
\begin{aligned}
&\frac{d^k \exp(\Om(t_i,t_i))}{dt^k}\\
&=\sum_{n=1}^\infty \frac 1{n!}\left(\underbrace{\Om(t_i,t_i)}_{=\bm{0}}\cdot(\Om^{n-1})^{(k)}(t_i,t_i)+\Om^{(k)}(t_i,t_i)\cdot \underbrace{\Om(t_i,t_i)^{n-1}}_{\Imm,\ n=1;\ \bm{0},\ n>1} +\Oo(\Om^{(k-1)}(t_i,t_i))\right)\\
&=\Om^{(k)}(t_i,t_i)+\Oo(\Om^{(k-1)}(t_i,t_i))=-\frac{1-\alpha}{1-\gamma}\Ss\cdot(\Ii)^{(k-1)}(t_i)+\frac1{1-\gamma}\Oo((\Ii)^{(k-2)}(t_i)).
\end{aligned}
\end{equation}
Now
\begin{equation*}
\begin{aligned}
\gv^{(k)}(t_i)=& (-\Ss(\betahatolsv-\betahvegf(t)))^{(k)}=\Ss\betahvegf^{(k)}(t_i)\\
\stackrel{(a)}{=}&-\frac1{1-\alpha}\exp(\Omi(t_i,t_i))^{(k)}\left(\alpha\cdot\svi+(1-\alpha)\cdot\Ss\left(\betahatolsv-\betahvegf(t_i)\right)\right)\\
\stackrel{(b)}{=}&\frac1{1-\gamma}\left(\Ss\cdot(\Ii)^{(k-1)}(t_i)\left(\alpha\cdot\underbrace{\svi}_{=-\sgn(\gv(t_i))}+(1-\alpha)\cdot\underbrace{\Ss\left(\betahatolsv-\betahvegf(t_i)\right)}_{=-\gv(t_i)}\right)\right.\\
&\left.+\Oo((\Ii)^{(k-2)}(t_i))\right)\\
=&\frac1{1-\gamma}\left(-\Ss(\Ii)^{(k-1)}(t_i)\left(\alpha\cdot\sgn(\gv(t_i))+(1-\alpha)\cdot\gv(t_i)\right)+\Oo((\Ii)^{(k-2)}(t_i))\right),
\end{aligned}
\end{equation*}
where $(a)$ follows from Equation \ref{eq:beta_ef} and $(b)$ follows from Equation \ref{eq:dexpom_dt}.
\end{proof}

\begin{lemma}
\label{thm:mat_exp_der}
\begin{equation}
\label{eq:mat_exp_der}
\frac{d^k \exp(\Xm(t))}{dt^k}=\sum_{n=1}^\infty \frac 1 {n!} \sum_{i=0}^k \binom{k}{i}\cdot\frac{d^i \Xm(t)}{dt^i}\cdot\frac{d^{k-i}\Xm(t)^{n-1}}{dt^{k-i}}.
\end{equation}
\end{lemma}

\begin{proof}~\\
For $n\geq 1$, according to the general Leibniz rule, $(fg)^{(k)}=\sum_{i=0}^{k} \binom{k}{i} f^{(i)}g^{(k-i)}$,
$$\frac{d^k \Xm(t)^n}{dt^k}=\frac{d^k\left(\Xm(t) \Xm(t)^{n-1}\right)}{dt^k}=\sum_{i=0}^k \binom{k}{i}\cdot\frac{d^i \Xm(t)}{dt^i}\cdot\frac{d^{k-i}\Xm(t)^{n-1}}{dt^{k-i}}.$$
Inserting this into the Taylor expansion of the matrix exponential, we obtain
$$\frac{d^k \exp(\Xm(t))}{dt^k}=\sum_{n=1}^\infty \frac 1 {n!} \frac {d^k \Xm(t)^n}{dt^k}=\sum_{n=1}^\infty \frac 1 {n!} \sum_{i=0}^k \binom{k}{i}\cdot\frac{d^i \Xm(t)}{dt^i}\cdot\frac{d^{k-i}\Xm(t)^{n-1}}{dt^{k-i}}.$$
\end{proof}

\end{document}